\documentclass[conference]{IEEEtran}
\newif\ifdraft\draftfalse

\newif\ifproceedings\proceedingsfalse

\usepackage[noadjust]{cite}

\ifCLASSINFOpdf
\else
\fi
\ifdraft
\usepackage[draft]{style/commenting}
\else
\usepackage[nompar]{commenting}
\fi

\declareauthor{lo}{LO}{red}
\authorcommand{lo}{comment}

\declareauthor{dw}{DW}{blue}
\authorcommand{dw}{comment}
\setdefaultauthor{dw}

\makeatletter
\renewcommand{\comm@todo@mpar}[1]{}
\makeatother

\def\divider{%
  \leavevmode\leaders\hrule height 0.6ex depth \dimexpr0.4pt-0.6ex\hfill%
  \kern0pt%
}

\renewcommand{\OpenCommentBraket}{$\lceil$}
\renewcommand{\ClosedCommentBraket}{$\rfloor$}
\usepackage[normalem]{ulem} % \sout{...} - strike out; added by Luke
\newcommand\defeq{\coloneqq}
\renewcommand\phi{\varphi}

\usepackage[british,UKenglish]{babel}
\usepackage[T1]{fontenc}
\usepackage{stmaryrd}
\usepackage{mathtools}
\usepackage{amsmath}
\usepackage{amssymb}
\usepackage{amsfonts}
\usepackage{amsthm}
\usepackage{thmtools,thm-restate}
\usepackage{tabularx}
\usepackage{proof}
\usepackage{xcolor}
\definecolor{oxblue}{RGB}{0,33,71}
\definecolor{oxlightblue}{RGB}{0,72,205}
\definecolor{oxlightblue2}{RGB}{161,196,208}
\definecolor{oxgold}{HTML}{a0630a%ac6a0a%b9720b%cc7e0c%e88f0e% 9e8340
}
\definecolor{oxyellow}{RGB}{243,222,116}
\definecolor{oxbrown}{HTML}{88562e}%153,101,21}
\definecolor{oxgreen}{HTML}{3e6111}%558618}
\definecolor{oxgreen2}{RGB}{105,145,59}
\definecolor{oxgreen3}{RGB}{185,207,150}

\definecolor{oxred}{HTML}{ac0b1d}%960a2c}
\definecolor{oxred2}{RGB}{190,15,52}
\definecolor{oxpink}{RGB}{235,196,203}
\definecolor{oxorange}{RGB}{207,122,48}

\definecolor{oxgrey}{RGB}{167,157,150}

\usepackage{listings}
\usepackage{stackengine}

\usepackage{tikz}
\usetikzlibrary{positioning}
\usetikzlibrary{arrows.meta, decorations.pathreplacing, shadows}

\newcommand\biglor{\bigvee}
\newcommand\bigland{\bigwedge}

\newcommand\sinti[2]{{#1}\llbracket#2\rrbracket}

\newcommand{\from}{:}
\newcommand{\As}{\mathcal A}
\newcommand{\Bs}{\mathcal B}

\newcommand{\Hf}{\mathcal H}
\newcommand{\Sf}{\mathcal S}
\newcommand{\Mf}{\mathcal M}
\newcommand{\Cf}{\mathcal C}
\newcommand{\Ad}{\mathfrak A}
\newcommand{\Bd}{\mathfrak B}
\newcommand{\Rd}{\mathfrak R}
\newcommand{\Sd}{\mathfrak S}
\newcommand{\Td}{\mathfrak T}
\newcommand{\igt}{\mathrm{gt}_\iota}

\newcommand\nat{\mathbb N}
\newcommand\bool{\mathbb B}
\newcommand\nin{\not\in}

\newcommand\crel{\sqsubseteq_c}
\newcommand\crelrev{\sqsupseteq_c}
\newcommand\prel{\sqsubseteq}
\newcommand\arel{\precsim}

\DeclareMathOperator{\argmax}{arg\, max}
\DeclareMathOperator{\argmin}{arg\, min}

\DeclareMathOperator{\dir}{dir}

\DeclareMathOperator{\free}{fv}
\DeclareMathOperator{\vars}{vars}

\DeclareMathOperator{\dom}{dom}

\DeclareMathOperator{\andf}{and}
\DeclareMathOperator{\orf}{or}

\DeclareMathOperator{\hexistsh}{exists}

\DeclareMathOperator{\On}{\mathbf{On}}
\DeclareMathOperator{\Lim}{\mathbf{Lim}}
\DeclareMathOperator{\LIA}{LIA}

\DeclareMathOperator{\FO}{{FO}}

\DeclareMathOperator{\app}{@}

\DeclareMathOperator{\posex}{posex}
\DeclareMathOperator{\nega}{neg}
\DeclareMathOperator{\Comp}{Comp}

\DeclareMathOperator{\Iter}{Iter}

\DeclareMathOperator{\Add}{Add}

\newcommand{\hexists}{\hexistsh}

\newcommand{\xtwoheadrightarrow}[2][]{
  \xrightarrow[#1]{#2}\mathrel{\mkern-14mu}\rightarrow
}

\stackMath
\newcommand\xxrightarrow[2][]{\mathrel{%
  \setbox2=\hbox{\stackon{\scriptstyle#1}{\scriptstyle#2}}%
  \stackunder[0pt]{%
    \xrightarrow{\makebox[\dimexpr\wd2\relax]{$\scriptstyle#2$}}%
  }{%
   \scriptstyle#1\,%
  }%
}}
\newcommand\xxtwoheadrightarrow[2][]{\mathrel{%
  \setbox2=\hbox{\stackon{\scriptstyle#1}{\scriptstyle#2}}%
  \stackunder[0pt]{%
    \xtwoheadrightarrow{\makebox[\dimexpr\wd2\relax]{$\scriptstyle#2$}}%
  }{%
   \scriptstyle#1\,%
  }%
}}
\newcommand{\lred}[1]{\xxrightarrow[\ell]{#1}}
\newcommand{\lredrt}{\xxtwoheadrightarrow[\ell]{}}
\newcommand{\sred}{\xxrightarrow[s]{}}
\newcommand{\bured}{\rightarrow_{\beta\upsilon}}
\newcommand{\buredrt}{\twoheadrightarrow_{\beta\upsilon}}
\newcommand{\bred}{\rightarrow_\beta}

\newcommand{\pred}{\rightarrow_\parallel}

\newcommand{\force}{\vartriangleright}

\newcommand{\Res}{\Rightarrow_{\Resh}}

\newcommand\sref[1]{Proof Step \ref{#1}}
\newcommand\srefs[2]{Proof Steps \ref{#1} and \ref{#2}}

\newcommand\btypes{\mathfrak I}

\newcommand{\stkout}[1]{\ifmmode\text{\sout{\ensuremath{#1}}}\else\sout{#1}\fi}
\usepackage{paralist}

\usepackage[inline]{enumitem}
\setlist[enumerate,1]{label={(\roman*)}}
\usepackage{bussproofs}

\usepackage[caption=false]{subfig}

\newlist{thmlist}{enumerate}{1}
\setlist[thmlist]{label=\textup{(\roman{thmlisti})},ref={(\roman{thmlisti})},noitemsep}

\newlist{complist}{enumerate}{1}
\setlist[complist,1]{noitemsep,label=\textbf{(C\arabic*)},leftmargin=2.5\parindent}

\usepackage{diagbox} %for diagonal division of a cell in tabular
\usepackage{makecell} %for diagonal division of a cell in tabular

\usepackage{pifont}% http://ctan.org/pkg/pifont
\tikzstyle{resnode}=[anchor=base,fill=black!15,inner sep=2pt]
\tikzstyle{unode}=[anchor=base]

\makeatletter

\renewcommand{\p@thmlisti}{\perh@ps{\thetheorem}}
\protected\def\perh@ps#1#2{\textup{#1#2}}
\newcommand{\itemrefperh@ps}[2]{\textup{#2}}
\newcommand{\itemref}[1]{\begingroup\let\perh@ps\itemrefperh@ps Part~\ref{#1}\endgroup}
\newcommand{\itemrefs}[2]{\begingroup\let\perh@ps\itemrefperh@ps Parts~\ref{#1} and~\ref{#2}\endgroup}

\renewcommand{\Res}{\vdash_\As}
\newcommand{\Resp}{\vdash_\Ad}
\newcommand{\Respstr}[1]{\vdash_{#1}}
\newcommand{\embed}[1]{\left[ #1\right]}

\newcommand{\vals}{\mathfrak a}
\newcommand{\Set}{\Gamma}
\newcommand{\Prgm}{\Pi}
\renewcommand{\Hf}{\mathcal F}

\declaretheorem[
    name=Theorem,
    Refname={Theorem,Theorems}]{theorem}
\declaretheorem[
    name=Lemma,
    Refname={Lemma,Lemmas},
    sibling=theorem]{lemma}
\declaretheorem[
    name=Proposition,
    Refname={Prop.,Propositions},
    sibling=theorem]{proposition}
\declaretheorem[
    name=Corollary,
    Refname={Corollary,Corollaries},
    sibling=theorem]{corollary}
\declaretheorem[
    name=Claim,
    Refname={Claim,Claims}
    ]{claim}

\declaretheorem[numbered=no,
    name=Assumption,
    Refname={Assumption,Assumptions}]{assumption}
\theoremstyle{definition}
\declaretheorem[
    name=Definition,
    Refname={Definition,Definitions},
    sibling=theorem]{definition}
\declaretheorem[
    name=Example,
    Refname={Example,Examples},
    sibling=theorem]{example}
\theoremstyle{remark}
\declaretheorem[
    name=Remark,
    Refname={Remark,Remarks},
    sibling=theorem]{remark}

    \declaretheoremstyle[
    spaceabove=-4pt, 
    spacebelow=6pt, 
    headfont=\it, 
    bodyfont = \normalfont,
    postheadspace=1em, 
    qed=$\blacksquare$, 
    headpunct={.}]{myproofstyle} %<---- change this name
\declaretheorem[name={Proof}, style=myproofstyle, unnumbered]{claimproof}

\let\oldproof\proof
\let\oldendproof\endproof
\def\proof{\setcounter{claim}{0}\begingroup\oldproof}
\def\endproof{\oldendproof \endgroup}

\usepackage{hyperref}
\usepackage[capitalize]{cleveref}

\Crefname{theorem}{Thm.}{Theorems}
\Crefname{corollary}{Cor.}{Corollary}
\crefname{proposition}{Prop.}{Propositions}
\Crefname{claim}{Claim}{Claims}
\Crefname{definition}{Def.}{Definitions}
\Crefname{fact}{Fact}{Facts}
\Crefname{conj}{Conjecture}{Conjectures}
\Crefname{example}{Ex.}{Examples}
\Crefname{example}{Rem.}{Remarks}
\Crefname{convention}{Convention}{Conventions}
\Crefname{lemma}{Lemma}{Lemmas}
\Crefname{section}{Sec.}{Sec.}
\Crefname{appendix}{App.}{App.}

\addtotheorempostheadhook[theorem]{\crefalias{thmlisti}{theorem}}
\addtotheorempostheadhook[lemma]{\crefalias{thmlisti}{lemma}}
\addtotheorempostheadhook[proposition]{\crefalias{thmlisti}{proposition}}
\addtotheorempostheadhook[corollary]{\crefalias{thmlisti}{corollary}}
\addtotheorempostheadhook[claim]{\crefalias{thmlisti}{claim}}
\addtotheorempostheadhook[fact]{\crefalias{thmlisti}{fact}}
\addtotheorempostheadhook[conjecture]{\crefalias{thmlisti}{conjecture}}
\addtotheorempostheadhook[example]{\crefalias{thmlisti}{example}}
\addtotheorempostheadhook[definition]{\crefalias{thmlisti}{definition}}
\addtotheorempostheadhook[remark]{\crefalias{thmlisti}{remark}}

\newcommand{\inferbinlics}[5]{\begin{minipage}{22ex}{\bfseries #1}\end{minipage} {
    \begin{minipage}{10ex}{$\infer{#4}{#2 &&& #3}$}\end{minipage}
    \par\noindent #5}}
\newcommand{\inferthrlics}[6]{\begin{minipage}{13ex}{\bfseries #1}\end{minipage} {
    \begin{minipage}{10ex}{$\infer{#5}{#2 &&& #3 &&& #4}$}\end{minipage}
    \par\noindent #6}}
\newcommand{\inferunlics}[4]{\begin{minipage}{22ex}{\bfseries #1}\end{minipage} {
    \begin{minipage}{10ex}{$\infer{#3}{#2}$}
    \end{minipage}
    \par\noindent
      #4
    }}

\makeindex

\begin{document}

%\title{HoCHC: a Refutationally Complete and Semantically Invariant Fragment of Higher-Order Logic Modulo Theories}

%\title{HoCHC: A Refutationally-complete and Semantically-invariant Higher-order Horn System Modulo Theories}

\title{HoCHC: A Refutationally Complete and Semantically Invariant System of Higher-order Logic Modulo Theories}

\iffalse
\lo{LICS19 PC members: Coquand, Pitts, Birkedal; Panangaden, Chaudhuri, Kov\'acs}
\fi

\author{\IEEEauthorblockN{C.-H. Luke Ong}
\IEEEauthorblockA{University of Oxford}
\and
\IEEEauthorblockN{Dominik Wagner}
\IEEEauthorblockA{University of Oxford}}



\ifproceedings
 \IEEEoverridecommandlockouts
 \IEEEpubid{\makebox[\columnwidth]{978-1-7281-3608-0/19/\$31.00~
 \copyright2019 IEEE \hfill}
\hspace{\columnsep}\makebox[\columnwidth]{ }}
\fi

\maketitle


\iffalse
\lo{Possible titles: 
\begin{enumerate}
\item HoCHC: a semi-decidable and semantically-invariant (or -robust) system of higher-order logic modulo theories
\item HoCHC: a refutationally complete and semantically invariant fragment of higher-order logic modulo theories
\item HoCHC: a refutationally-complete and semantically-invariant system of higher-order logic modulo theories
\end{enumerate}}

\lo{To keep within the 12-page limit, I propose we move \cref{sec:compact-theories} to the Appendix. Reason: the main contribution is the refutational completeness proof (which is already quite long), \cref{sec:compact-theories} provides a generalisation, which is technical and less interesting than the sections that follow.}
\fi

\begin{abstract}
  %In this paper we continue the investigation into the algorithmic and model theoretic properties of higher-order constrained Horn clauses (HoCHC), a fragment of higher-order logic with background theories proposed by Cathcart Burn et al.~for the verification of higher-order functional programs. 

  %Notwithstanding the fact that full higher-order logic with respect to standard semantics is not recursively enumerable, 
  We present a simple resolution proof system for \emph{higher-order constrained Horn clauses} (HoCHC)---a system of higher-order logic modulo theories---and prove its soundness and refutational completeness w.r.t.~both standard and Henkin semantics.
  As corollaries, we obtain the compactness theorem and semi-decidability of HoCHC for semi-decidable background theories, and we prove that HoCHC satisfies a canonical model property.
  Moreover a variant of the well-known translation from higher-order to 1st-order logic is shown to be sound and complete for HoCHC in both semantics.
  We illustrate how to transfer decidability results for (fragments of) 1st-order logic modulo theories to our higher-order setting, using as example the Bernays-Sch\"onfinkel-Ramsey fragment of HoCHC modulo a restricted form of Linear Integer Arithmetic.
  
  % As a corollary, we conclude that higher-order constrained Horn clauses inherit many properties
  % \begin{enumerate*}
    
  % \end{enumerate*}

  % alternative proof of the equivalence of standard, continuous and monotone semantics for higher-order constrained Horn clauses.

  % For the completeness proof we establish novel model theoretic properties which are refinements of known negative results:
  % \begin{enumerate*}
  % \item The structure obtained by iterating the immediate consequence operator is a \emph{least} model for a carefully chosen relation and
  % \item the immediate consequence operator is ``quasi-continuous''.
  % \end{enumerate*}

  %  This gives rise to another semi-decision procedure provided that the background theory is decidable.
 \end{abstract}


% \fi
\IEEEpeerreviewmaketitle

%\dw{At the moment we sometimes use latin and arabic numbers for enumerations/thmlists. Do you/LICS have preferences for either of them to make it more consistent?}
%\lo{Let's be consistent, and use Arabic.}
%\dw{okay, I changed it to the form ``1)''}

% \lo{Do you know how to change the {\tt cref} default names so that Theorem and Proposition (say) become Thm.~and Prop.?}
% \dw{in {\tt headerlics.tex} {\texttt{\textbackslash Crefname\{thm\}\{Theorem\}\{Theorems\}}} can be changed to arbitrary strings for the singular and plural names.}
% \lo{I have tried but it doesn't work.}
% \dw{That's strange.}
% \lo{It works now.}
% \dw{Do you also want the abbreviated version in the theorem statement?}
% \lo{Not really.}

\section{Introduction}


%\lo{Write the introduction as a kind of long abstract so that the busy referee can form a view of the paper just by reading it. Explain the key ideas and main contributions. }

%[Briefly describe algorithmic verification of higher-order programs via higher-order logic modulo theories.]

Cathcart Burn et al.~\cite{BOR18} recently advocated an automatic,
programming-language independent approach
%\footnote{See the tool implementation \emph{Horus} at {\url{http://mjolnir.cs.ox.ac.uk/horus/}}} 
% \dw{Can we omit that footnote to save space? Horus is now mentioned as related work.} \lo{Yes. Done.}
to verify safety properties of higher-order programs by framing them as solvability problems for systems of higher-order constraints.
These systems consist of Horn clauses of higher-order logic, containing constraints expressed in some suitable background theory.
Consider the functional program:
\[
  \begin{array}{l}
  \mathsf{let}\; \mathit{add} \, x \, y = x + y \\
  \mathsf{letrec}\; \mathit{iter} \, f \, s \, n =
    \mathsf{if}\; {n \leq 0}\; \mathsf{then}\ {s}\ \mathsf{else}\; f\, n \, (\mathit{iter}\, f\, s\, (n - 1)) \\
    \qquad \mathsf{in}\;\lambda n. \mathsf{assert} \;
    \big( n \geq 1 \, \to \, (\mathit{iter}\, \mathit{add} \, n \, n > n + n) \big) 
    %(n \leq \mathit{iter}\, \mathit{add}\; 0\; n)
  \end{array}
\]
Thus $(\mathit{iter} \, 
\mathit{add} \, n \, n)$ computes the value $n+\sum_{i=1}^n i$.
% \[{f \, n \, \big(f \, (n-1) \, (f \, (n-2) \, (\cdots (f \, 1 \, s) \cdots ))\big)}.\]

To verify that the program is \emph{safe} (i.e.~the assertion is never violated), 
it suffices to find 
%\changed[dw]{
overapproximations of the input-output-graph (i.e.\ \emph{invariants})
of the functions that imply
%} 
the required property.
The idea then is to express the problem of finding such a program invariant, \emph{logically}, as a satisfiability problem for the following higher-order constrained system:
\begin{example}[Invariant as system of higher-order constraints]
\label{ex:iter-working}
\[
\begin{array}{l}
\forall x, y, z \, . \, \big(z = x + y \to \Add\,  x \, y \, z\big) \\
\forall f, s, n, x \, .\, \big( n \leq 0 \wedge s = x \to \Iter\, f \, s \, n \, x \big) \\
\forall f, s, n, x \, .\, \big(n > 0 \wedge \exists y \, . \, (\Iter \, f \, s \, (n-1) \, y \wedge f \, n \, y \, x)\\ 
\qquad \qquad \qquad \qquad \to \Iter \, f \, s \, n \, x\big) \\
\forall n, x  \, . \, \big ( n \geq 1 \wedge \Iter \, \Add \, n \, n \, x \; \to \; x > n + n \big)
\end{array}
%E n,x:int. n > 1 ^ iter add n n x ^ x <= n + n
\]
The above are Horn clauses of higher-order logic, obtained by transformation from the preceding program; $\Add \from \iota \to \iota \to \iota \to o$ and $\Iter \from (\iota \to \iota \to \iota \to o) \to \iota \to \iota \to \iota \to o$ are higher-order relations,
and the binary predicates $({\leq}, {>}, \cdots)$ are formulas of the
background theory, Linear Integer Arithmetic (LIA).

%\changed[dw]{
Since the the
assertion in the program is violated for $n=1$, the clauses are unsatisfiable.
%}
\end{example}
%\lo{We could frame the above as a \LaTeX\ example, so that \cref{ex:HoCHC} can refer to it.}

\emph{Is higher-order logic modulo theories a sensible algorithmic approach to verification?
Is it well-founded?}

To set the scene, recall that 1st-order logic is semi-decidable:
1st-order validities\footnote{Define $\mathbf{V}^n(P)$ to be the set of valid sentences of $n$th-order logic with 2-place predicate $P$. 
Then $\mathbf{V}^1({=})$ is recursively enumerable} are recursively enumerable;\addtocounter{footnote}{-1}\addtocounter{Hfootnote}{-1} moreover if a formula is unsatisfiable then it is provable by resolution \cite{DBLP:journals/jacm/DavisP60,R65}.
By contrast, higher-order logic in standard semantics is wildly undecidable.
E.g.~the set\footnotemark ~$\mathbf{V}^2(=)$ of valid sentences of the 2nd-order 
%\dw{should we consistently use either 1st-/2nd- or first-/second-order?} \lo{Yes, the former.} \dw{I've changed it everywhere} 
language of equality is not even analytical \cite{Enderton01}.

This does not necessarily spell doom for the higher-order logic approach. 
One could consider higher-order logic in \emph{Henkin semantics}
\cite{DBLP:journals/jsyml/Henkin50}, which is, after all, ``nothing
but many-sorted 1st-order logic with comprehension
axioms''~\cite{Enderton01} (see also \cite{vanBehthemD83,Leivant94}).
% \dw{change to [6,7]}
% \lo{The package {\tt cite} is supposed to address this, but it doesn't in this case.}
%for a comprehensive survey of higher-order logic)}.
%\lo{(One could also restrict to appropriate fragments of HOL, which is our approach.)}
%Cathcart Burn et al.~\cite{BOR18}, however, posit that standard semantics is the appropriate semantics for higher-order constrained Horn clauses;
%they argue that standard semantics is simple and natural, and hence the semantics of choice for higher-order logic in verification from the perspective of specification; 
However, because the standard semantics is natural and comparatively simple, 
it seems to be the semantics of choice in program verification (e.g.~monadic 2nd-order logic in model checking, and HOL theorem prover \cite{GordonM93,GordonP94} in automated deduction) and in program specification.
%\dw{completeness and thus semantics is generally not regarded as ``important'' in interactive theorem provers (see discussion in \url{https://math.stackexchange.com/questions/176592/how-are-the-full-semantics-of-sol-and-hol-specified})}
% \lo{In HoL, semantics is relevant because it \emph{specifies} the proof system.
%   ``HoL in standard semantics'' is a reference to the standard HoL proof system; ``HoL in Henkin semantics'' is a reference to many-sorted 1st-order logic augmented with comprehension axioms.}
% \dw{As far as I understand, they take standard semantics as an \emph{idealised} semantics, with respect to which the proof system is clearly sound. They could have taken any other natural semantics like monotone, continuous or Henkin. Importantly, they obviously cannot give any guarantee about completeness w.r.t. standard semantics, which they could have if they had taken Henkin semantics. In so far it seems a bit odd to choose standard semantics.}
% \changed[lo]{Yet, desirable properties such as completeness \cite{DBLP:journals/jsyml/Henkin50}, compactness, and sound-and-complete translation to first-order logic hold only for higher-order logic in Henkin semantics, and not standard semantics \cite{DBLP:journals/bsl/Vaananen01}.}

In this paper, we study the algorithmic, model-theoretic and semantical properties of higher-order Horn clauses with a 1st-order background theory.

%\subsubsection*{Contributions and Outline}
\paragraph{A Complete Resolution Proof System for HoCHC}
%\dw{is it necessary to list \cref{sec:prelims} here (it seems more convincing to me to jump to the main contribution here)} \lo{No, I won't mention \cref{sec:prelims}.}
% \dw{Is it okay to combine contributions and outline (to save space and avoid repetitions)?}
% \lo{It is certainly possible.}
The main technical contribution of this paper is the design of a simple resolution proof system for \emph{higher-order constrained Horn clauses} (HoCHC) where the background theory has a unique model \cite{BOR18}, 
and its refutational completeness proof with respect to the standard semantics (\cref{sec:resolution}). 
The proof system and its refutational completeness proof are generalised in \cref{sec:compact-theories} to arbitrary \emph{compact} background theories, which may have more than one model.
%arbitrary Henkin frames (in particular standard frames) in \cref{sec:resolution}. 
%\lo{Briefly explain Henkin frame. It is probably not a standard notion. Cf.~Henkin structure.} 
%\dw{is it sufficient to mention this (most important) special case here?)}
%\changed[dw]{\sout{To our knowledge, this is the first refutationally-complete resolution proof system for a significant fragment of higher-order logic modulo a background theory in standard semantics.}}
% \dw{Shouldn't we add something like ``non-monadic'' here because clearly there is e.g.\ S2S, which is even decidable?}
% \lo{I take your point. Because it is not clear what ``non-monadic fragment of HoL'' means, it may be better to omit it.}

%\changed[dw]{
The completeness proof hinges on a novel model-theoretic insight: we
prove that the immediate consequence operator is
\emph{quasi-continuous}, although
it is not continuous in the standard Scott sense.
Thus, the immediate consequence operator gives rise to a syntactic explanation for unsatisfiability. Moreover, we adapt the proof of the standardisation theorem of the $\lambda$-calculus in \cite{K00} to argue that this explanation can be captured by the rules of the resolution proof system.
%}

%Besides, it is well-known that 
% We achieve the simplicity of the proof system by resorting to a normal form which allows us to keep reasoning about logical connectives at a minimum level and to focus on the essentials. In the course of this work, we prove that this is indeed a normal form, i.e.\ also seemingly more general formulas can be transformed in a meaning-preserving way to this normal form.
%\lo{TODO. 1) Say something about the proof strategy. 
%2) Possibly(?) explain that this works for arbitrary Henkin frame.
%}
%\dw{Would you add even more details?}
%\lo{I would if we could, but we have run out of space.}

\paragraph{Canonical Model Property} As shown in \cite{BOR18}, a
disadvantage of the standard semantics is failure of the least model
property (w.r.t.~the pointwise ordering).
%\changed[dw]{
However, we prove in \cref{sec:quasi-mon} that the immediate consequence operator is ``sufficiently'' monotone and hence (by an extension of the
Knaster-Tarski theorem) gives rise to a model of all satisfiable instances.
%}

%\dw{is it okay to call it ``least'' although it isn't an ordering?} 
%\lo{Notice least now in quotation. I think we should call it \emph{canonical} model property.} \dw{okay} 
% (w.r.t.~a carefully chosen relation that respects formulas without negation and universal quantification).

\paragraph{Compactness Theorem and Semi-decidability of HoCHC}
A well-known feature of higher-order logic in standard semantics is failure of the compactness theorem.
As a consequence of HoCHC's refutational completeness, it follows that the {compactness theorem} \emph{does} hold for HoCHC (in standard semantics):
for every unsatisfiable {set} $\Set$ of HoCHCs, there is a finite subset $\Set' \subseteq \Set$ which is unsatisfiable.

Moreover, if the consistency of conjunctions of atoms in the background theory  is semi-decidable, so is HoCHC unsatisfiability.
Crucially, this underpins the \emph{practicality} of the HoCHC-based
approach to program verification.

\paragraph{Semantic Invariance}
The soundness and completeness of our resolution proof system has
another pleasing corollary: satisfiability of HoCHC does \emph{not}
depend on the choice of semantics\footnote{within the reasonable
  bounds formalised by \emph{(complete) frames}}~(\cref{SEC:EQUIV}). % choice of Henkin frame (\cref{sec:equiv}).
In particular, this constitutes an alternative proof of the equivalence of standard, monotone and continuous semantics for HoCHCs, without exhibiting explicit translations between semantics.
%\lo{What about ``effectively continuous semantics'', the semantics of Charlalambidis et al.?} \dw{this is already proven in \cite[Proposition 5.14]{CHRW13}}
Moreover, this demonstrates that, in contrast to (full) higher-order logic, satisfiability of HoCHCs with respect to standard semantics on the one hand, and to Henkin semantics on the other, coincide.
%\lo{Is HoCHC the ``largest'' such fragment?} \dw{I have no clue,
%although most obvious extension fail, but this can be future work}

%\changed[dw]{
Semantic invariance is an important advantage for program
verification. 
It follows that one can use (the simpler and more intuitive) standard semantics for specification, 
%but use continuous (which enjoys a richer structure) or Henkin semantics (for which more refined proof systems are complete \cite{A71,H72,BK98,BBCW18}) to compute and reason about solution methods and in static analysis. 
% \changed[lo]{
  but use Henkin semantics for the development of refined proof systems that are complete \cite{A71,H72,BK98,BBCW18},  
and use continuous semantics (which enjoys a richer structure) to
construct solution methods and in static analysis. % } 
% \dw{the sentence is a bit difficult to understand}
% \dw{thanks a lot. that seems to be clearer!}
%}

\paragraph{Complete 1st-order Translation}

As suggested by the equivalence of standard and Henkin semantics, we show that there is a variant of the standard translation of higher-order logic into 1st-order logic which is sound and complete also
for standard semantics, when restricted to HoCHC (\cref{sec:appenc}). 
%\changed[dw]{\sout{Precisely, we show that a finite set $S$ of HoCHCs is satisfiable iff its translation into 1st-order Horn constraints $\lfloor S \rfloor$ is satisfiable.}}
% Moreover if the unsatisfiability of goal clauses of the background theory is semi-decidable, then unsatisfiability of $\lfloor S \rfloor$ is also semi-decidable.
% \dw{thnk about} To underpin the practicality of this approach, we argue that the target logic of this encoding is semi-decidable if the background theory is decidable by building on work on automated theorem proving for 1st-order theories.

\paragraph{Decidable Fragments of HoCHC}
%\changed[dw]{
Satisfiability of finite sets of HoCHCs is trivially decidable for background theories with finite domains.
In \cref{sec:decidable}, we identify a fragment\ifproceedings \else\footnote{Another one (higher-order Datalog) is presented in \cref{sec:hodatalog}.} % \lo{Remove the preceding footnote -- we will not have space for appendix in the proceedings. Space permitting, mention / remark the decidability of Higher-order Datalog in the main text.}
\fi of HoCHC (the Bernays-Sch\"onfinkel-Ramsey fragment of HoCHC modulo a restricted form of Linear Integer Arithmetic) with a decidable satisfiability problem by showing \mbox{equi-satisfiability} to clauses w.r.t.\ a finite number of such background theories.
%}

% In \cref{sec:decidable}, we put the canonical model property of the canonical structure to good use by leveraging the fact that the structure is a model of satisfiable sets of HoCHCs which can be finitely computed for background theories with finite domains.
% We can identify two fragments of HoCHC \changed[dw]{(higher-order Datalog and the Bernays-Sch\"onfinkel-Ramsey
% fragment of HoCHC modulo a restricted form of Linear Integer
% Arithmetic)} with a decidable satisfiability problem by showing \mbox{equi-satisfiability} to clauses w.r.t.\ a finite number of such background theories.
  
% \changed[lo]{In \cref{sec:decidable}, we put the 1st-order translation to good use by leveraging it in our search for decidable fragments of HoCHCs. 
% The key insight is that given a finite set $S$ of HoCHCs, if the %universe(s) of the 1st-order background theory is 
% (interpretation of the) base types are finite, so is the canonical structure $\As_{P_S}$, 
% %generated by iterating the immediate consequence operator, 
% and we can check if $\As_{P_S} \models S$ holds.
% We can identify two fragments of HoCHC with a decidable satisfiability problem via translation to well-known 1st-order fragments.}

\paragraph*{Outline} 
We begin with some key definitions in \cref{sec:prelims}.
Then we show that even standard semantics satisfies a canonical model
property (\cref{sec:quasi-mon}).
In \cref{sec:resolution}, we present the resolution proof system for HoCHC and prove its completeness.
In \cref{sec:equiv} we show HoCHC's semantic invariance and in \cref{sec:compact-theories} we generalise the refutational completeness proof to compact background theories, which may have more than one model.
In \cref{sec:appenc} we present a 1st-order translation % encoding
of higher-order logic and prove it complete when restricted to HoCHC.
In \cref{sec:decidable} we exhibit decidable fragments of HoCHC.
Finally, we discuss related work in \cref{sec:relwork}, and conclude
in \cref{sec:conc}.

\ifproceedings
%\changed[dw]{
  For the extended version of the paper refer to \cite{OW19}.
%}
\fi

\section{Technical Preliminaries}
\label{sec:prelims}
\begin{figure*}[!t]
  \centering
  \begin{gather*}
    \infer[(\mathrm{Var})]{\Delta\vdash x\from\Delta(x)}{x\in\dom(\Delta)}\qquad\infer[(\mathrm{Cst})]{\Delta\vdash c\from\sigma}{c\from\sigma\in\Sigma}\qquad
    \infer[(\mathrm{App})]{\Delta\vdash
      M_1M_2\from\sigma_2}{\Delta\vdash
      M_1\from\sigma_1\to\sigma_2&\Delta\vdash M_2\from\sigma_1}\qquad
    \infer[(\mathrm{Abs})]{\Delta\vdash\lambda x\ldotp M\from\Delta(x)\to\rho}{\Delta\vdash M\from\rho}\\[2pt]
    \infer[(\mathrm{And/Or})]{\Delta\vdash\circ\from o\to o\to
      o}{\circ\in\{\land,\lor\}}\qquad\infer[(\mathrm{Neg})]{\Delta\vdash\neg
      M\from o}{\Delta\vdash M\from o}\qquad
    \infer[(\mathrm{Ex})]{\Delta\vdash\exists_\tau\from(\tau\to o)\to o}{}
  \end{gather*}
  \caption{Typing judgements}
  \label{fig:tjud}
\end{figure*}

%\lo{TODO: This section need to be condensed.}

This section introduces the syntax and semantics of a restricted form
of higher-order logic (\cref{sec:relhol}), higher-order constrained
Horn clauses (\cref{ch:defHoCHC}) and programs (\cref{sec:nftrans}).

\subsection{Relational Higher-order Logic}
\label{sec:relhol}
\subsubsection{Syntax}
For a fixed set $\btypes$ (intuitively the types of individuals),
% \dw{better name? ($\Bd,\mathbb B$ already in use)} \lo{I am ok with
% $\btypes$.}
% \dw{is the difference between $\btypes$ and $\iota$ enough or is
%   $\mathfrak I$ better}
% \lo{I think $\mathfrak I$ is better than $\btypes$.}
% \dw{ok}
the set of
\emph{argument types}, \emph{relational types}, \emph{1st-order
  types} and \emph{types} (generated by $\btypes$) are mutual recursively defined by
% \begin{align*}
%   \tau&::=\iota\mid\rho\\
%   \rho&::= o\mid\tau\to\rho\\
%   \sigma_{\FO}&::=o\mid\iota\mid\iota\to\sigma_{\FO}\\
%   \sigma&::=\rho\mid\sigma_{\FO},
            %   \end{align*}
\[
\begin{array}{lrclr}
\hbox{\em Argument type} & \tau & ::= & \iota\mid\rho\\
\hbox{\em Relational type} &  \rho&::= & o\mid\tau\to\rho\\
\hbox{\em 1st-order type} & \sigma_{\FO}&::=&\iota\mid \iota\to o\mid\iota\to\sigma_{\FO}\\
\hbox{\em Type} &  \sigma&::=&\rho\mid\sigma_{\FO},
\end{array}
\]
%\dw{$o$ is now not 1st-order type} \lo{Noted.}
where $\iota\in\btypes$.
We sometimes abbreviate the (1st-order) type $\underbrace{\iota\to\cdots\to\iota}_{n}\to\iota$ to $\iota^n\to\iota$ (similarly for $\iota^n\to o$). 
%We sometimes use the abbreviations $\iota^n\to\iota$ and $\iota^n\to o$ for the (1st-order) types $\underbrace{\iota\to\cdots\to\iota}_{n}\to\iota$ and $\underbrace{\iota\to\cdots\to\iota}_{n}\to o$, respectively. 
For types $\tau_1\to\cdots\to\tau_n\to\sigma$ we also write $\overline\tau\to\sigma$. 
Intuitively, $o$ is the type of the truth values (or Booleans). 
Besides, $\sigma_{\FO}$ contains all (1st-order) types of the form $\iota^n\to\iota$ or $\iota^n\to o$, 
i.e.\ all arguments are of type $\iota$. Moreover, each relational type has the form $\overline\tau\to o$.

%A \emph{type environment} is a set of distinct typed variables $x\from\tau$. We assume that for each type there are at least countably infinite many variables of that type and write $\Delta(x)=\tau$ if $\tau$ is the unique argument type such that $x\from\tau\in\Delta$. Besides, $\dom(\Delta)$ is the set of variables such that $x\from\tau\in\Delta$. 
A \emph{type environment} (typically $\Delta$) is a function mapping
variables 
%\changed[dw]{
(typically denoted by $x,y,z$ etc.)
%} 
to {argument} types; 
%\changed[dw]{\sout{such that $\Delta^{-1}(\tau)$ is infinite for each {argument} type $\tau$}};
for $x \in \dom(\Delta)$, we write $x : \tau\in\Delta$ to mean $\Delta(x) = \tau$.
A \emph{signature} is a set of distinct typed \emph{symbols}
$c\from\sigma$, where $c\nin\dom(\Delta)$ and $c$ is not one of the
\emph{logical symbols} $\neg$, $\land$, $\lor$ and $\exists_\tau$ (for
argument types $\tau$, which we omit frequently).
It is \emph{1st-order} if for each $c\from\sigma\in\Sigma$, $\sigma$ is 1st-order. 
We often write $c\in\Sigma$ if $c\from\sigma\in\Sigma$ for some $\sigma$.

The set of \emph{$\Sigma$-pre-terms} is given by
\begin{align*}
  M::=x\mid c\mid\neg\mid\land\mid\lor\mid\exists_\tau\mid
  MM\mid\lambda x\ldotp M
\end{align*}
where $c\in\Sigma$.
Following the usual conventions we assume that application associates to the left and the scope of abstractions extend as far to the right as possible.
We also write $M\,\overline N$ and $\lambda\overline x\ldotp M'$ for $M\,N_1\cdots N_n$ and $\lambda x_1\ldotp\cdots\lambda x_n\ldotp M'$, respectively, assuming implicitly that $M$ is not an application.
Besides, we abbreviate $\exists_\tau (\lambda x\ldotp M)$ as $\exists x\ldotp M$.
Moreover, we identify terms up to $\alpha$-equivalence and adopt Barendregt's \emph{variable convention} \cite{B12}.

The typing judgement $\Delta\vdash M\from\sigma$ is defined in \cref{fig:tjud}.
We say that $M$ is \emph{$\Sigma$-term} if $\Delta\vdash M\from\sigma$
for some $\sigma$ and it is a \emph{$\Sigma$-formula} if $\sigma=o$. % We use the notion of \emph{subterms} in the  standard sense.
A $\Sigma$-formula is a \emph{1st-order $\Sigma$-formula} if its
construction is restricted to symbols $c\from\sigma_{\FO}\in\Sigma$
and variables ${x\from\iota}\in\Delta$, and uses no
$\lambda$-abstraction.
 %\changed[dw]{it only contains symbols $c\from\sigma_{\FO}\in\Sigma$ and variables $x\from\iota\in\Delta$}.% for all subterms $N$ of $M$, $\Delta\vdash N\from\sigma_{\FO}$.
 %\lo{The preceding definition is a little ambiguous, and seems to admit $\lambda x^\iota.x$ as a 1st-order $\Sigma$-term.  [It is also worth noting that we implicitly assume that the term $f(t_1, t_2) : \iota$ (say, where $f$ is an arity-2 function symbol in $\Sigma$) is written using application $(f \, t_1) \, t_2$.] Thus I propose: A $\Sigma$-term ($\Sigma$-formula) $M$ is a \emph{1st-order $\Sigma$-term} (\emph{1st-order $\Sigma$-formula}) if it is constructed without $\lambda$-abstraction, and restricted to symbols $c\from\sigma_{\FO}\in\Sigma$ and variables ${x\from\iota}\in\Delta$.}
%\dw{it does not allow $\lambda x^\iota\ldotp x$ because $\iota$ is not relational (see Abs typing rule) but it \emph{does} allow $\lambda x,y\ldotp x\leq y$. So we should change it to your definition}
%\dw{\sout{We call a pre-term of the form $\lambda x\ldotp M$ a \emph{$\lambda$-abstraction}.}} 
Finally, for a $\Sigma$-term $M$,
$\free(M)$ is the set of free variables, and $M$ is a \emph{closed} % ground
$\Sigma$-term if $\free(M)=\emptyset$.

%\lo{ENGLISH: ``Besides'' feels a little overused. I have replaced / removed some occurrences.}

\begin{remark}
  \label{rem:isimple}
   It follows from the definitions that
   \begin{inparaenum}[(i)]
   \item 
   %in contrast to similar definitions in the literature, 
   each term $\Delta\vdash M\from \iota^n\to\iota$ can only contain variables of type $\iota$ and constants of non-relational 1st-order type, and contains neither $\lambda$-abstractions nor logical symbols (a similar approach is adopted in \cite{CHRW13});
   \item %moreover, the negation symbol ($\neg$) 
   $\neg$ can only occur in a term if applied to a formula (and not in pre-terms of the form $R\,\neg$). 
   \iffalse  
   \lo{Possibly confusing here. $\exists_o \neg$ and $R \, \neg$ are not typable. So when is it possible?} 
   \lo{Item (ii) is still confusing to me: $\neg$ and $\neg \, x^o$ are terms.}
   \lo{By (the highlighted) ``term'' in item (ii), do you mean a term of type $\iota^n \to \iota$?}
   \fi
 \end{inparaenum}
 \end{remark}
 The following kind of terms is particularly significant:
 \begin{definition}
 A $\Sigma$-term is \emph{positive existential} if the logical constant ``$\neg$'' is not a subterm.
   % \begin{thmlist}
   % \item A $\Sigma$-term $M$ is \emph{positive existential} if the logical constant ``$\neg$'' is not a subterm of $M$.
   % \item A $\Sigma$-formula $F$ is \emph{positive existential} if the logical constant $\neg$ is not a subterm of~$F$.
   % \end{thmlist}
 \end{definition}
 \iffalse
 Since we chose not to include a logical constant for universal quantification, this means that it is impossible to (implicitly) quantify variables universally in positive existential formulas.
 \fi 
 % \lo{Is the preceding sentence really necessary?}

For $\Sigma$-terms $M, N_1,\ldots,N_n$ and variables
$x_1,\ldots,x_n$ satisfying $\Delta\vdash N_i\from\Delta(x_i)$, the \emph{(simultaneous) substitution}
$M[N_1/x_1,\ldots,N_n/x_n]$ is defined in the standard way.
 
\subsubsection{Semantics}
%\changed[dw]{
There are two classic semantics for higher-order logic:
\emph{standard} and \emph{Henkin semantics} \cite{DBLP:journals/jsyml/Henkin50}.
Whereas in standard semantics the interpretation of higher types is uniquely determined by the domains of individuals (quantifiers range over \emph{all} set-theoretic functions of the appropriate type), 
it can be \emph{stipulated} quite liberally in Henkin semantics.

% \lo{I agree with your proposal is to change ``(pre-)Henkin frame $\cal H$'' to ``(pre-)frame $\cal F$''.}

Formally,
a \emph{pre-frame} $\Hf$ assigns to each type $\sigma$ a
non-empty set $\sinti\Hf\sigma$ such that 
\begin{enumerate}
\item $\sinti\Hf o\defeq\bool \defeq \{0,1\}$ and for each type $\sigma_1\to\sigma_2$,
$\sinti\Hf{\sigma_1\to\sigma_2}\subseteq[\sinti\Hf{\sigma_1}\to\sinti\Hf{\sigma_2}]$
\item $\andf,\orf\in\sinti\Hf{o\to o\to o}$
\item $\exists_\tau\in\sinti\Hf{(\tau\to o)\to o}$ for each argument
  type $\tau$
  \end{enumerate}
  where $[\sinti\Hf{\sigma_1}\to\sinti\Hf{\sigma_2}]$ is the set of
functions $\sinti\Hf{\sigma_1}\to\sinti\Hf{\sigma_2}$ and
\begin{align*}
  \andf(b_1)(b_2)&\defeq\min\{b_1,b_2\} \qquad
                   \orf(b_1)(b_2) \defeq\max\{b_1,b_2\}\\
  \hexists_\tau(r)&\defeq\max\{r(s)\mid s\in\sinti\Hf\tau\}
\end{align*}

%\begin{align*}
%   \andf(b_1)(b_2)&\defeq\min\{b_1,b_2\}\\
%   \orf(b_1)(b_2)&\defeq\max\{b_1,b_2\}\\
%   \hexists_\tau(r)&\defeq\max\{r(s)\mid s\in\sinti\Hf\tau\}\\
%   \bot^\Hf_{\overline\tau\to o}(\overline r)&\defeq 0\\
%   \top^\Hf_{\overline\tau\to o}(\overline r)&\defeq 1
% \end{align*}
%and for relational types $\rho$ and $\Rd\subseteq\sinti\Hf\rho$, 
% inductively by \dw{just say pointwise?}
% \begin{align*}
%   \bigsqcup\nolimits_o\Rd&\defeq\max\Rd&\text{for }\Rd\subseteq\sinti\Hf o\\
%   \left(\bigsqcup\nolimits_{\tau\to\rho}\Rd\right)(s)&\defeq\bigsqcup\nolimits_\rho\{r(s)\mid
%                                                        r\in\Rd\}\\
%                          &&\hspace{-4cm}\text{for $\Rd\subseteq\sinti\Hf{\tau\to\rho}$ and } s\in\sinti \Hf\tau.
% \end{align*}
% In the following, we omit the subscript from $\bigsqcup_\sigma$ since
% it can be inferred.
% Note that for $\rho=\overline\tau\to o$ and $\overline s\in\sinti \Hf{\overline\tau}$, $(\bigsqcup\Rd)(\overline s)=1$ iff there exists $r\in\Rd$ such that $r(\overline s)=1$. 

% In the following, we omit the subscript since it can be inferred, we omit the superscript whenever the pre-Henkin frame is clear from the context and just write $r\prel r'$.

%Next, we consider some examples of pre-Henkin frames \cite{BOR18,AJ95}:
\begin{example}[Pre-frames] %\cite{BOR18,AJ95}
  For every $\iota\in\btypes$, we fix an arbitrary non-empty set $D_\iota$. We
  define $\Sf$, $\Mf$ and $\Cf$, which we call the \emph{standard}, \emph{monotone} and
  \emph{continuous frame}, respectively,
  recursively by
  $\sinti{\Sf}
  o\defeq\sinti
  {\Mf}
  o\defeq\sinti
  {\Cf}
  o\defeq\bool$;
  $\sinti{\Sf}\iota\defeq\sinti
  {\Mf}\iota\defeq\sinti
  {\Cf}\iota\defeq D_\iota$
  for $\iota\in\btypes$; and
  \begin{align*}
    \sinti {\Sf}{\tau\to\sigma}&\defeq [\sinti {\Sf}\tau\to\sinti {\Sf}\sigma]\\
    \sinti {\Mf}{\tau\to\sigma}&\defeq [\sinti {\Mf}\tau\xrightarrow{m}\sinti {\Mf}\sigma]\\
    \sinti {\Cf}{\tau\to\sigma}&\defeq [\sinti {\Cf}\tau\xrightarrow{c}\sinti {\Cf}\sigma],
  \end{align*}
%\dw{previously $\sigma$ was $\rho$}
  where %$[\sinti {\Mf}\tau\xrightarrow{m}\sinti {\Mf}\rho]$ ($[\sinti {\Cf}\tau\xrightarrow{c}\sinti {\Cf}\rho]$) is the set of monotone (continuous) functions $\sinti {\Mf}\tau\to\sinti {\Mf}\rho$ ($\sinti {\Cf}\tau\to\sinti {\Cf}\rho$) with respect to the respective pointwise ordering. 
  $[P \xrightarrow{m} P']$ ($[P \xrightarrow{c} P']$) is the set of
  monotone (continuous) functions from the posets $P$ to $P'$ (cf. \cite{AJ95}).
\end{example}
%For the purpose of the remainder of this subsection, 

Let $\Sigma$ be a signature and $\Hf$ be a pre-frame. 
A \emph{$(\Sigma,\Hf)$-structure} $\As$ assigns to each $c\from\sigma\in\Sigma$ an element
$c^\As\in\sinti\Hf\sigma$ and we set $\sinti\As\sigma\defeq\sinti\Hf\sigma$
for types $\sigma$. 
A \emph{$(\Delta,\Hf)$-valuation} $\alpha$ is a
function such that for every $x\from\tau\in\Delta$,
$\alpha(x)\in\sinti\Hf\tau$. 
For a $(\Delta,\Hf)$-valuation $\alpha$, variable $x$ and $r\in\sinti\Hf{\Delta(x)}$,
$\alpha[x\mapsto r]$ is defined in the usual way.

The \emph{denotation} $\sinti\As M(\alpha)$ of a term $M$ with respect
to $\As$ and $\alpha$ is defined recursively % inductively
by
% \begin{align*}
%   \sinti\As x(\alpha)&\defeq \alpha(x) \qquad
%   \sinti\As c(\alpha)\defeq c^\As\\
%   \sinti\As\land(\alpha)&\defeq\andf \qquad
%   \sinti\As\lor(\alpha)\defeq\orf\\
%   \sinti\As{\exists_\tau}(\alpha)&\defeq\hexists_\tau\qquad
%   \sinti\As{\neg M}(\alpha)\defeq 1-\sinti\As M(\alpha)\\
%   \sinti\As{M_1 \, M_2}(\alpha)&\defeq \sinti\As{M_1}(\alpha)(\sinti\As{M_2}(\alpha))\\
%   \sinti\As{\lambda x\ldotp M}(\alpha)&\defeq\embed{\lambda r\in\sinti\As{\Delta(x)}\ldotp\sinti\As M(\alpha[x\mapsto r])}_{\Delta(x)\to\rho}
% \end{align*}

\vspace*{-2mm}\noindent
\begin{minipage}{0.6\linewidth}
  \begin{align*}
    \sinti\As x(\alpha)&\defeq \alpha(x) \qquad
  \\
  \sinti\As\land(\alpha)&\defeq\andf \qquad
  \\
  \sinti\As{\exists_\tau}(\alpha)&\defeq\hexists_\tau\qquad
    \\
  \sinti\As{M_1 \, M_2}(\alpha)&\defeq \sinti\As{M_1}(\alpha)(\sinti\As{M_2}(\alpha))\\
  \sinti\As{\lambda x\ldotp M}(\alpha)&\defeq\embed{\lambda r\in\sinti\As{\Delta(x)}\ldotp\sinti\As M(\alpha[x\mapsto r])}_{\Delta(x)\to\rho}
  \end{align*}
\end{minipage}
\begin{minipage}{0\linewidth}
  \begin{align*}
    \hspace*{-43mm}\sinti\As c(\alpha)&\defeq c^\As\\
    \hspace*{-43mm}\sinti\As\lor(\alpha)&\defeq\orf\\
    \hspace*{-43mm}\sinti\As{\neg M}(\alpha)&\defeq 1-\sinti\As
                                             M(\alpha)\\
    \hspace*{-43mm}\\
    \hspace*{-43mm}
  \end{align*}
\end{minipage}
\vspace*{0mm}

\noindent (assuming $\Delta\vdash M\from\rho$ in the last case), where
$\embed{r}_\sigma=r$ if $r\in\sinti\As\sigma$ and otherwise
$\embed{r}_\sigma\in\sinti\As\sigma$ is arbitrary.
Thus, for each term $\Delta\vdash M\from\sigma$, $\sinti\As
M(\alpha)\in\sinti\As\sigma$. % If $M$ does not contain variables then
                              % for any $(\Delta,\Hf)$-valuations
                              % $\alpha$ and $\alpha'$, $\sinti\As
                              % M(\alpha)=\sinti\As M(\alpha')$  and
                              % therefore we just write $\sinti\As M$
                              % for $\sinti\As M(\alpha)$.

Being independent of valuations, the
denotation of closed terms $M$ is abbreviated as $\sinti\As M$.
Besides, for $\Sigma$-formulas $F$, we write $\As, \alpha\models F$ if $\sinti\As F(\alpha)=1$, and $\As\models F$ if for all $\alpha'$, $\As,\alpha'\models F$.
We extend $\models$ in the usual way to sets of formulas.

A \emph{frame} is a pre-frame $\Hf$ that satisfies the
% \dw{is it right to have capital letters?}
% \lo{It is not completely clear here.  The rule is: use upper case for proper noun. 
% Our axiom is a version of the Comprehension Axiom.
% On balance, let's use lower case.}
                                                                   %                                                                    \changed[dw]{

{

\advance\leftmargini -2mm
\begin{quote}
\emph{Comprehension Axiom}:
  % \begin{complist}
  %   \setcounter{complisti}{2}
  % \item\label{def:Hcomp}
    for each signature $\Sigma$, type environment
    $\Delta$, $(\Sigma,\Hf)$-structure $\As$, $(\Delta,\Hf)$-valuation
    $\alpha$, positive existential $\Sigma$-term $\lambda x\ldotp
    M$, and $r\in\sinti\As{\Delta(x)}$,
    $\sinti\As{\lambda x\ldotp M}(\alpha)(r)=\sinti\As
    M(\alpha[x\mapsto r])$.
    % \end{complist}
  \end{quote}
}
%}
% \begin{quote}
% for each signature $\Sigma$, type environment $\Delta$, $(\Sigma,\Hf)$-structure $\As$, $(\Delta,\Hf)$-valuation $\alpha$, and positive existential $\Sigma$-term $\lambda x\ldotp M$,
% $\As\langle \lambda x\ldotp M\rangle(\alpha)=\sinti\As{\lambda x\ldotp M}(\alpha)$.
% \end{quote}

Our comprehension axiom 
%\changed[dw]{
ensures that positive existential terms are interpreted in the expected way;
%} 
it is  non-standard in that it is % they are
restricted to positive existential formulas.
%}

%\changed[dw]{
As a consequence, if $\Hf$ is a frame then for every relational type $\overline\tau\to o$, $\top_{\overline\tau\to o}\in\sinti\Hf{\overline\tau\to o}$, where $1=:\top_{\overline\tau\to o}(\overline r)=\sinti\As{\lambda\overline x\ldotp y}(\alpha[y\mapsto 1])(\overline r)$.

  % For a type environment $\Delta$ let $\top_\Delta$ be the
  % valuation satisfying $\top_\Delta(x)=\top_\rho$ for
  % $x\from\rho\in\Delta$ and $\top_\Delta(x)=h_\iota$ for
  % $x\from\iota\in\Delta$.
  % If a term $M$ is closed, its denotation does not depend on the valuation, and we set $\sinti\As M\defeq\sinti\As M(\top_\Delta)$.
%}

% if it is clear from the context.

\paragraph{Complete Frames}
% \dw{check where completeness of frames is assumed}
For types $\sigma$, let ${\prel_\sigma}\subseteq\sinti\Hf\sigma\times\sinti\Hf\sigma$ be the usual partial
order defined pointwise for higher types, 
%\changed[dw]{
which is the discrete order on $\sinti\Hf\iota$ and the ``less than or equal'' relation on $\sinti\Hf o$.
%}

For relational types $\rho$ and $\Rd\subseteq\sinti\Hf\rho$, the
least upper bound $\bigsqcup_\rho\Rd$ is defined pointwise, by recursion
on $\rho$. In particular,
$\bigsqcup_{\overline\tau\to o}\emptyset=\bot_{\overline\tau\to o}$,
where $\bot_{\overline\tau\to o}(\overline r) \defeq 0$.
For a singleton set $\{f\}\subseteq\sinti\Hf{\iota^n\to\iota}$ we define $\bigsqcup_{\iota^n\to\iota}\{f\} \defeq f$. 
 Throughout the paper, we omit type subscripts to reduce clutter
 because they can be inferred.

 A (pre-)frame $\Hf$ is \emph{complete} if for every relational $\rho$ and $\Rd\subseteq\sinti\Hf\rho$, $\bigsqcup\Rd\in\sinti\Hf\rho$,
 i.e.\ each $\sinti\Hf\rho$ is
a complete lattice ordered by $\prel_\rho$ with least upper bounds $\bigsqcup_\rho$.

\begin{example}[complete frames]
  \label{ex:Henkin}
  $\Sf$ is trivially a complete frame. 
  It is not difficult to prove that $\Mf$ and $\Cf$ are also
  complete frames \ifproceedings\cite{OW19}\else(\cref{sec:appsmcHenkin})\fi.
\end{example}

\paragraph{1st-order Structures}
Let $\Sigma$ be a 1st-order signature. A \emph{1st-order $\Sigma$-structure} is a $(\Sigma,\Sf)$-structure. Note that by taking standard frames this coincides with the standard definition in a purely 1st-order setting (cf.\ e.g.\ \cite{CK13}).
\begin{example}
  \label{ex:folstr}
  % \begin{thmlist}
  % \item 
    In the examples we will primarily be concerned with the signature of \emph{Linear Integer Arithmetic}\footnote{with the usual types $0,1\from\iota$; ${+}, {-}\from\iota\to\iota\to\iota$ and $\triangleleft\from\iota\to\iota\to o$ for $\triangleleft\;\in\{<,\leq,=,\neq,\geq,>\}$; and we use the common abbreviation $n$ for $\underbrace{1+\cdots+1}_n$, where $1\leq n\in\nat$} $\Sigma_{\LIA} \defeq \{0,1,+,-,<,\leq,=,\neq,\geq,>\}$ and its standard model $\As_{\LIA}$.
\end{example}

\subsection{Higher-order Constrained Horn Clauses}
\label{ch:defHoCHC}
\begin{assumption}
  Henceforth, we fix a 1st-order signature $\Sigma$ over a single
  type of individuals $\iota$ 
  %\changed[dw]{\sout{(for which we can assume an equality symbol)}} 
  and a 1st-order $\Sigma$-structure $\As$.
  
%Note that this always allows us to include equality of the base type.
Moreover, we fix a signature $\Sigma'$ extending $\Sigma$ with (only) symbols of relational
type, and a type environment $\Delta$ such that $\Delta^{-1}(\tau)$ is infinite for each {argument} type $\tau$.
\end{assumption}
% \dw{Do you prefer the ``old'' assumption style (i.e. numbered and italic?)}
% \lo{I think it is a good thing to be explicit about
%   assumptions. Numbering it is not crucial, but italicising it makes
%   it stand out, which is useful.}
% \dw{ok}
Intuitively, $\Sigma$ and $\As$ correspond to the language and interpretation of the background theory, e.g.\ $\Sigma_{\LIA}$ together with its standard model $\As_{\LIA}$. 
In particular, we first focus on background theories with a single model. 
In \cref{sec:compact-theories} we extend our results to a more general setting.

%For reasons of convenience, we defined 1st-order structures with respect to standard frames. However, in general there are many other Henkin frames $\Hf$ that have rich enough function spaces and thus allow us to regard $\As$ also as a $(\Sigma,\Hf)$-structure. Furthermore, we will be interested in whether 1st-order structures can be expanded to larger (higher-order) signatures. This is made precise by the following:
We are interested in whether 1st-order structures can be expanded to larger (higher-order) signatures. 
This is made precise by the following:
\begin{definition}
  \begin{thmlist}
  \item A frame $\Hf$ \emph{expands} $\As$ if
    $\sinti\Hf\iota=\sinti\As\iota$ and $c^\As\in\sinti\Hf\sigma$ for all
    $c\from\sigma\in\Sigma$.
  \item Suppose  $\Hf$ expands $\As$. Then a $(\Sigma',\Hf)$-structure $\Bs$ is a \emph{$(\Sigma',\Hf)$-expansion of $\As$} if $c^\As=c^\Bs$ for all $c\in\Sigma$.
  \end{thmlist}
\end{definition}
\begin{remark}
%\changed[dw]{
  \begin{thmlist}
  \item\label{rem:deneq} By \cref{rem:isimple} the denotation of terms $\Delta\vdash
  M\from\iota^n\to\iota$ is the same for all
  $(\Sigma',\Hf)$-expansions of $\As$ and $(\Delta,\Hf)$-valuations
  agreeing on $\Delta^{-1}(\tau)$.
  \item\label{rem:exlattice} In case $\Hf$ is complete, the $(\Sigma',\Hf)$-expansions of $\As$
  ordered by $\prel$ constitute a complete lattice
  with least upper
  bounds $\bigsqcup$, where $\prel$ and $\bigsqcup$ are lifted in a
  pointwise fashion to $(\Sigma',\Hf)$-expansions of
  $\As$.\footnote{This is possible because $(\Sigma',\Hf)$-expansions
    of $\As$ agree on symbols of type $\iota^n\to\iota$.}
\end{thmlist}
%}
\end{remark}

%   where $\bot^\Hf_{\Sigma'}$ and $\bigsqcup\Bd$ (for a set $\Bd$ of
% expansions of $\As$) are the $\Sigma'$-expansion of $\As$
% defined by $R^{\bot^\Hf_{\Sigma'}}\defeq\bot^\Hf_\rho$ and $R^{\bigsqcup\Bd}\defeq\bigsqcup\{R^\Bs\mid\Bs\in\Bd\}$, respectively, for relational
% symbols ${R\from\rho}\in\Sigma'\setminus\Sigma$. Note that by definition
% of (pre-)Henkin frames, all $\As^\Hf_{P,\beta}$ are indeed
% $(\Sigma',\Hf)$-structures and expansions of $\As$. If no confusion arises, we abbreviate $\As^\Hf_{P,\beta}$ to $\As^\Hf_\beta$ in the following.

% We omit the signature $\Sigma'$ in the following whenever no confusion arises otherwise.

% In the tradition of 1st-order logic and in order not to lose focus on the essentials, we develop our proof systems not for HoCHP but for problems in a closely related ``clausal'' form. 
% In particular, we do not want to deal with terms of the form $R\,\overline M$, $x\,\overline M$ or $(\lambda x\ldotp N)\overline M$, where $\overline M$ or $N$ contain logical symbols.

Next, we introduce higher-order constrained Horn clauses and their
satisfiability problem.

\begin{definition}
  \begin{thmlist}
  \item An \emph{atom} is a $\Sigma'$-formula that does not contain a logical symbol.
  \item An atom is a \emph{background atom} if it is also a 1st-order $\Sigma$-formula. 
  %\changed[dw]{\sout{containing only variables of type $\iota$ (if any).}} 
  Otherwise it is a \emph{foreground atom}.
  % \item An \emph{easy foreground atom} is a foreground atom of the form $x\,\overline M$.
  \end{thmlist}
\end{definition}
\iffalse
\lo{In lambda-calculus, terms of the form $x \, \overline M$ are typically the hardest to analyse. 
So it feels counterintuitive to call them \emph{trivial}. 
Where in the (automated deduction?) literature are them called trivial?}
\dw{I came up with that name because clauses only containing this type of foreground atoms can be refuted in one step. Do you have an alternative suggestion for a name?}
\lo{Yes I understand. Perhaps call them \emph{easy}, followed by a brief explanation?}
\dw{for an explanation see below}
\fi
Note that a foreground atom has one of the following forms: 
\begin{inparaenum}[(i)]
\item $R\,\overline M$ where $R\in{(\Sigma'\setminus\Sigma)}$, 
\item $x\,\overline M$, or % and  
%\changed[dw]{\sout{(and $x$ is not 1st-order)}}
\item $(\lambda y\ldotp N)\overline M$.
\end{inparaenum}
%Furthermore, $x$ in $x\,\overline M$ is clearly not a 1st-order variable.

We use $\phi$ and $A$ (and variants thereof) to refer to background atoms and general atoms, respectively.

\begin{definition}[HoCHC]
  \begin{thmlist}
  \item A \emph{goal clause} is a disjunction $\neg A_1\lor\cdots\lor\neg A_n$, where each $A_i$ is an atom.
  We write $\bot$ to mean the empty (goal) clause. 
  % \item $G$ is an \emph{easy goal clause} if each $A_i$ is an easy foreground atom. 
  \item If $G$ is a goal clause,  $R\in {(\Sigma'\setminus\Sigma)}$ and the variables in $\overline x$ are distinct, then $G\lor R\,\overline x$ is a \emph{definite clause}. 
  \item A \emph{(higher-order) constrained Horn clause (HoCHC)} is a goal or definite clause.
  \end{thmlist}
\end{definition}
\iffalse
\changed[dw]{HoCHCs in which all foreground atoms are easy are
  ``easy'' in the sense that their satisfiability is solely determined
  by their background atoms (cf.\ \cref{lem:ignoresimple}).}
\dw{Is this really necessary given the limited space?}
\lo{Probably not. This can go if pressed for space.}
\fi
In the following we transform the higher-order sentences in \cref{ex:iter-working} into HoCHCs (by first converting to prenex normal form and then omitting the universal quantifiers).
\begin{example}[A system of HoCHCs]
  \label{ex:HoCHC}
  Let
  $\Sigma'=\Sigma_{\LIA}\cup\{\Add\from\iota\to\iota\to\iota\to
  o,\Iter\from(\iota\to\iota\to\iota\to o)\to\iota\to\iota\to\iota\to
  o\}$ and let $\Delta$ be a type environment satisfying
  $\Delta(x)=\Delta(y)=\Delta(z)=\Delta(n)=\Delta(s)=\iota$ and
  $\Delta(f)=\iota\to\iota\to\iota\to o$. 
  \begin{align*}
    \neg(z=x+y)\lor \Add\,x\,y\,z\\
    \neg(n\leq 0)\lor\neg (s=x) \lor \Iter\,f\,s\,n\,x\\
    \neg(n>0)\lor\neg \Iter\,f\,s\,(n-1)\,y\lor\neg(f\,n\,y\,x) \lor \Iter\,f\,s\,n\,x\\
    \neg (n\geq 1)\lor\neg\Iter\,\Add\,n\,n\,x\lor\neg(x\leq n+n)
  \end{align*}
  %\changed[dw]{
  We refer to the first three (definite) HoCHCs as $D_1$ to $D_3$ and to the last (goal) HoCHC as $G$.%}
%  \lo{The expression $n-1$ in the 3rd clause is not a $\Sigma'$-term.}
%  \dw{now it is, isn't it?}
%  \lo{yes it is now.}
\end{example}
% \lo{A concern in explicitly relating \cref{ex:HoCHC} to the higher-order sentences in \cref{ex:iter-working} is that the above translation to HoCHCs does not distinguish existentially and universally quantified variables. And this might confuse referees.
%   For this reason, \cref{ex:program} may be a better place to make the connection with \cref{ex:iter-working}, but then the 4th sentence in \cref{ex:iter-working}---the safety assertion---is not discussed in our translation so far.}
% \dw{I don't understand your concern. In \cref{ex:iter-working} all variables are explicitly universally quantified (after computing the prenex normal form. Here all variables are implicitly universally quantified.)}
% \lo{You are absolutely right. I had overlooked the conversion to prenex normal form.}

% \dw{what do you think of using $\Set$ instead of $S$ for sets of
%   HoCHCs (motivation: make clear that this is a new ``category'' and
%   increase difference to standard frame $\Sf$); similarly for $\Prgm$
%   and $P$}
% \lo{Good idea!}
\begin{definition}
  \label{def:sat}
  Let $\Set$ be a set of HoCHCs, and suppose $\Hf$ is a frame expanding $\As$. 
  \begin{thmlist}
  \item $\Set$ is \emph{$(\As,\Hf)$-satisfiable} if there
    exists a $(\Sigma',\Hf)$-expansion $\Bs$ of
    $\As$ satisfying $\Bs\models\Set$.
  \item $\Set$ is \emph{$\As$-Henkin-satisfiable} if it is $(\As,\Hf)$-satisfiable for
    some frame $\Hf$ expanding $\As$.
  \item $\Set$ is \emph{$\As$-standard-satisfiable}
    % (\emph{$\As$-monotone-}, \emph{$\As$-continuous-satisfiable})
    if
    it is $(\As,\Sf)$-satisfiable.%  ($(\As,\Mf)$-,
    % $(\As,\Cf)$-satisfiable).
  \item $\Set$ is \emph{$\As$-monotone-satisfiable} if
    it is $(\As,\Mf)$-satisfiable.
    \item $\Set$ is \emph{$\As$-continuous-satisfiable} if
    it is $(\As,\Cf)$-satisfiable.
  \end{thmlist}
\end{definition}
Observe that $\As$-Henkin satisfiability is trivially implied by all
notions of satisfiability in \cref{def:sat}.

\subsection{Programs}
\label{sec:nftrans}
Whilst HoCHCs have a simple syntax (thus yielding a simple proof
system), our completeness proof relies on programs, which are syntactically slightly more complex.
\begin{definition}\label{def:program}
  A \emph{program} (usually denoted by $\Prgm$) is a set of $\Sigma'$-formulas
  $\{\neg F_R\lor R\,\overline x_R\mid R\in{(\Sigma'\setminus\Sigma)}\}$ such that
  for each $R\in\Sigma'\setminus\Sigma$, $F_R$ is positive existential,
  the
  variables in $\overline x_R$ are distinct, and $\free(F_R)\subseteq\free(R\,\overline x_R)$.
\end{definition}
% If $F$ is a positive existential $\Sigma'$-formula then $\neg F$ is a \emph{$(\Sigma,\Sigma')$-goal formula}. Besides, 

%  A $\Sigma'$-term which is a $(\Sigma,\Sigma')$-definite or goal formula is a \emph{$(\Sigma,\Sigma')$-Horn formula}. 

For each goal clause $G$ there is a closed positive existential formula\footnote{%\sout{$\posex(\neg A_1\lor\cdots\lor\neg A_n)\defeq\exists\overline y\ldotp A_1\land\cdots\land A_n$, where $\overline y=\bigcup_{i=1}^n\free(A_i)$}%
see \ifproceedings\cite{OW19} \else\cref{sec:appnftrans} \fi for details}
$\posex(G)$ %(without free variables) 
such that for each frame $\Hf$ and $(\Sigma',\Hf)$-structure $\Bs$, 
% <<<<<<< HEAD
% $\Bs\not\models_\Hf G$ iff $\Bs\models_\Hf\posex(G)$.
% Similarly, for a finite set of HoCHCs $S$, there exists a program\footnote{refer to \cref{sec:appprelims} for details} $P_S$ such that for each Henkin frame $\Hf$ and $(\Sigma',\Hf)$-structure $\Bs$, $\Bs\models_\Hf\{D\in S\mid D\text{ definite}\}$ iff $\Bs\models_\Hf P_S$.
% % \lo{Footnotes repositioned (away from maths~expressions) to avoid confusing footnote labels with exponents.}
% % \dw{does this give enough information?}\lo{yes}
$\Bs\not\models G$ iff $\Bs\models\posex(G)$.\addtocounter{footnote}{-1}\addtocounter{Hfootnote}{-1}
Similarly, for each finite set of HoCHCs $\Set$, there exists a
program\footnotemark{} $\Prgm_\Set$ such that for each frame $\Hf$ and $(\Sigma',\Hf)$-structure $\Bs$, $\Bs\models\{D\in\Set\mid D\text{ definite}\}$ iff $\Bs\models \Prgm_\Set$.
%\lo{Footnotes repositioned (away from maths~expressions) to avoid confusing footnote labels with exponents.}
%\lo{The definition of $\posex$ in the footnote may be  confusing: although it (correctly) existentially quantifies all variables, it overlooks the fact that some variables should \emph{morally} be universally quantified, or kept free in case $\forall$ is not available.}
%\dw{I don't understand your point here. Variables in goal clauses $G$ are (implicitly) universally quantified and thus (explicitly) existentially quantified in $\posex(G)$. But I agree that the exact definition might be unnecessary/obvious.}
% \dw{does this give enough information?}\lo{yes}
% >>>>>>> eca22d8e8a99755341c68f7ffca02829dcad347e
    
\begin{example}[Program]
\label{ex:program} The following program corresponds to
  the set of (definite) HoCHCs of \cref{ex:HoCHC} (modulo renaming of variables): 
\[
\begin{array}{l}
    \neg(z=x+y) \, \lor \, \Add\,x\,y\,z\\
    \neg \big((n\leq 0 \; \land \; s=x) \; \lor\\
    \quad \;\;(\exists y\ldotp n>0\land\Iter\,f\,s\,(n-1)\,y\land f\,n\,y\,x)\big) \lor \Iter\,f\,s\,n\,x.
\end{array}
\]
  % \begin{align*}
  %   \{\neg(&z=x+y)&&\hspace{-0.3cm}\lor \Add\,x\,y\,z,\\
  %   \neg((&n\leq 0\land s=x)\lor\\
  %   (&\exists y\ldotp n>0\land\Iter\,f\,s\,(n-1)\,y\land f\,n\,y\,x))&&\hspace{-0.3cm}\lor \Iter\,f\,s\,n\,x\}.
  % \end{align*}
\end{example}

% Therefore, we get:
% \begin{corollary}
%   \label{lem:assocprgmnew} $S$ is $(\As,\Hf)$-satisfiable if and only if $(P(S),F(S))$ is $(\As,\Hf)$-solvable.
% \end{corollary}

\section{Canonical Model Property}
% \dw{call this section ``Least Model Property of $\As_P$'' instead or is this confusing because it isn't an ordering?}
% \lo{What about ``Canonical Model Property of $\As_P$''?}
\label{sec:quasi-mon}

The introduction of monotone semantics for HoCHC in \cite{BOR18} was
partly motivated by the observation that the least model property (w.r.t. the
pointwise ordering $\prel$) fails for standard semantics (but holds for monotone semantics):
% Sets of HoCHCs do not in general have least standard models with respect to the pointwise ordering
% $\prel$ \cite{BOR18}.
% \dw{omit example?}
% \lo{Yes \cref{ex:notminmod} could be moved to the Appendix, but it would be nice to retain \cref{ex:notminmod2} so as to show that $\mrel$ accomplishes something ``useful'', which is unexpected.}
  \begin{example}
    \label{ex:notminmod}
    Consider the program $\Prgm$
    \begin{align*}
      \neg x_R\,U&\lor R\,x_R& \neg x_U\neq x_U\lor U\,x_U
    \end{align*}
    with signature $\Sigma'=\Sigma_{\LIA}\cup\{R\from((\iota\to o)\to o)\to
    o,U\from\iota\to o\}$, a type environment $\Delta$ satisfying
    $\Delta(x_R)=(\iota\to o)\to o$ and $\Delta(x_U)=\iota$ taken from \cite{BOR18}. Let $\Hf=\Sf$ be the standard frame and let $\nega\in\sinti \Sf{(\iota\to o)\to o}$ be such that $\nega(s)=1$ iff $s=\bot_{\iota\to o}$. % and $\nega(\top^\Sf_{\iota\to o})=0$ (which clearly exists in $\sinti \Sf{(\iota\to o)\to o}$).

    There are (at least) two expansions $\Bs_1$ and $\Bs_2$ defined by
    $U^{\Bs_1}=\bot_{\iota\to o}$ and $R^{\Bs_1}(s)=1$ iff $s(\bot_{\iota\to o})=1$, and $U^{\Bs_2}=\top_{\iota\to o}$ and $R^{\Bs_2}(s)=1$ iff $s(\top_{\iota\to o})=1$, respectively.
    
    % \dw{margin}
    % \begin{align*}
    %   U^{\Bs_1}&=\bot^\Sf_{\iota\to o}&U^{\Bs_2}&=\top^\Sf_{\iota\to o}\\
    %   R^{\Bs_1}(s)&=
    %                 \begin{cases}
    %                   1&\text{if }s(\bot^\Sf_{\iota\to o})=1 \\
    %                   0&\text{otherwise}
    %                 \end{cases}&
    %                              R^{\Bs_2}(s)&=
    %                                            \begin{cases}
    %                                              1&\text{if }s(\top^\Sf_{\iota\to o})=1 \\
    %                                              0&\text{otherwise}
    %                                            \end{cases}
    % \end{align*}
    Note that $\Bs_1\models \Prgm$, $\Bs_2\models \Prgm$ and there are no models smaller than any of these with respect to the pointwise ordering $\prel$. Furthermore, neither $\Bs_1\prel\Bs_2$ nor $\Bs_2\prel\Bs_1$ holds because $R^{\Bs_1}(\nega)=1>0=R^{\Bs_2}(\nega)$ and for any $n\in\sinti \Sf\iota$, $U^{\Bs_2}(n)=1>0=U^{\Bs_1}(n)$.
  \end{example}
  
In this section, we sharpen and extend the result: HoCHC \emph{does} enjoy a \emph{canonical} (though not
least w.r.t.\ $\prel$) model property. More precisely, the structure obtained by
iterating the \emph{immediate consequence operator} (see e.g.\ \cite{CHRW13}) is a model of all
satisfiable HoCHCs.

\begin{assumption}
  \label{ass:setting}
For \cref{sec:resolution,sec:quasi-mon} we fix a
%\changed[dw]{
complete
%} 
frame $\Hf$ expanding $\As$.
Furthermore, let $\Set$ be a finite set of HoCHCs and let $\Prgm=\Prgm_\Set$ (the program corresponding to $\Set$).
\end{assumption}

If no confusion arises, we refrain from mentioning $\Sigma'$, $\Delta$
and $\Hf$ explicitly. % \dw{also useful in e.g.\ $T_P^\Hf$?} \lo{yes} \dw{just to clarify: you would leave it as ``$T_P^\Hf$''?} \lo{Yes, I think ``$T_P^\Hf$'' is ok -- I would leave it. One advantage is the indication that it works for arbitrary Henkin frame $\Hf$.} ok

Given an expansion $\Bs$ of $\As$, the \emph{immediate consequence operator} $T_\Prgm$ returns the expansion $T_\Prgm(\Bs)$ of $\As$ defined by
$R^{T_\Prgm(\Bs)} \defeq \sinti\Bs{\lambda\overline x_R\ldotp F_R}$,
for relational symbols $R\in\Sigma'\setminus\Sigma$. (Recall that
$F_R$ is the unique positive existential formula such that $\neg
F_R\lor R\,\overline x_R\in \Prgm$.)
Observe that the prefixed points of $T_\Prgm$
(i.e. structures $\Bs$ such that $T_\Prgm(\Bs)\prel\Bs$) are precisely
the models of $\Prgm$.

Unfortunately, the immediate consequence operator is not monotone
w.r.t.\ $\prel$. Hence, we cannot apply the Knaster-Tarski
theorem. Therefore, we introduce the notion of
\emph{quasi-monotonicity} and a slightly stronger version of that theorem.
This is a warm-up for \cref{subsec:quasicont}, where we propose
\emph{quasi-continuity} and 
%\changed[dw]{
a version
%}
% an extension
of Kleene's fixed point theorem.

\subsection{Quasi-monotonicity}
\begin{assumption}
Let $L$ be a complete lattice ordered by $\leq$ with least upper
bounds $\biglor$ and least element $\bot$. Furthermore, let $F\from
L\to L$ be an (endo-)function.  
\end{assumption}
We define
\begin{align*}
  a_{\beta+1}&\defeq F(a_\beta)&(\beta\in\On)\\
  a_\gamma&\defeq \biglor_{\beta<\gamma}a_\beta&(\gamma\in\Lim)\\
  a_F&\defeq \biglor_{\beta\in\On}a_\beta
\end{align*}
In particular, $a_0=\bot$. Clearly, $a_F,a_\beta\in L$ for all ordinals $\beta$.

\begin{definition}
  \label{def:qm}Let ${\arel}\subseteq L\times L$ be a relation.
  \begin{thmlist}
  \item $\arel$ is \emph{compatible} with $\leq$ if
    \begin{enumerate}[noitemsep,label=(C\arabic*),leftmargin=8.0mm]
    \item\label{def:rel1} for all $a,b,c\in L$, if $a\arel b$ and $b\leq
      c$ then $a\arel c$,
    \item\label{def:rel2} for all $a\in L$ and $A\subseteq\{b\in L\mid
      b\arel a\}$, $\biglor A\arel a$.
    % \item%\label{def:rel11}
    %   for all $A\subseteq L$ and $a\in A$,
    %   if $a\arel a$ then $a\arel\biglor A$,  
    \end{enumerate}
    % \dw{(C3) is redundant (consequence of \cref{def:rel1}),
    %   isn't it?}
  \item $F$ is \emph{quasi-monotone} if for all $a,b\in L$, $a\arel b$ implies $F(a)\arel F(b)$.
  \end{thmlist}
\end{definition}
In particular, $\leq$ is compatible to itself and $\bot\arel
a$ for $a\in L$.
%\changed[dw]{\sout{\ref{def:rel1} and \ref{def:rel2} enforce a basic compatibility of $\arel$ with the lattice structure of $L$.}}

\begin{restatable}{proposition}{knastertarski}
  \label{lem:knaster-tarski} 
  \begin{enumerate*}
  \item $F(a_F)\leq a_F$ and 
  \item if $\arel$ is
  compatible with $\leq$, $F$ is quasi-monotone and $b\in L$ satisfies $F(b)\leq b$ then $a_F\arel b$.
  \end{enumerate*}
\end{restatable}
The proof idea is the same as for the standard Knaster-Tarski theorem,
which can be recovered from the above by using $\leq$ for $\arel$.

\subsection{Application to the Immediate Consequence Operator}
The idea now is to instantiate $L$ with the complete lattice of
expansions of $\As$ (see \cref{rem:exlattice}), and $F$ with
the immediate consequence operator $T_\Prgm$. We denote the structure at stage $\beta$ by $\As_\beta$
and the limit structure by $\As_\Prgm$.

%\changed[dw]{\sout{Note that  $\lambda\overline x_R\ldotp F_R$ does not contain free variables and $T^\Hf_P(\Bs)$ is a $(\Sigma',\Hf)$-structure by the definition of Henkin frames.}}
% Using $T^\Hf_P$ %the immediate consequence operator 
% we can define $\As^\Hf_{P,0} \defeq \bot^\Hf_{\Sigma'}$;
% \begin{align*}
% %  \As^\Hf_{P,0}&\defeq \bot^\Hf_{\Sigma'}\\
%   \As^\Hf_{P,\beta+1}&\defeq T^\Hf_P(\As^\Hf_{P,\beta})&\text{for }\beta\in\On\\
%   \As^\Hf_{P,\gamma}&\defeq \bigsqcup_{\beta<\gamma}\As^\Hf_{P,\beta}&\text{for }\gamma\in\Lim\\
%   \As^\Hf_P&\defeq \bigsqcup_{\beta\in\On}\As^\Hf_{P,\beta}
% \end{align*}
% where $\bot^\Hf_{\Sigma'}$ and $\bigsqcup\Bd$ (for a set $\Bd$ of
% expansions of $\As$) are the $\Sigma'$-expansion of $\As$
% defined by $R^{\bot^\Hf_{\Sigma'}}\defeq\bot^\Hf_\rho$ and $R^{\bigsqcup\Bd}\defeq\bigsqcup\{R^\Bs\mid\Bs\in\Bd\}$, respectively, for relational
% symbols ${R\from\rho}\in\Sigma'\setminus\Sigma$. Note that by definition
% of (pre-)Henkin frames, all $\As^\Hf_{P,\beta}$ are indeed
% $(\Sigma',\Hf)$-structures and expansions of $\As$. If no confusion arises, we abbreviate $\As^\Hf_{P,\beta}$ to $\As^\Hf_\beta$ in the following.

Intuitively, we start from the $\prel$-minimal structure assigning
$\bot_\rho$ to every $R\from\rho\in\Sigma'\setminus\Sigma$ and 
%\changed[dw]{
we incrementally extend the structure
%} %greedily enlarge 
to satisfy more of the program.
$\As_\Prgm$ is a prefixed point of $T_\Prgm$ (\cref{lem:knaster-tarski}). Therefore,
\begin{corollary}
  % \begin{thmlist}
  % \item\label{lem:plateau} There is an ordinal $\gamma$ satisfying $\As^\Hf_P=\As^\Hf_\gamma$.
  % \item
    \label{lem:modp}$\As_\Prgm\models \Prgm$ and $\As_\Prgm\models\{D\in\Set\mid D\text{ definite}\}$.
  % \end{thmlist}
\end{corollary}

Next, suppose there are relations
${\arel_\sigma}\subseteq\sinti\Hf\sigma\times\sinti\Hf\sigma$ (for types
$\sigma$) 
%\changed[dw]{
compatible with $\prel_\sigma$,
%}, 
and
\begin{enumerate}
\item if $\Bs\arel\Bs'$ and $\alpha\arel\alpha'$ then $\sinti\Bs
  M(\alpha)\arel\sinti{\Bs'} M(\alpha')$, and
\item $b\arel b'$ iff $b\leq b'$ for $b,b'\in\sinti\Bs o=\bool$,
\end{enumerate}
where we omit type subscripts and lift $\arel$ in the usual pointwise manner to structures and valuations.
Then $T_\Prgm$ is quasi-monotone. Besides, if $\Bs$ is an expansion
of $\As$ satisfying $\Bs\models \Prgm$ then by \cref{lem:knaster-tarski}, for closed positive existential
formulas $F$, $\sinti{\As_\Prgm}F\leq\sinti\Bs F$. Consequently,
$\As_\Prgm\models\Set$ if $\Bs\models\Set$.

The main obstacle (and where $\prel$ fails) is to ensure that $\arel$
is compatible with applications, i.e.\ if $r\arel_{\tau\to\rho} r'$ and $s\arel_\tau s'$
then $r(s)\arel_\rho r'(s')$. Therefore, we simply \emph{define} it
that way:
\begin{definition}
  \label{def:arelm}
  We define a relation ${\arel_\sigma}\subseteq{\sinti\Hf\sigma\times\sinti\Hf\sigma}$ as follows by recursion on the type $\sigma$:
% \[
% \begin{array}{l}
%   n\arel_{\iota} n' \; \defeq \; n=n' \qquad (n,n'\in\sinti\Hf{\iota})\\
%   b\arel_o b' \; \defeq \; b\leq b' \qquad (b,b'\in\sinti\Hf o)\\
%   r\arel_{\tau\to\sigma} r' \; \defeq \; \forall s, s'\in\sinti
%                                        \Hf\tau\ldotp
%                                        s\arel_\tau
%                                        s'\rightarrow {}\\
%   \qquad \qquad \qquad \qquad r(s)\arel_\sigma r'(s') \qquad (r,r'\in\sinti \Hf{\tau\to\sigma})
% \end{array}
% \]
\[
\begin{array}{lr}
  n\arel_{\iota} n' \; \defeq \; n=n' & (n,n'\in\sinti\Hf{\iota})\\
  b\arel_o b' \; \defeq \; b\leq b' & (b,b'\in\sinti\Hf o)\\
  r\arel_{\tau\to\sigma} r' \; \defeq \; \forall s, s'\in\sinti
                                       \Hf\tau\ldotp
                                       s\arel_\tau
                                       s'\rightarrow {}\hspace*{-5mm}\\
  \qquad \qquad \qquad \qquad r(s)\arel_\sigma r'(s') & (r,r'\in\sinti \Hf{\tau\to\sigma})
\end{array}
\]
% \dw{what do you think of this alignment?}
% \lo{Fine. The intended parse is clear.}
\end{definition}
$\arel$ is transitive but neither reflexive (\cref{ex:qmnotrt}) nor
antisymmetric, in general, and coincides with the pointwise ordering $\prel$ on the monotone
frame $\Mf$ \ifproceedings\cite{OW19}\else(\cref{lem:moncoin})\fi.
\begin{example}
  \begin{thmlist}
  \item\label{ex:qmlog1} For all relational types $\rho$ and
    $s\in\sinti\Hf\rho$, $\bot_\rho\arel s\arel\top_\rho$.
  \item\label{ex:qmlog2} $\orf\arel\orf$, $\andf\arel\andf$ and for
    argument types $\tau$, $\hexists_\tau\arel\hexists_\tau$\ifproceedings\else\footnote{the argument for the
      latter is not entirely trivial but similar as in \cref{ex:qclog2}}\fi.
  % \item\label{ex:qmlog2} Next, suppose $r,r'\in\sinti \Hf{\tau\to o}$ are such that $r\arel r'$ and $\hexistsh_\tau(r)=1$. Hence, there exists $s\in\sinti \Hf\tau$ satisfying $r(s)=1$. If $\tau=\iota$ then $r'(s)=1$ and otherwise $r'(\top_\tau^\Hf)=1$ because $s\arel\top^\Hf_\tau$. Consequently, $\hexists_\tau(r')=1$ and $\hexists$ is quasi-monotone, too.
  \item\label{ex:qmnotrt} Let $\Hf=\Sf$ be the standard frame and let
    $\nega\in\sinti\Sf{(\iota\to o)\to o}$ as in
    \cref{ex:notminmod}. Recall that $\bot_{\iota\to
      o}\arel\top_{\iota\to o}$. However,
    $\nega(\bot_{\iota\to
      o})=1>0=\nega(\top_{\iota\to o})$. This shows that
    $\arel$ is not reflexive, in general.
  \end{thmlist}
\end{example}
\begin{example}
\label{ex:notminmod2}
  For the structures $\Bs_1$ and $\Bs_2$ of \cref{ex:notminmod} it
  holds that $\Bs_1=\As_\Prgm$ and $\Bs_1\arel\Bs_2$ because due to
  $\bot_{\iota\to o}\arel\top_{\iota\to o}$, for any $s\arel
  s'$, $s(\bot_{\iota\to o})\leq s'(\top_{\iota\to o})$ and
  therefore $R^{\Bs_1}(s)\leq R^{\Bs_2}(s')$. In particular, the fact
  that $R^{\Bs_1}(\nega)>R^{\Bs_2}(\nega)$ is not a concern because
  $\nega\arel\nega$ does \emph{not} hold.
\end{example}

A simple induction \ifproceedings\else(cf.\ \cref{cor:relm})\fi on the type $\sigma$ shows that $\arel_\sigma$
%\changed[dw]{
is compatible with $\prel_\sigma$.
%}
Furthermore,
\begin{restatable}{lemma}{termmon}
  \label{lem:termmon}
  Let $\Bs\arel\Bs'$ be expansions of $\As$, $\alpha\arel\alpha'$ be valuations and let $M$ be a positive existential term.
  Then $\sinti\Bs M(\alpha)\arel\sinti{\Bs'} M(\alpha')$.
\end{restatable}
Consequently, the immediate consequence operator is quasi-monotone and
we conclude:
\begin{restatable}% [Properties of $\As^\Hf_P$ (cont.)]
  {theorem}{propap}
\label{thm:canmod} If $\Set$ is $(\As,\Hf)$-satisfiable then $\As_\Prgm\models\Set$.
\end{restatable}

% Again, we frequently abbreviate $\sqsubseteq^\Hf_{m,\rho}$ as $\mrel$ and lift all notions to valuations and structures in the obvious pointwise manner. 
% Similarly as $\crel$, the relation $\mrel$ (we omit $\Hf$ and $\sigma$ from $\sqsubseteq^\Hf_{m,\sigma}$ whenever they are clear from the context) 
% \iffalse
% \lo{Need to define $\mrel$ explicitly in terms of $\sqsubseteq^\Hf_{m,\sigma}$.} 
% \dw{$\mrel$ is just $\sqsubseteq^\Hf_{m,\sigma}$ when $\Hf$ and the sub/superscript have been omitted}
% \lo{Yes -- my point is that we should say this explicitly.}
% \dw{In \cref{sec:prelims}, I've written:

%   ``Throughout the paper, we omit type subscripts or Henkin frame
%   superscripts to reduce clutter if they can be inferred or are clear from the context.''

% The aim was to avoid repeating this boring sort of sentence but I see that a reader might be confused otherwise.}
% \lo{Noted. I think that a brief clarification \emph{in situ} can only help the busy referee (and is therefore in our interest), even if redundant \emph{formally}.}
% \dw{ok, done}
% \fi

% All definitions have been set up in such a way that the denotation of positive existential terms is monotone with respect to $\mrel$:

% Consequently, also the immediate consequence operator is monotone with
% respect to~$\mrel$ and in a similar spirit to \cref{subsec:quasicont} we get:

\section{Resolution Proof System}
\label{sec:resolution}
\begin{figure*}[!t]
  \centering
  {%\small
    \begin{prooftree}
      \AxiomC{$\overbrace{\neg(n\geq 1)\lor
          \tikz[baseline]{\node[resnode] {$\neg\Iter\,\Add\,n\,n\,x$};}
        \lor\neg(x\leq n+n)}^G$}
      \AxiomC{$D_3$}
      \LeftLabel{Resolution}
      \def\extraVskip{1pt}
      \BinaryInfC{$\neg(n\geq 1)\lor
        \tikz[baseline]{\node[unode] {\uwave{$\neg(n>0)$}};}\lor
        \tikz[baseline]{\node[unode] {\uwave{$\neg \Iter\,\Add\,n\,(n-1)\,y$}};} \lor
        \tikz[baseline]{\node[resnode] {\uwave{$\neg\Add\,n\,y\,x$}};}
      \lor\neg(x\leq n+n)$}

      \AxiomC{$D_1$}
      \LeftLabel{Resolution}
      \BinaryInfC{$\neg(n\geq 1)\lor\neg(n>0)\lor
        \tikz[baseline]{\node[resnode]
          {$\neg\Iter\,\Add\,n\,(n-1)\,y$};}
        \lor
        \tikz[baseline]{\node[unode] {\uwave{$\neg (x=n+y)$}};}
        \lor\neg(x\leq n+n)$}
      \AxiomC{$D_2$}
      \LeftLabel{Resolution}
      \BinaryInfC{\vspace*{-100pt}$\neg(n\geq 1)\lor\neg(n>0)\lor
        \tikz[baseline]{\node[unode] {\uwave{$\neg (n-1\leq
              0)$}};}\lor
        \tikz[baseline]{\node[unode] {\uwave{$\neg(n=y)$}};}
        \lor\neg (x=n+y)\lor\neg(x\leq
        n+n)$\vspace*{-100pt}}
      \LeftLabel{Constraint refutation}
      \UnaryInfC{$\bot$}
    \end{prooftree}
    }
    \caption{Refutation of the set of HoCHCs from
      \cref{ex:HoCHC}. Atoms involved in resolution steps are shaded;
      atoms that are added are wavy-underlined.}
    \label{fig:Sref}
  \end{figure*}

%\changed[dw]{
Our resolution proof system is remarkably simple,
consisting of only three rules:
  \begin{enumerate*}
  \item a higher-order version of the usual resolution rule \cite{R65}
    between a pair of goal and definite clauses (thus yielding a goal clause),
  \item a rule for $\beta$-reductions on leftmost (outermost) positions
    of atoms in goal clauses and
  \item a rule to refute certain goal clauses which are not satisfied
    by the model of the background theory (similar to \cite{BGW94}).
  \end{enumerate*}
%}

\medskip
\inferbinlics{Resolution}{\neg R\,\overline M\lor G}{G'\lor R\,\overline x}{G\lor \big(G'[\overline M/\overline x]\big)}
\iffalse
\lo{Is there an (implicit?) assumption: $R \not= R' \implies {\overline x}_R \cap {\overline x}_{R'} = \emptyset$?}\dw{no}
\fi

\inferunlics{$\beta$-Reduction}{\neg (\lambda x\ldotp L)M\,\overline N\lor G}{\neg L[M/x]\, \overline N\lor G}{}

\inferunlics{Constraint
  refutation}{G\lor\neg\phi_1\lor\cdots\lor\neg\phi_n}{\bot}{provided
  that each atom in $G$ has the form\footnote{where $x$ is a variable}
  $x\,\overline M$, each $\phi_i$ is a background atom and there exists a valuation $\alpha$ such that $\As,\alpha\models\phi_1\land\cdots\land\phi_n$.}

\smallskip
% Notice that resolution is always between a (compatible) pair of goal and definite clauses,
% and the rules $\beta$-Reduction and Constraint Refutation are only applicable to goal clauses. Besides, the conclusion of each rule is a goal clause.
% %Clearly every (pre-)term obtained by any of the rules is a goal clause.
\iffalse
Though undecidable in general for higher-order logic
\cite{L72,H73,G81}, unification is trivial when restricted to
HoCHC. \dw{is this very meaningful?}
\lo{I would omit this sentence.
There is no unification in the three rules, only substitution.}
\dw{I agree}
\fi
    %     \lo{There is actually no proper unification in the three rules, only substitution.} \dw{yes, $[\overline M/\overline x]$ is a (trivial) unifier of $R\,\overline M$ and $R\,\overline x$}

%To get some intuitive understanding of the proof system we first consider the following example.
\begin{example}[Refutation proof]
%Consider again the set of HoCHCs from \cref{ex:HoCHC}.
  A refutation of the set of HoCHCs from \cref{ex:HoCHC} is given in \cref{fig:Sref}.
  %\dw{What do you think about leaving the atoms which are involved in resolution steps just bold and additionally underline the \emph{added} atoms. This might improve readability if it is just printed black and white.}
  %\lo{Sorry - I don't much like the look of bold (the variables become Roman). Let me know what you think of boxing and underlining.}
  %\dw{That's fine with me.}
  % \dw{I think shading and underlining wavyly is easier to distinguish from the inference lines than boxing and underlining.}
  % \lo{This is perfectly fine by me.}
  The last inference is admissible because for any valuation satisfying $\alpha(n)=\alpha(y)=1$ and $\alpha(x)=2$,
  \begin{align*}
    \As_{\LIA},\alpha\models\,&(n\geq 1)\land(n>0)\land (n-1\leq 0)\,\land\\&(n=y)\land (x=n+y)\land(x\leq n+n).
  \end{align*}
\end{example}

%\changed[dw]{
Since variables are implicitly universally quantified, the rules have to
be applied modulo the renaming of (free) variables; we
write $\Set'\Res\Set'\cup\{G\}$ if $G$ can be thus derived from the clauses
in $\Set'$ using
the above rules and $\Res^*$ for the reflexive, transitive closure of $\Res$.
%}

% Next, recall that background atoms only contain variables of
% type $\iota$; and $\Bs\models S$ implies
% $\Bs,\alpha\models G$ for all goal clauses $G\in S$ and valuations $\alpha$ satisfying
% $\alpha(x)=\top^\Hf_\rho$ for $x\from\rho\in\Delta$. Consequently, we get:
% \dw{foo}
% Now, soundness of the proof system is %rather
%  straightforward:

\begin{restatable}[Soundness]{proposition}{soundness}
  \label{prop:soundness}
  Let $\Set$ be a set of HoCHCs.
  
  If $\Set\Res^* \Set'\cup\{\bot\}$ (for some $\Set'$) then $\Set$ is
  $(\As,\Hf)$-un\-satisfiable,
  %\changed[dw]{
  and this holds even if $\Hf$ is not complete.
  %}
\end{restatable}
\begin{proof}[Proof sketch]
  %\changed[dw]{
  The most interesting case occurs when the constraint
    refutation rule is applied to $G\defeq\biglor_{i=1}^m\neg
  x_i\,\overline M_i\lor\biglor_{j=1}^n\neg\phi_j$. Being of
  relational type, each variable $x_i$ cannot
  occur in any $\phi_j$. Thus, modifying witnesses $\alpha$ of
  $\As,\alpha\models \phi_1\land\cdots\land\phi_n$ to satisfy
  $\alpha'(x)=\top_\rho$ for $x\from\rho\in\Delta$, we conclude
  $\Bs,\alpha'\not\models G$ for
all expansions $\Bs$ of $\As$.
%}
\end{proof}
%\changed[dw]{
Observe that the argument makes use of
$\top_\rho\in\sinti\Hf\rho$, which is a consequence of the
comprehension axiom.
%}

The following completeness theorem is significantly more difficult. 
In fact, we will not prove it until \cref{sec:completeness}.
\begin{restatable}[Completeness]{theorem}{completeness}
  \label{thm:completeness}
  If $\Set$ is $(\As,\Hf)$-unsatisfiable then $\Set\Res^*\{\bot\}\cup\Set'$ for some $\Set'$.
\end{restatable}
% The following observation is a trivial consequence of the soundness and refutational completeness of the proof system:
% \begin{corollary}
%   \label{cor:soundcomplete}
%   Let $S$ be a set of HoCHCs. 

%   $S$ is $(\As,\Hf)$-unsatisfiable if and only if there exists $S'$ such that $S\Res^*\{\bot\}\cup S'$.
% \end{corollary}

Consequently, the resolution proof system gives rise to a
semi-decision procedure for the $(\As,\Hf)$-unsatisfiability problem
provided it is (semi-)decidable whether
%\changed[dw]{
a goal clause of background atoms is not satisfied by the
  background theory\footnote{i.e.\ whether there exists a valuation $\alpha$ such that $\As,\alpha\models\phi_1\land\cdots\land\phi_n$}.
%}
% conjunctions of background atoms are satisfiable in the background theory

\paragraph*{Outline of the Completeness Proof}
%\changed[dw]{\sout{Next, we give a brief outline of the completeness proof.}}% , which occupies the remainder of this section:
\begin{enumerate}[label=\textbf{(S\arabic*)},leftmargin=2\parindent]
% \item\label{it:outlinecomp1} First, we % show how to
%   iteratively
%   construct a canonical model\footnote{in fact the ``least'' model (in a non-standard sense) of the definite clauses (\cref{sec:quasi-mon})} of the definite clauses in $S$ (\cref{sec:canmod}).
\item\label{it:outlinecomp2} First, we prove that some goal
  clause is not satisfied by the canonical structure already after a
  \emph{finite} number of iterations if $\Set$ is $(\As,\Hf)$-unsatisfiable (\cref{subsec:quasicont}).
\item\label{it:outlinecomp3} Consequently, there is a \emph{syntactic} reason
  for $\Set$'s $(\As,\Hf)$-unsatisfiability (by ``unfolding definitions'') (\cref{sec:syntactic}).
\item\label{it:outlinecomp4} Finally, we prove that the ``unfolding''
  actually only needs to take place at the leftmost (outermost)
  positions of atoms (\cref{sec:leftmost}), which can be captured by the resolution proof system (\cref{sec:completeness}).
\end{enumerate}
Observe that \sref{it:outlinecomp2} is model theoretic / semantic,
whilst \srefs{it:outlinecomp3}{it:outlinecomp4} are proof theoretic /
syntactic.

\subsection{Quasi-Continuity}
\label{subsec:quasicont}
\label{SUBSEC:QUASICONT}

Whilst in \cref{sec:quasi-mon} we have shown that $\As_\Prgm$ is a
model of the definite clauses, we now examine the consequences of
$\As_\Prgm\not\models G$ for some goal clause $G\in\Set$. 
Unlike the 1st-order case \cite{BGMR15}, stage $\omega$ is not a fixed point of $T_\Prgm$ in general, as the following example illustrates:
\begin{example}
  \label{ex:notstopomega}
  Consider the following program:
  % \begin{align*}
  %   \neg (x_R=0\lor R\,(x_R-1))&\lor R\,x_R\\
  %   \neg x_U\,R&\lor U\,x_U
  % \end{align*}
  \[
    \neg (x_R=0\lor R\,(x_R-1)) \lor R\,x_R \qquad \quad
    \neg {(x_U\,R)} \, \lor \, U\,x_U
  \]
  where $\Sigma'=\Sigma_{\LIA}\cup\{R\from\iota\to o,U\from((\iota\to
  o)\to o)\to o\}$, $\Delta(x_R)=\iota$ and $\Delta(x_U)=(\iota\to o)\to o$. Let $\As$ be the standard model of Linear Integer Arithmetic
  $\As_{\LIA}$ and let $\Hf=\Sf$ be the standard frame. For ease of
  notation, we introduce functions
  $r_\alpha\from\sinti\Sf\iota\to\bool$ such that
  $r_\alpha(n)=1$ iff $0\leq n<\alpha$, and
  $\delta_\alpha\from\sinti\Sf{\iota\to o}\to\bool$ such that
  $\delta_\alpha(r)=1$ iff $r=r_\alpha$, where $\alpha\in\omega\cup\{\omega\}$.
  Then it holds $R^{\As_n}=r_n$, $U^{\As_0}=\bot_{(\iota\to o)\to o}$ and $U^{\As_n}(s)=s(r_{n-1})$ for $n>0$. Therefore $R^{\As_\omega}=r_\omega$ and $U^{\As_\omega}(s)=1$ iff there exists $n<\omega$ satisfying $s(r_n)=1$.
  % \begin{align*}
  %  % R^{\As^\Sf_\omega}&=r_\omega\\
  %   U^{\As^\Sf_\omega}(s)&=
  %                   \begin{cases}
  %                     1&\text{if there exists $n<\omega$ satisfying }s(r_n)=1\\
  %                     0&\text{otherwise}
  %                   \end{cases}
  % \end{align*}
  In particular, $U^{\As_\omega}(\delta_\omega)=0$. On the other
  hand, \changed[dw]{$U^{\As_{\omega+1}}(\delta_\omega)=\sinti{\As_\omega}{\lambda x_U\ldotp
      x_U\,R}(\delta_\omega)=1$.}
  % $\As_\omega,\top_\Delta[x_U\mapsto
  % \delta_\omega]\models x_U\,R$ holds and therefore
  % $U^{\As_{\omega+1}}(\delta_\omega)=1$.
  Consequently, $\As_\omega\neq\As_{\omega+1}$.
  %\lo{This could be relegated to Appendix if necessary.}
\end{example}
Nonetheless, there still exists a (finite) $n\in\omega$ satisfying
$\As_n\not\models G$ if $\As_\Prgm\not\models G$ (\cref{thm:contmain}). We make use
of a similar strategy to establishing the canonical model property: we
introduce the notion of \emph{quasi-continuity}, state a version of Kleene's
% \dw{is it ok to call \cref{lem:kleene} that way?} 
% \lo{Because the continuity is only quasi, I suppose one could call it either a version of Kleene or of Knaster-Tarski.}
fixed point theorem and prove the immediate consequence operator
to be quasi-continuous.

\begin{definition}
  Let ${\arel}\subseteq L\times L$ be a relation.
  \begin{thmlist}
  \item $\arel$ is \emph{$\arel$-directed} if for every $a,b\in L$,
    \changed[dw]{$a\arel a$ and}
    % $a=b$ or
    there exists $c\in L$ satisfying $a,b\arel c$.

    For $a\in L$ we write $\dir_\arel(a)$ for the set of $\arel$-directed subsets $D$ of $L$ satisfying $a\arel\biglor D$.
  \item $F$ is \emph{quasi-continuous} if for all $a\in L$ and $D\in\dir_\arel(a)$, $F(a)\arel\biglor_{b\in D}F(b)$.
  \end{thmlist}
\end{definition}
%\changed[dw]{
  Thus, every quasi-continuous function is in particular
  quasi-monotone if $\arel$ is reflexive.
%\sout{Thus, every singleton set $\{a\}\subseteq L$ is $\arel$-directed even if $a\arel a$ does not hold and every quasi-continuous function is in particular quasi-monotone.}
%}
\begin{restatable}{proposition}{kleene}
  \label{lem:kleene}
  If $\arel$ 
  is compatible with $\leq$
  and $F$ is quasi-continuous then
  \begin{enumerate*}
  \item for all ordinals $\beta\leq\beta'$, $a_\beta\arel
    a_{\beta'}$ and
  \item
    $a_F\arel a_\omega$.
  \end{enumerate*}
\end{restatable}
Combined with \cref{lem:knaster-tarski} this yields
Kleene's fixed point theorem in the case of ${\arel}\defeq{\leq}$.

Similarly as in \cref{sec:quasi-mon}, we need relations
${\arel_\sigma}\subseteq\sinti\Hf\sigma\times\sinti\Hf\sigma$ which
behave well
% are compatible
with applications in order for the immediate consequence
operator to be quasi-continuous. Therefore, we stipulate 
%\changed[dw]{
(overloading the
notation of \cref{sec:quasi-mon}):
%}:
\begin{definition}
  \label{def:arelc}
We define ${\arel_\sigma} \subseteq {\sinti\Hf\sigma\times\sinti\Hf\sigma}$ by recursion on the type $\sigma$: 
% \[
% \begin{array}{l}
%   b\arel_o b'
%   \; \defeq \; b\leq b' \qquad (b,b'\in\sinti \Hf o)\\
%   n\arel_\iota n'
%   \; \defeq \; n=n' \qquad (n,n'\in\sinti\Hf\iota)\\
%   r\arel_{\tau\to\sigma} r'
%   \; \defeq \; \forall s\in\sinti \Hf\tau,\Sd'\in\dir_{\arel_\tau}(s)\ldotp\\
%   \qquad\qquad \quad r(s)\arel_\sigma\bigsqcup_{s'\in\Sd'}r'(s') \qquad (r,r'\in\sinti \Hf{\tau\to\sigma})
% \end{array}
% \]
\[
\begin{array}{lr}
  b\arel_o b'
  \; \defeq \; b\leq b' & (b,b'\in\sinti \Hf o)\\
  n\arel_\iota n'
  \; \defeq \; n=n' & (n,n'\in\sinti\Hf\iota)\\
  r\arel_{\tau\to\sigma} r'
  \; \defeq \; \forall s\in\sinti \Hf\tau,\Sd'\in\dir_{\arel_\tau}(s)\ldotp\hspace*{-5mm}\\
  \qquad\qquad\qquad r(s)\arel_\sigma\bigsqcup_{s'\in\Sd'}r'(s') & (r,r'\in\sinti \Hf{\tau\to\sigma})
\end{array}
\]
% \dw{different alignment, too} \lo{Fine.}.
\end{definition}
There is an elementary inductive argument
\ifproceedings\else(cf.~\cref{cor:relc}) \fi
that each $\arel_\sigma$ 
%\changed[dw]{
is compatible with $\prel_\sigma$.
%} 
We lift $\arel$ to structures
and valuations in a pointwise way, and abbreviate $\dir_\arel$ as $\dir$.

\begin{restatable}{lemma}{contmain}
  \label{lem:contmain}
  Let $M$ be a positive existential term, $\Bs$ be an expansion of $\As$, $\Bd'\in\dir(\Bs)$, $\alpha$ be a valuation and let $\vals'\in\dir(\alpha)$.
  Then\footnote{By \cref{rem:deneq} the right-hand side is
  well-defined.}
\ifproceedings
\begin{align*}
    \label{eq:ncontmain}
    \sinti\Bs M(\alpha)\arel\bigsqcup_{\Bs'\in\Bd',\alpha'\in\vals'}\sinti{\Bs'} M(\alpha').
\end{align*}
\else
\begin{align}
    \label{eq:ncontmain}
    \sinti\Bs M(\alpha)\arel\bigsqcup_{\Bs'\in\Bd',\alpha'\in\vals'}\sinti{\Bs'} M(\alpha').
\end{align}
\fi
\end{restatable}
Consequently, the immediate consequence operator is
quasi-continuous. Moreover, by \cref{lem:kleene}, for closed positive existential formulas
$F$,
$\sinti{\As_\Prgm}F\leq\max_{n\in\omega}\sinti{\As_n}F$.
Therefore, we get the following result, which is key for the
refutational completeness of the proof system.
\begin{theorem}
  \label{thm:contmain}
  Let $G$ be a goal clause. 
  If $\As_\Prgm\not\models G$ then there exists $n\in\omega$ such that $\As_n\not\models G$.
\end{theorem}

\subsection{Syntactic Unfolding}
\label{sec:syntactic}
\label{SEC:SYNTACTIC}
%\changed[dw]{
  Having established \sref{it:outlinecomp2}, we study a functional
  relation $\pred$ on positive existential terms, which is a syntactic
  counterpart of the immediate consequence
  operator. Essentially\footnote{For a formal definition refer to \ifproceedings\cite{OW19}\else\cref{fig:parallel} in
    \cref{sec:appsyntactic}\fi.},
  % \lo{Rework footnote - proceedings
% version has not appendix.}
% \dw{there is a proceedings flag. is it okay in this style?}
% \lo{Yes. Thanks.}
  it
  holds $M\pred N$ if $N$ is obtained from $M$ by replacing all
  occurrences of symbols $R\in\Sigma'\setminus\Sigma$ with
  $\lambda\overline x_R\ldotp F_R$, which is reminiscent of the
  definition of $R^{T_\Prgm(\Bs)}$. Therefore:
%}

% Having established \sref{it:outlinecomp2}, 
% we need to devise a \emph{syntactic} method to capture that $\As^\Hf_n\not\models G$ does indeed hold (\sref{it:outlinecomp3}).

% To that end we consider a functional relation $\pred$ on positive existential terms.
% The idea is that $M \pred N$ if $N$ is obtained from $M$ by replacing all occurrences of symbols $R\in\Sigma'\setminus\Sigma$ by $\lambda\overline x_R\ldotp F_R$, which is reminiscent of the definition of the immediate consequence operator (specifically of $R^{T_P^\Hf(\Bs)}$).
% A similar idea is exploited in \cite{CHRW13}. 
% Formally, this relation is defined in \cref{fig:parallel} in \cref{sec:appsyntactic}.
% Writing $M\pred^n N$ to mean $\underbrace{M\pred\cdots\pred N}_{n}$,
% the intuition that $\pred$ allows us to capture the effect of the immediate consequence operator syntactically is made precise by the following:
\begin{restatable}{proposition}{parallelcorr}
  \label{lem:parallelcorr}
  Let $\Bs$ be an expansion of $\As$ and let $M$ and $N$ be positive existential terms satisfying $M\pred N$. Then for all valuations $\alpha$, $\sinti{T_\Prgm(\Bs)} M(\alpha)=\sinti{\Bs}N(\alpha)$.
\end{restatable}
A similar idea is exploited in \cite{CHRW13}.

% Nonetheless, there is a mismatch between $\pred$ and the resolution proof system in three respects:
% \begin{enumerate}[label=\textbf{(M\arabic*)},leftmargin=2.5\parindent]
% \item\label{it:mismatch1} $\pred$ replaces potentially many relational symbols by their ``definitions'' whilst the resolution proof systems just considers one at a time.
% \item\label{it:mismatch2} $\pred$ also acts ``inside'' atoms whilst the resolution proof system only takes relational symbols at the ``leftmost'' position of atoms into account.
% \item\label{it:mismatch3} $\pred$ is a relation on positive existential terms whilst the resolution proof system operates on clauses.
% \end{enumerate}

Next, let $\upsilon\defeq\{(R,\lambda\overline x_R\ldotp F_R)\mid R\in\Sigma'\setminus\Sigma\}$ and $\beta\upsilon\defeq\beta\cup\upsilon$. 
Besides, let $\bured$ be the compatible closure \cite[p.~51]{B12} of $\beta\upsilon$. 
It is easy to see that ${\pred}\subseteq{\buredrt}$, where $\buredrt$ is the reflexive, transitive closure of $\bured$.

\subsection{Leftmost (Outermost) Reduction}
\label{sec:leftmost}
\label{SEC:LEFTMOST}
\begin{figure*}[!t]
  \subfloat[Definition of leftmost (outermost) reductions.]{% <-- don't forget
    \label{fig:lred}
    $\begin{aligned}
        \begin{tabularx}{\textwidth}{XcXcXcX}
    &\infer[R\in\Sigma'\setminus\Sigma]{R\,\overline M\lred 1(\lambda\overline x_R\ldotp F_R)\overline M}{}&&\infer{(\lambda x\ldotp L)M\,\overline N\lred 1 L[M/x]\overline N}{}&
   &\infer[\circ\in\{\land,\lor\}]{M_1\circ M_2\lred{m_1+m_2}
      N_1\circ N_2}{M_1\lred{m_1} N_1&
      &M_2\lred{m_2} N_2}&\\[5pt]
    &\infer{\exists x\ldotp M\lred m\exists x\ldotp N}{M\lred m N}&
  &\infer{M\lred 0 M}{}&
    &\infer{L\lred{m_1+m_2} N}{L\lred{m_1} M&M\lred{m_2} N}&
    \end{tabularx}
  \end{aligned}$
}

\subfloat[Definition of standard reductions (by $\overline M\sred\overline N$ we mean $M_j\sred N_j$ for each $1\leq j\leq n$,
assuming $\overline M$ is $M_1,\ldots,M_n$ and $\overline
N$ is $N_1,\ldots,N_n$).]{%
  \label{fig:std}
  $\begin{aligned}
  \begin{tabularx}{\textwidth}{XcXcXcX}
    &\infer[L\lredrt c\,\overline M,
    c\in\Sigma'\cup\{\land,\lor,\exists_\tau\}]{L\sred c\,\overline
      N}{\overline M\sred\overline N}&&\infer[L\lredrt x\,\overline
    M]{L\sred x\,\overline N}{\overline M\sred\overline N}&
    &\infer[L\lredrt (\lambda x\ldotp M')\overline M]{L\sred (\lambda x\ldotp N')\overline N}{M'\sred N'&\overline M\sred\overline N}.&
  \end{tabularx}
\end{aligned}$
}
\caption{Leftmost outermost and standard reductions.}
\label{fig:los}
\end{figure*}

% \begin{figure*}[!t]
% \centering
% \begin{align*}
%   \begin{tabularx}{\textwidth}{XcXcXcX}
%     &\infer[L\lredrt c\,\overline M,
%     c\in\Sigma'\cup\{\land,\lor,\exists_\tau\}]{L\sred c\,\overline
%       N}{\overline M\sred\overline N}&&\infer[L\lredrt x\,\overline
%     M]{L\sred x\,\overline N}{\overline M\sred\overline N}&
%     &\infer[L\lredrt (\lambda x\ldotp M')\overline M]{L\sred (\lambda x\ldotp N')\overline N}{M'\sred N'&\overline M\sred\overline N}.&
%   \end{tabularx}
% \end{align*}
% \caption{%foobar2.
% Definition of the relation $\sred$. By $\overline L\lredrt\overline M$ we mean $L_j\lredrt M_j$ for each $1\leq j\leq n$ assuming $\overline L=(L_1,\ldots,L_n)$ and $\overline M=(M_1,\ldots,M_n)$; similarly for $\overline M\sred\overline N$.}
% \label{fig:std}
% \end{figure*}

%\changed[dw]{
There is an important mismatch between the relation $\bured$
  and the rules of the proof system: in contrast to the former, the
  latter only take leftmost (outermost) positions of atoms into
  account. Fortunately, arbitrary sequences of $\beta\upsilon$-reductions can be
mimicked by sequences which are standard in the sense that purely
leftmost reductions are followed by purely non-leftmost
ones (\cref{cor:sthi}).

\cref{fig:los} defines $\lred\cdot$ and $\sred$, which formalise
leftmost (outermost) and
standard reductions, respectively. We write $M\lredrt N$ if $M\lred{m}
N$ for some $m$, where $m$ corresponds to the number of leftmost
$\beta\upsilon$-reductions having been performed.
The idea is that $L\sred N$ if for some $M$, $L\lredrt M$ and we can
obtain $N$ from $M$ by performing standard $\beta\upsilon$-reductions only on
non-leftmost positions.
%}

% Next, the notion of leftmost reduction ($\lred{\cdot}$) is made precise by the relation
% defined inductively in \cref{fig:lred}
%  and we write $M\lredrt N$ if $M\lred{m} N$ for some $m$. Intuitively,
%  $M\lred{m}N$ holds if $m$-many leftmost $\bured$-reductions have been performed.

%  % It is easy to see that ${\pred}\not\subseteq{\lredrt}$, i.e.\ we also have to do $\bured$-reductions at
%  % ``non-leftmost positions'' to mimic the effect of $\pred$. 
%  Therefore, we define the relation $\sred$ in \cref{fig:std}.
% The idea is that $L\sred N$ if for some $M$, $L\lredrt M$ and we can obtain $N$ from $M$ by performing $\bured$-reductions only on ``non-leftmost positions''.
% % The proof of the following basic properties of $\sred$ can be found in \cref{app:leftmost}.
% % \begin{restatable}[Basic Properties of $\sred$]{lemma}{propssred}
% %   \begin{thmlist}
% %   \item\label{lem:isrefl} $\sred$ is reflexive (on positive existential terms).
% %   \item\label{lem:ssubsetbu} $\sred\;\subseteq\;\buredrt$.
% %   \item\label{lem:ilsr} If $L\sred N$ and $\overline O\sred\overline Q$ then $L\overline O\sred N\,\overline Q$.
% %   \item\label{lem:hs} If $K\lredrt L\sred N$ then $K\sred N$.
% %   \item\label{lem:srepl} If $L\sred N$ and $O\sred Q$ then $L[O/z]\sred N[Q/z]$.
% %   \end{thmlist}
% % \end{restatable}

% It turns out that if $M$ $\beta\upsilon$-reduces to $N$ then there
% exists a $\beta\upsilon$-reduction sequence that first only applies to
% redexes in ``leftmost positions'' and then only in ``non-leftmost
% positions''.

\begin{restatable}{lemma}{elimbum}
  \label{lem:elimbum}
  If $K\sred M\bured N$ then $K\sred N$.
\end{restatable}
The proof of this proposition is very similar to the proof of the standardisation theorem in the
$\lambda$-calculus as presented in \cite{K00}
and relies on the
insight that if all of $\overline K\sred\overline M$, $K'\sred M'$ and
$O\sred Q$ hold then $K'[O/x]\,\overline K\sred M'[Q/x]\,\overline M$.

% Consequently, by \cref{lem:elimbum} and the fact that $\sred$ is reflexive we obtain
\begin{corollary}
  \label{cor:sthi}
  Let $M$ and $N$ be positive existential terms such that $M\buredrt N$.
  Then $M\sred N$.
\end{corollary}

Next, we consider the relation $\force$ on positive existential
formulas and valuations inductively defined by:
\begin{gather*}
  \infer{\alpha\force x\,\overline M}{}\qquad\infer{\alpha\force\phi}{\As,\alpha\models\phi}\qquad \infer{\alpha\force\exists x\ldotp
    M}{r\in\sinti \Hf{\Delta(x)}&\alpha[x\mapsto r]\force M}\\[5pt]
  \infer{\alpha\force M_1\lor M_2}{i\in\{1,2\}&\alpha\force M_i}\qquad\infer{\alpha\force M_1\land M_2}{\alpha\force M_1&\alpha\force M_2}
\end{gather*}
Intuitively,
$\alpha\force F$ if for some $\alpha'$ (agreeing with $\alpha$ on $\Delta^{-1}(\iota)$), $\As_0,\alpha'\models F$ and there are no
$\lambda$-abstractions in relevant leftmost positions.
\begin{remark}
  \label{rem:refute}
  If $G$ is a goal clause and $\alpha\force\posex(G)$ (for some $\alpha$) then $G$ has the
  form $\biglor_{i=1}^m\neg x_i\,\overline
  M_i\lor\biglor_{j=1}^n\neg\phi_j$ and $G$ can be refuted by the 
  constraint refutation rule in one step.
\end{remark}

\begin{restatable}{lemma}{sdecompose}
  \label{lem:sdecompose}
  Let $G$ be a goal clause, $F$ be a $\beta$-normal positive existential formula and $\alpha$ be a valuation such that $\As_0,\alpha\models F$ and $\posex(G)\sred F$.
  Then there exists a positive existential formula $F'$ satisfying $\posex(G)\lredrt F'$ and $\alpha\force F'$.
\end{restatable}

\subsection{Concluding Refutational Completeness}
\label{sec:completeness}
Finally, we establish a connection between the (abstract)
relation $\lredrt$ on positive existential terms and the resolution
proof system on clauses. We define
a function $\mu$ assigning natural numbers or $\omega$ to positive
existential formulas $E$ by
\begin{align*}
  \mu(E)&\defeq\min\left(\{\omega\}\cup\{m\mid E\lred{m} F\text{ and }\alpha\force F\text{
          for some }\alpha\}\right)
          %     \\
  % \mu(S')&\defeq\min(\{\omega\}\cup
  %          \{m\mid \;&\text{exists }G\in S'\text{, $F$ and
  %                      $\alpha$ s.t.}\\
  % &\posex(G)\lred m F \text{ and } \alpha\force F\}).
\end{align*}
which is extended to  non-empty sets $\Set'$
of HoCHCs by $\mu(\Set')\defeq\min\{\mu(\posex(G))\mid G\in\Set\}$. 

We can use the resolution proof system to derive a set of HoCHCs $\Set''$ with a strictly smaller measure by simulating a $\lred{1}$-reduction step:
\begin{restatable}{proposition}{mudec}
  \label{lem:mudec}
  Let $\Set'\supseteq\Set$ be a set of HoCHCs satisfying $0<\mu(\Set')<\omega$. 
  Then there exists $\Set''\supseteq\Set$ satisfying $\Set'\Res\Set''$ and $\mu(\Set'')<\mu(\Set')$.
\end{restatable}
\begin{example}
  Consider the HoCHCs $\Set=\{\neg (x_R\geq 5)\lor R\,x_R,\neg
  R\,(x_R-5)\lor R\,x_R,\neg R\,5\}$. It holds
  that $R\,5\lred{1}(\lambda
  x_R\ldotp x_R\geq 5\lor R\,(x_R-5))\,5\lred{1}5\geq 5\lor
  R\,(5-5)$ and $\mu(\Set)=2$. Furthermore, $\Set\Res\Set\cup\{\neg 5\geq 5\}$
  and  $\mu(\Set\cup\{\neg 5\geq 5\})=0$.
\end{example}
%\changed[dw]{
Combining everything, we finally obtain:
%}
% Moreover,
% a set of HoCHCs $\Set$ can be refuted in one step by the constraint
% refutation rule if it contains a goal clause $G$ such that
% $\alpha\vartriangleright\posex(G)$.
% Therefore, we finally get:
\completeness*
\begin{proof}
  By \cref{lem:modp}, $\As_\Prgm\models D$ for all definite clauses
  $D\in\Set$. Since $\Set$ is $(\As,\Hf)$-unsatisfiable there exists a goal
  clause $G\in\Set$ satisfying $\As_\Prgm\not\models G$. By
  \cref{thm:contmain} there exists $n\in\omega$ such that
  $\As_n\not\models G$. Let $F_n$ be such that $\posex(G)\pred^n
  F_n$ \changed[dw]{(where $\pred^n$ is the $n$-fold composition of
    $\pred$)}. By \cref{lem:parallelcorr}, $\As_0,\alpha\models F_n$
  (for any $\alpha$ as $F_n$ is closed). Let $F_n'$ be the $\beta$-normal form of $F_n$. By \cref{cor:sthi,lem:sdecompose} there exists $F'$ such that $\posex(G)\lredrt F'$ and $\alpha\force F'$. Consequently, $\mu(\Set)<\omega$. By \cref{lem:mudec} there exists $\Set'\supseteq\Set$ satisfying $\Set\Res^* \Set'$ and $\mu(\Set')=0$. 

  Hence, there exists $G\in\Set'$ and $\alpha$ such that $\posex(G)\lred{0} F'$ and
  $\alpha\force F'$. Clearly, this implies $F'=\posex(G)$, and by \cref{rem:refute}, $\Set\Res^* \Set'\Res\{\bot\}\cup\Set'$.
\end{proof}

% \begin{remark}
% \changed[dw]{\sout{The resolution proof system and its refutational completeness proof are generalised in} \cref{sec:compact-theories} \sout{to arbitrary \emph{compact} background theories, which may have more than one model.}}
% \end{remark}

\subsection{Compactness of HoCHC}
\label{sec:compact-HoCHC}
%The reason why we took $S$ to be finite in Assumption~\ref{ass:setting} is that otherwise the corresponding program would need to contain infinitary disjunctions, which are not part of the syntax as we defined it. 
%If we simply add them (but HoCHCs remain finitary) we can carry out exactly the same reasoning to derive that also every \emph{infinite}, $(\As,\Hf)$-unsatisfiable set of HoCHCs can be refuted in the proof system. 
The reason why we restrict $\Set$ to be finite is to achieve correspondence with programs (\cref{def:program}), which are finite expressions. 
If we simply extend programs with infinitary disjunctions (but keep HoCHCs finitary) we can carry out exactly the same reasoning to derive that also every \emph{infinite}, $(\As,\Hf)$-unsatisfiable set of HoCHCs can be refuted in the proof system.
Consequently: 

\begin{theorem}[Compactness]
\label{thm:compactness}
%\emph{HoCHC is compact}:
For every $(\As,\Hf)$-unsatisfiable set $\Set$ of HoCHCs there exists a finite subset $\Set'\subseteq\Set$ which is $(\As,\Hf)$-unsatisfiable.
\end{theorem}
% \dw{does this need a proper formal proof, i.e.\ duplicating everything or is the above (informal) argument sufficient?}
% \lo{No - provided we have checked that the same reasoning works. I will do that later.}
% \lo{Worth highlighting it as a theorem.}
% \dw{check again}
% \lo{This is nice.
% It is well-known that the compactness theorem does not hold in second (and higher) order logic.}

\section{Semantic Invariance}
\label{sec:equiv}
\label{SEC:EQUIV}
%\changed[dw]{
\cite{BOR18} details an explicit translation between standard and
monotone models of HoCHCs, thus yielding the equivalence of
$\As$-standard- and $\As$-monotone-satisfiability.

As a consequence of
the Completeness \cref{thm:completeness} for the proof system, $(\As,\Hf)$-unsatisfiability for
\emph{any} complete frame $\Hf$ implies the existence of a refutation, which in turn
entails $(\As,\Hf')$-unsatisfiability for \emph{any} frame $\Hf'$ by
the Soundness \cref{prop:soundness}.

Therefore, exploiting \cref{ex:Henkin}, we obtain an
equivalence result encompassing a much wider class of semantics:
%}
\begin{theorem}[Semantic Invariance]
  \label{thm:equivsem}
  Let $\Set$ be a set of HoCHCs. Then the following are equivalent:
  \begin{thmlist}
  \item $\Set$ is $\As$-standard-satisfiable,
  \item $\Set$ is $\As$-Henkin-satisfiable,
  \item $\Set$ is $\As$-monotone-satisfiable,
  \item $\Set$ is $\As$-continuous-satisfiable,
  \item $\Set$ is $(\As,\Hf)$-satisfiable, where $\Hf$ is a
    %\changed[dw]{
    complete %} 
    frame
    expanding $\As$.
  \end{thmlist}
\end{theorem}

% \begin{proposition}
%   \label{thm:equivframes}
%   Let $S$ be a set of HoCHCs and let $\Hf$ and $\Hf'$ be Henkin frames
%   expanding $\As$
%   such that $S$ is $(\As,\Hf)$-unsatisfiable. Then $S$ is also $(\As,\Hf')$-unsatisfiable.
% \end{proposition}
% \begin{proof}
%   By the Completeness \cref{thm:completeness}, $S\Res^*\{\bot\}\cup S'$ for some $S'$ and by soundness of resolution (\cref{prop:soundness}), $S$ is $(\As,\Hf')$-unsatisfiable.
% \end{proof}

% Finally, using \cref{ex:Henkin}, we conclude:

Thus, we call a set of $\Set$ of HoCHCs \emph{$\As$-satisfiable} if it satisfies any of the equivalent conditions of \cref{thm:equivsem}.
\section{Compact Theories}
\label{sec:compact-theories}

% \lo{Here we introduce compactness as a property of a set of structures that supports the standard notion of compactness.
% Theories do not feature in the section directly.
% The title can seem a little odd.}

% \lo{Two thoughts. (1) Because of the page limit, we could move this section to the appendix, and insert a remark (perhaps in \cref{sec:compact-HoCHC}) mentioning the main result of this section.
% (2) In the long version, I suggest moving \cref{sec:compact-HoCHC}
      %       to this section.}
%\changed[dw]{
  In this section, we extend our results to background theories (over $\Sigma$)
  with a \emph{set} $\Ad$ of models (i.e.\ $\Sigma$-structures), calling a set of HoCHCs \emph{$\Ad$-satisfiable} if
    it is $\As$-satisfiable for some $\As\in\Ad$. Otherwise it is \emph{$\Ad$-unsatisfiable}.

  Observe that the Completeness \cref{thm:completeness} critically relies on the
  observation that $\As$-unsatisfiability can be traced back to
  the failure of \emph{a single} goal clause of background atoms
  (manifested in the constraint refutation rule). Therefore, it is
  natural to generalise constraint refutation to a rule refuting
  \emph{sets} of $\Ad$-unsatisfiable\footnote{Note that for a set
    $\Set$ of goal
    clauses of background atoms,
    $\As$-satisfiability is in essence not about the existence of
    $\As\in\Ad$ and an \emph{expansion} $\Bs$ of $\As$ such that $\Bs\models\Set$ but only about the existence of
    $\As\in\Ad$ such that $\As\models\Set$.}
  goal clauses of background atoms:
%}

\medskip

\inferthrlics{Comp. const.
  refutation}{G_1\lor\biglor_{i=1}^{m_1}\neg\phi_{1,i}}{\hspace{-0.75cm}\ldots
  \hspace{-0.75cm}}{G_n\lor\biglor_{i=1}^{m_n}\neg\phi_{n,i}}{\bot}{provided
  that each atom in each $G_i$ has the form $x\,\overline M$, each $\phi_{i,j}$ is a background atom
  and $\{\neg\phi_{j,1}\lor\cdots\lor\neg\phi_{j,m_j}\mid 1\leq
  j\leq n\}$ is $\Ad$-unsatisfiable.}

\smallskip
\noindent and let $\Resp$ be defined accordingly.
%\changed[dw]{
However, to match the rule's finitary nature, $\Ad$ needs to be restricted a little:
%}
\begin{definition}
\label{def:compact}
  A set $\Ad$ of 1st-order $\Sigma$-structures is \emph{compact} if for all $\Ad$-unsatisfiable sets $\Set$ of goal clauses of background atoms there exists a finite $\Set' \subseteq\Set$ which is $\Ad$-unsatisfiable.
\end{definition}

%\lo{In \cref{def:compact}(ii), we \emph{could} further restrict $S$ to be a set of goal clauses of background atoms. Advantages: This is simpler and suffices for \cref{thm:compcomp} [``$\Leftarrow$''].}

%\dw{I am not sure whether this is helpful because a priori it is not obvious at all whether such theories exist at all and in how far this coincides with 1st-order compactness properties. } \lo{Ok - I assume you are referring to the following comment about renaming compact.}

%\lo{Instead of compact, I suggest if we rename the set $\Ad$ to \emph{compactifying} or \emph{compactness-supporting} or \emph{compactness basis}. For one thing, it would make the following sentence less confusing.}

In particular, every finite $\Ad$ is compact. Then we obtain:

\begin{restatable}[Soundness and Completeness]{theorem}{compcomp}
  \label{thm:compcomp}
  Let $\Ad$ be a compact set of $\Sigma$-structures
  and $\Set$ be a set of HoCHCs. Then
  $\Set$ is $\Ad$-unsatisfiable iff $\Set\Resp^* \Set'\cup\{\bot\}$ for some $\Set'$.
\end{restatable}
% \begin{proof}
%   The ``if''-direction is straightforward. For the converse, suppose
%   that $S$ is $\Ad$-unsatisfiable.
%   By the Completeness \cref{thm:completeness}, for each $\As\in\Ad$ there exist easy $G_\As$ and
%   background atoms $\phi_{\As,i}$ and $S_\As$ such that
%   $\As\not\models\neg\phi_{\As,1}\lor\cdots\lor\neg\phi_{\As,m_\As}$
%   and $S\Res^*
%   S_\As\cup\{G_\As\lor\neg\phi_{\As,1}\lor\cdots\lor\neg\phi_{\As,m_\As}\}=S'_\As$. Hence,
%   $\{\neg\phi_{\As,1}\lor\cdots\lor\neg\phi_{\As,m_\As}\mid\As\in\Ad\}$
%   is $\Ad$-unsatisfiable and by compactness of $\Ad$ there exists finite $\Ad'\subseteq\Ad$
%   such that $\{\neg\phi_{\As,1}\lor\cdots\lor\neg\phi_{\As,m_\As}\mid\As\in\Ad'\}$
%   is $\Ad$-unsatisfiable. Consequently, $S\Resp^*\{S'_\As\mid\As\in\Ad'\}\Resp\{\bot\}\cup\{S'_\As\mid\As\in\Ad'\}$.
%   \lo{Nice argument.}
% \end{proof}
%\changed[dw]{
As an interesting special case, this shows that the proof system is
also sound and complete in the unconstrained setting:
by the compactness theorem for 1st-order logic the set of 1st-order $\Sigma$-structures (possibly interpreting
\mbox{(in-)equality} symbols as (non-)identity)
% $\Ad=\{\As\mid\As\text{ is a 1st-order $\Sigma$-structure}\}$
is compact. 
Consequently, there does not exist a $\Sigma'$-structure $\Bs$
(interpreting (in-)equality as \mbox{(non-)identity}) satisfying $\Bs\models\Set$ iff $\Set$ is refutable.
%}

\section{1st-order Translation}
% \section{Applicative Encoding}
\label{sec:appenc}
\label{SEC:APPENC}

% \lo{Changed ``encoding'' in \cref{sec:appenc} to ``translation'', for consistency with the usage in the Introduction.}
\begin{figure*}[!t]
  \label{fig:appenc}
  \centering
  \begin{align*}
    \lfloor D_1\rfloor&=\neg(z=x+y)\lor H\,(\app\, (\app\, (\app\,\Add\,x)\,y)\,z)\\
    \lfloor D_2\rfloor&=\neg(n\leq 0)\lor\neg(s=x)\lor H\,(\app\,(\app\, (\app\, (\app\,\Iter\,f)\,s)\,n)\,x)\\
    \lfloor D_3\rfloor&=\neg(n>0)\lor\neg H\,(\app\,(\app\, (\app\, (\app\,\Iter\,f)\,s)\,(n-1))\,y)\lor\neg H\,(\app\,(\app\,(\app\,f\,n)\,y)\,x)\lor H\,(\app\,(\app\, (\app\, (\app\,\Iter\,f)\,s)\,n)\,x)\\
    \lfloor G\rfloor&=\neg(n\geq 1)\lor\neg H\,(\app\,(\app\, (\app\, (\app\,\Iter\,\Add)\,n)\,n)\,x)\lor\neg(x\leq n+n)\\
    \Comp_{\iota^3\to o}&=H\,(\app\,(\app\,(\app c_{\iota^3\to o}\, x_1)\,x_2)\,x_3)
  \end{align*}
  \caption{1st-order translation of the set of HoCHCs of \cref{ex:HoCHC}.}
  \label{fig:appenc}
\end{figure*}

%\changed[dw]{
It is folklore that there is a 1st-order translation of
higher-order logic which is sound and complete for Henkin semantics
(see e.g.\ \cite{BD83,K91,BKPU16}).
The essence of the technique is to replace all symbols by constants
(of a base type) and encode application using dedicated binary
function symbols.

For the reasons discussed in the introduction this
translation is however not in
general complete for standard semantics.
In this section, we present a particularly simple 1st-order
translation of HoCHC which is sound and complete even for standard
semantics.
Fortunately, the target fragment is still semi-decidable.
%}

% We present a method to reduce the HoCHC \mbox{$\Ad$-satisfiability} problem to the satisfiability problem of 1st-order logic modulo a theory, using a translation in the spirit of
% e.g.\ \cite{BD83,K91,BKPU16}. We prove that this translation is sound
% and complete, even for standard semantics. Fortunately, the target fragment is still semi-decidable.

We do not need to consider HoCHCs containing $\lambda$-abstractions
because
% \dw{are such line breaks considered good style?} \lo{I am fine with it -- I don't see anything odd.}
they can be eliminated by a logical counterpart of
$\lambda$-lifting \cite{J85}
% \dw{is this the right reference?}
% \lo{Yes, I believe so.}
(i.e.\ introducing new relational
symbols and adding appropriate ``definitions'' for them\ifproceedings\ \cite{OW19}\else, see \cref{sec:elimnest,thm:ellambdahochc})\fi.
% we can replace every $\lambda$-abstraction $\lambda y\ldotp M$, where $\Delta\vdash\lambda y\ldotp M\from\tau'\to\overline\tau\to o$ and $\free(\lambda y\ldotp M)=\overline x$, 
% by $R_M\,\overline x$, where $R_M\from\Delta(\overline x)\to\tau'\to\overline\tau\to o$ is fresh, and add the clause $\neg M\,\overline z \, \lor \, R_M\,\overline x\,y\,\overline z$, 
% where $\Delta(\overline z)=\overline\tau$ and all of $\overline x$, $y$ and $\overline z$ are distinct.
This constitutes a considerable generalisation
of the ``polarity-dependent renaming'' for 1st-order logic
\cite{PG86,NW01}.

Therefore, the following is without loss of generality:
\begin{assumption}
  Henceforth, we fix a finite set $\Set$ of HoCHCs which does not contain
  $\lambda$-abstractions and a set $\Ad$ of 1st-order $\Sigma$-structures.
\end{assumption}

% Clearly, the signature $\Sigma$ and the $\Sigma$-structure $\As$ can be regarded as a many-sorted signature and structure, respectively (with only one sort $\iota$).
% \lo{I am not sure about the preceding sentence: 
% I take many-sorted to mean more than one sort.}
% \dw{I've defined a 1st-order signature to be a
%   higher-order signature containing only symbols of 1st-order type
%   and a higher-order signature is by definition single-sorted. A
%   many-sorted signature, on the other hand, is always 1st-order and
%   may contain only a single sort (i.e. ``possibly many-sorted''). Hence, there is an obvious
%   correspondence between 1st-order (higher-order) signatures and
%   many-sorted 1st-order signatures.}

% \dw{Obviously, this is not very elegant from a theoretical
%   perspective. Do you think it would be better to develop HoL in a
%   many-sorted context although most of the time it will be restricted
%   to a single sort?}

Let $\btypes=\{\iota\}\cup\{\lfloor\rho\rfloor\mid\rho\text{
  relational}\}$ (and we set
$\lfloor\iota^n\to\iota\rfloor\defeq\iota^n\to\iota$). Clearly, we can
regard $\Sigma$ and each $\As\in\Ad$ as a 1st-order signature and structure, respectively, over the
extended set of types of individuals.

We assume a type environment $\lfloor\Delta\rfloor$ such that for
$x\from\tau\in\Delta$, $\lfloor\Delta\rfloor(x)=\lfloor\tau\rfloor$
and define $\lfloor\Sigma'\rfloor$ to be the following 1st-order
extension of $\Sigma$:
\begin{align*}
  \Sigma&\cup\{c_R\from\lfloor\rho\rfloor\mid
  {R\from\rho}\in{\Sigma'\setminus\Sigma}\}\\
                                       &\cup\{c_\rho\from\lfloor\rho\rfloor\mid\rho\text{ relational}\}\\
                                       &\cup\{\app_{\tau,\rho}\from\lfloor\tau\to\rho\rfloor\to\lfloor\tau\rfloor\to\lfloor\rho\rfloor\mid\tau\to\rho\text{
                                         relational}\}\\
                                       &\cup\{H\from\lfloor o\rfloor\to o\}
\end{align*}
To reduce clutter, we often omit the subscripts from $\app$.

% \begin{enumerate*}[noitemsep]
% \item $c_R\from\lfloor\rho\rfloor$, for each ${R\from\rho}\in{\Sigma'\setminus\Sigma}$;
% \item $c_\rho\from\lfloor\rho\rfloor$, for each relational $\rho$;
% \item $\app_{\tau,\rho}\from\lfloor\tau\to\rho\rfloor\to\lfloor\tau\rfloor\to\lfloor\rho\rfloor$, for each relational $\tau\to\rho$, and
% \item $H\from\lfloor o\rfloor\to o$.
% \end{enumerate*}
% following observation is trivial:

Intuitively, $\app$ encodes application, relational symbols $R\in\Sigma'\setminus\Sigma$ become
constants $c_R$, %\changed[dw]{
$H$ maps the ``bogus booleans``
$\lfloor o\rfloor$ to $o$
%} 
and the following \emph{comprehension axiom} $\Comp_\rho$ (for
relational $\rho=\tau_1\to\cdots\to\tau_n\to o$) asserts the existence of an element (the interpretation of $c_\rho$) corresponding to $\top_\rho$:
\begin{align*}
   \Comp_\rho \defeq H\, (\app\, (\cdots(\app\,(\app\,c_\rho\,x_1)\,x_2)\cdots )\,x_n)
 \end{align*}
 where the $x_i$ are distinct variables of type $\lfloor\tau_i \rfloor$.

For a $\Sigma'$-term $M$ containing neither logical symbols nor $\lambda$-abstractions, we define $\lfloor M\rfloor'$ by structural recursion: %induction:
\begin{align*}
  \lfloor x\rfloor'&\defeq x\\
  \lfloor R\rfloor'&\defeq c_R&\text{if }R\in\Sigma'\setminus\Sigma\\
  \lfloor c\,\overline N\rfloor'&\defeq c\,\overline N &\text{if $c\in\Sigma$}\\
  \lfloor M\,\overline N\,N'\rfloor'&\defeq \app\,\lfloor M\,\overline N\rfloor'\,\lfloor N'\rfloor'&\text{if $M\nin\Sigma$}
\end{align*}
% \dw{should the side conditions be put in parentheses as e.g. in \cref{def:arelm}}
% \lo{I think not. There is a difference. The side expressions in
%   \cref{def:arelm} are universal quantifications. Those that are here
%   are proper side conditions.}
Thus, terms of the background theory are unchanged by $\lfloor \cdot
\rfloor'$ and
\iffalse
\lo{SELF-NOTE: In the last clause, $M$ is either a variable, or $R \in \Sigma' \setminus \Sigma$.}
\fi
by \cref{rem:isimple}, for each $\Sigma'$-term
$\Delta\vdash M\from\sigma$
which is not a background atom,
$\lfloor\Delta\rfloor\vdash\lfloor
M\rfloor'\from\lfloor\sigma\rfloor$.
The following operator $\lfloor\cdot\rfloor$ ensures that also foreground
atoms have type $o$ (not $\lfloor o\rfloor$)
  \[
\lfloor A \rfloor \defeq
\left\{
\begin{array}{ll}
  A & \hbox{if $A=c\,\overline N$ with $c \in \Sigma$}\\
  H \, \lfloor A \rfloor' & \hbox{otherwise ($A$ is a foreground atom).}
\end{array}
\right.
\]
and we lift $\lfloor\cdot\rfloor$ to HoCHCs by
\begin{align*}
  \lfloor(\neg) A_1\lor\cdots\lor(\neg) A_n\rfloor&\defeq (\neg)\lfloor A_1\rfloor\lor\cdots\lor(\neg)\lfloor A_n\rfloor
\end{align*}
% \dw{is this clear enough}
% \lo{yes}
% \begin{align*}
%   \lfloor\neg A_1\lor\cdots\lor\neg A_n\rfloor&\defeq \neg\lfloor A_1\rfloor\lor\cdots\lor\neg\lfloor A_n\rfloor\\
%   \lfloor\neg A_1\lor\cdots\lor\neg A_n\lor R\,\overline x\rfloor&\defeq\neg\lfloor A_1\rfloor\lor\cdots\lor\neg\lfloor A_n\rfloor\lor\lfloor R\,\overline x\rfloor
% \end{align*}
% where
%   \[
% \lfloor A \rfloor \defeq
% \left\{
% \begin{array}{ll}
%   A & \hbox{if $A=c\,\overline N$ with $c \in \Sigma$}\\
%   H \, \lfloor A \rfloor' & \hbox{otherwise ($A$ is a foreground atom).}
% \end{array}
% \right.
% \]
Finally, for $\Set$ we set
 \begin{align*}
   \lfloor\Set\rfloor\defeq\{\lfloor C\rfloor\mid C\in\Set\}\cup\{\Comp_\rho\mid {x\from\rho} \in\Delta\text{ occurs in }\Set\}.
 \end{align*}
 Note that $\lfloor\Set\rfloor$ is a finite set of 1st-order Horn clauses\footnote{in the standard sense} of the (1st-order) language of $\lfloor \Sigma' \rfloor$.
 %\lo{Have we defined $\lfloor\Sigma'\rfloor$-Horn clauses?}

% We say that $\lfloor S\rfloor$ is \emph{$\Ad$-satisfiable} if there exists $\As\in\Ad$ and a $\lfloor\Sigma\rfloor$-expansion $\Bs$ of $\As$ such that $\Bs\models\lfloor S\rfloor$, and it is $\Ad$-unsatisfiable otherwise.
\begin{example}[1st-order translation $\lfloor \cdot \rfloor$]
  Consider again the set $\Set$ of HoCHCs from \cref{ex:HoCHC}. 
  Applying the translation $\lfloor \cdot \rfloor$ to $\Set$ we get the 1st-order clauses in \cref{fig:appenc}.
\end{example}

For $\As\in\Ad$ and a $\Sigma'$-expansion $\Bs$ of $\As$, let
$\lfloor\Bs\rfloor$ be the 1st-order $\lfloor\Sigma\rfloor$-expansion
of $\As$ defined by
\[
  \begin{array}{c}
    \sinti{\lfloor\Bs\rfloor}{\lfloor\rho\rfloor} \defeq
    \sinti\Bs\rho\qquad c_R^{\lfloor\Bs\rfloor} \defeq R^\Bs\qquad
    c_\rho^{\lfloor\Bs\rfloor}\defeq \top_\rho\\
    \app^{\lfloor\Bs\rfloor}_{\tau,\rho'}(r)(s)\defeq r(s)\qquad H^{\lfloor\Bs\rfloor}(b)\defeq b
  \end{array}
\]
for relational $\rho$ and $\tau\to\rho'$,
$R\in\Sigma'\setminus\Sigma$, $r\in\sinti{\lfloor\Bs\rfloor}{\lfloor\tau\to\rho'\rfloor}$,
     $s\in\sinti{\lfloor\Bs\rfloor}{\lfloor\tau\rfloor}$ and
     $b\in\sinti{\lfloor\Bs\rfloor}{\lfloor o\rfloor}=\bool$.
     % \dw{is that list useful/necessary}
     % \lo{Yes I think so.}
% \begin{enumerate*}[noitemsep]
%    \item $\sinti{\lfloor\Bs\rfloor}{\lfloor\rho\rfloor} \defeq \sinti\Bs\rho$ for relational $\rho$,
%    \item $c_R^{\lfloor\Bs\rfloor} \defeq R^\Bs$ for
%      $R\in\Sigma'\setminus\Sigma$,
%    \item $c_\rho^{\lfloor\Bs\rfloor}\defeq \top_\rho^\Sf$ for relational $\rho$,
%    \item $\app_{\tau,\rho}(r)(s)\defeq r(s)$ for relational
%      $\tau\to\rho$,
%      $r\in\sinti{\lfloor\Bs\rfloor}{\lfloor\tau\to\rho\rfloor}$ and
%      $s\in\sinti{\lfloor\Bs\rfloor}{\lfloor\tau\rfloor}$, and
%    \item $H^{\lfloor\Bs\rfloor}(b)\defeq b$ for $b\in\sinti{\lfloor\Bs\rfloor}{\lfloor o\rfloor}=\bool$.
%    \end{enumerate*}
   It is easy to see that $\Bs\models\Set$ implies $\lfloor\Bs\rfloor\models\lfloor\Set\rfloor$. Consequently:
   
 \begin{proposition}%[Soundness of Translation]
   \label{prop:folpreservesat}
   If $\Set$ is $\Ad$-satisfiable then $\lfloor\Set\rfloor$ is $\Ad$-satisfiable.
 \end{proposition}

 Conversely, applications of the (higher-order) resolution rule can be matched by 1st-order resolution inferences between the corresponding translated % encoded
 clauses. Besides, the 1st-order translation contains comprehension
 axioms $\Comp_\rho$, which complements the instantiation of relational variables
 with $\top_\rho$ in the proof of the Soundness \cref{prop:soundness}. Therefore, we obtain:
 \begin{restatable}{lemma}{follifting}
   \label{lem:follifting} Let $\Set'$ be a set of HoCHCs not containing $\lambda$-abstractions and suppose $\Set'\Respstr{\Ad}\Set'\cup\{G\}$. Then
   \begin{thmlist}
   \item $G$ does not contain $\lambda$-abstractions
   \item if $G\neq\bot$ then $\lfloor \Set'\rfloor\models\lfloor \Set'\cup\{G\}\rfloor$
   \item if $G=\bot$ then $\lfloor \Set'\rfloor$ is $\Ad$-unsatisfiable.
   \end{thmlist}
 \end{restatable}
 By the Completeness \cref{thm:compcomp} we conclude:
 % \dw{Is it okay to refer to the appendix without explicitly mentioning it?}
 % \lo{Not ideal, but it is not a cardinal sin.}
 \begin{corollary}%[Completeness of Translation]
  If $\Ad$ is compact and $\Set$ is $\Ad$-unsatisfiable then $\lfloor \Set\rfloor$ is $\Ad$-unsatisfiable.
\end{corollary}

\begin{theorem}
  Assuming that $\Ad$ is compact, $\Set$ is $\Ad$-satisfiable iff $\lfloor \Set\rfloor$ is $\Ad$-satisfiable.
\end{theorem}
    %     Furthermore, by \cite[Theorem 24]{BGW94} we get:

It is remarkable that our translation does not require extensionality
axioms and
only a very restricted form of comprehension axioms (cf.~\cite{vanBehthemD83}).

% \dw{comment on sufficient completeness or too technical?}
Finally, if $\Ad$ is compact, %\emph{term-generated}
\emph{definable}\footnote{or \emph{term-generated} \cite{BGW94}, i.e.\ for every  $\As\in\Ad$ and $a\in\sinti\As\iota$ there exists a closed $\Sigma$-term $M$ such that $\sinti\As M=a$}
% \dw{it might still be worth mentioning the \emph{term-generated} because this term is used in \cite{BGW94}, which I am referring to here} \lo{Agreed; done.} 
and $\Ad$-unsatisfiability of goal clauses of background atoms is semi-decidable, 
then $\Ad$-unsatisfiability of $\lfloor\Set\rfloor$ is also semi-decidable \cite[Thm.~24]{BGW94}.

\iffalse
\lo{N.B. We actually do not need (indeed cannot take advantage of) decidability of $\Ad$-unsatisfiability of goal clauses of background atoms. I.e.~even if $\Ad$-unsatisfiability is decidable (e.g.~LIA), we cannot strengthen the result that $\lfloor\Set\rfloor$ is semi-decidable.}
\dw{Yes. When implementing the proof system decidability can come in handy, though.}
\lo{Agreed.}
\fi

%\lo{\emph{Definable} is a common name (standard in semantics of programming languages) for such a property.}

% \begin{theorem}
%   Assume that $\Ad$ is compact and that for every $\As\in\Ad$ and
% $a\in\sinti\As\iota$ there exists a ground $\Sigma$-term $M$ such that $\sinti\As M=a$.
% If $\Ad$-satisfiability of background goal clauses is
%   semi-decidability then $\Ad$-satisfiability of $\lfloor S\rfloor$ is
%   semi-decidable, too.
% \end{theorem}

% Note that for the theory of positive equality and uninterpreted functions (cf.\ \cref{sec:extlp}) as considered in \cite{CHRW13} this encoding results in a set of 1st-order Horn clauses effectively \emph{without} a background theory.

\section{Decidable Fragments}
  \label{sec:decidable}

Satisfiability of HoCHC is undecidable in general because already its 1st-order fragments are undecidable for Linear Integer Arithmetic \cite{D72,HVW17} or the unconstrained setting\footnote{i.e.\ the background theory imposes no restriction at all} \cite{M03}.
  
  \begin{remark}
    \label{rem:dec}
    Despite these negative results, $\Ad$-satisfiability of finite $\Set$ is decidable if
    % \begin{enumerate*}
    % \item $S$ is a finite set of HoCHCs and
    % \item
      $\Ad$ is a finite set of $\Sigma$-structures such that for each $\As\in\Ad$ and type $\sigma$, $\sinti\As\sigma$ is finite.
    % \end{enumerate*}
    This is a consequence of \cref{thm:canmod} and the fact that we can compute each $\As_{\Prgm_\Set}$ explicitly and check whether $\As_{\Prgm_\Set}\models\Set$ holds.
  \end{remark}
Thanks to this insight, we have identified two decidable fragments of HoCHC, one of which is presented as follows; 
we leave the other ({higher-order Datalog}) to \ifproceedings\cite{OW19}\else\cref{sec:hodatalog}\fi.
% \lo{N.B. No appendix in proceedings version.}
% \dw{there is a proceedings-flag. is this ok?}
% \lo{Yes. Thanks.}

   \subsection{Combining the Bernays-Sch\"onfinkel-Ramsey Fragment of
     HoCHC with Simple Linear Integer Arithmetic}
   \label{sec:HoBHC(SLA)}
  Some authors \cite{GM09,HVW17b} have studied 1st-order clauses without function symbols (the so-called \emph{Bernays-Sch\"onfinkel-Ramsey class}\footnote{Precisely the set of sentences that, when written in prenex normal form, have a $\exists^{*}\forall^{*}$-quantifier prefix and contain no function symbols.}) extended with a restricted form of Linear Integer Arithmetic.
  The fragment enjoys the attractive property that every clause set is equi-satisfiable with a finite set of its ground instances, which implies decidability \cite{GM09,HVW17b}.  
  % \dw{Do you think that this should be omitted or did you delete it by mistake?}
  % \lo{I thought you had used {\tt sout} to remove it (?).
  %   But I completely agree that this is an important explanation, and should be reinstated.}
  % \dw{I think it was duplicated and I removed one version}
  %\changed[dw]{\sout{The fragment enjoys the attractive property that every set of FoSSLAs is equi-satisfiable with a finite set of its ground instances, which implies decidability of this fragment \cite{GM09,HVW17}.}}

  In this section, we transfer this result to our higher-order Horn setting.

  \begin{assumption}
    Let $\Sigma$ be a (1st-order) signature extending $\Sigma_{\LIA}$ with constant symbols $c\from\iota$, and let $\Sigma'\supseteq\Sigma$ be a relational extension of $\Sigma$.
  \end{assumption}
  % Then they restrict background atoms to
  % the following:
  \begin{definition}
    \begin{thmlist}
      \item   A $\Sigma$-atom is \emph{simple} % \dw{name already used}
    % \lo{I suggest ``basic''.}
    % \dw{I've renamed ``simple background atoms'' into ``trivial background atoms'' to stay consistent with the literature here} 
    % \lo{OK.}
    if it has the form $x\leq M$,
    $M\leq x$ or
    $x\leq y$,
    where $M$ is closed\footnote{or \emph{ground} because atoms do not contain (existential) quantifiers by definition}.
% $M\triangleleft N$, $x\triangleleft M$ or $x\trianglelefteq y$, where $M,N$ are ground, $\triangleleft\in\{<,\leq,=,\neq,\geq,>\}$ and
%     $\trianglelefteq\;\in\{\leq,=,\geq\}$.
    \item 
      A HoCHC is a \emph{higher-order simple linear arithmetic Bernays-Sch\"onfinkel-Ramsey Horn clause (HoBHC(SLA))} 
      if it has the form $\neg\phi_1\lor\cdots\lor\neg\phi_n\lor C$, where each $\phi_i$ is a simple (linear arithmetic) background atom and $C$ is $\bot$ or it does not contain symbols from $\Sigma$.
    \end{thmlist}
  \end{definition}

Note that we could also have allowed background atoms of the form $M\triangleleft N$, $x\triangleleft M$ and $x\trianglelefteq y$, where $M,N$ are closed,
%\lo{ground (closed?)}\dw{I don't define ``ground'' any more (but closed), see footnote above}, 
$\triangleleft\in\{<,\leq,=,\neq,\geq,>\}$ and $\trianglelefteq\;\in\{\leq,=,\geq\}$ \cite{HVW17b}.

\begin{example}
\label{eg:HoBHC(SLA)}
  Let
  $\Sigma=\Sigma_{\LIA}\cup\{c,d\from\iota\}$, $\Sigma'=\Sigma\cup\{R\from\iota\to
  o,U\from(\iota\to o)\to\iota\to o\}$,
  $\Delta(x)=\Delta(y)=\Delta(z)=\iota$ and $\Delta(f)=\iota\to o$.
  The following is a set of HoBHC(SLA):
  \begin{align*}
    &\neg (x\leq c+d-5)\lor R\,x\\
    &\neg f\,x\lor\neg (y\leq x)\lor\neg (x\leq d)\lor U\,f\,y\\
    &\neg (c\leq x)\lor\neg (x\leq -1)\\
    &\neg (x\leq d-5)\lor\neg (d-5\leq x)\lor\neg (y\leq c-10)\,\lor\\
    &\hspace{1cm}\neg (c-10\leq y)\lor\neg U\,(\lambda z\ldotp R\,x)\,y.
  \end{align*}
  \iffalse
  \lo{Every HoBHC(SLA) has a disjunct $C$ which does not contain any $\Sigma$-symbol.
    There seems to be a problem with clause 3: every disjunct contains a $\Sigma$-symbol ($-1$ is in $\Sigma$).}
  \dw{$C$ may be empty}
  \lo{What is $\Sigma$ in this example?}
  \lo{We need to know $\Sigma$ in order to check that the above clauses
    are HoBHC(SLA).}
  \dw{I've made it more explicit now}
  \fi
\end{example}

\begin{assumption}
  \label{ass:HoBHCSLA}
Let $\Ad$ be the set of $\Sigma$-expansions of $\As_{\LIA}$ and let $\Set$ be a finite set of HoBHC(SLA).
\end{assumption}
%\changed[dw]{
As in the 1st-order case, only the relations between
ground terms are relevant. Therefore, we replace ground terms $M$
with (fresh) constant symbols $c_M$ and consider only structures in
which ``$\leq$'' is interpreted consistently (with the meaning of
the constants).
%}

Formally, let $\igt(\Set)$ be the set of closed terms of type $\iota$
occurring in $\Set$ and
we define
\begin{align*}
  \Sigma^\flat&\defeq\{{\leq}\from{\iota\to\iota\to
  o}\}\cup\{c_M\from\iota\mid M\in\igt(\Set)\}\\
  (\Sigma')^\flat&\defeq {\Sigma^\flat\cup(\Sigma'\setminus\Sigma)}
\end{align*}

% \dw{add example?}

%\lo{For the clause-set $S$ of \cref{eg:HoBHC(SLA)}, the (maximal) elements of $\igt(S)$ are: $c + d - 5, d,  c, -1, d-5, c-10$.}

For $\As\in\Ad$, let $\As^\flat$ be the
1-order $\Sigma^\flat$-structure defined by
\[
  \begin{array}{c}
    \sinti{\As^\flat}\iota \defeq \igt(\Set)\quad
    {\leq^{\As^\flat}}(M)(N) \defeq \sinti\As{M\leq N}\quad
    c_M^{\As^\flat}\defeq M
  \end{array}
\]
for $M,N\in\igt(\Set)$,
and let $\Ad^\flat\defeq\{\As^\flat\mid\As\in\Ad\}$.

Furthermore, for simple atoms $x\leq M$ and $M\leq x$, we set ${(x\leq M)^\flat}\defeq {x\leq c_M}$ and ${(M\leq x)^\flat}\defeq {c_M\leq x}$. 
For all other atoms $A$ (i.e.\ $x\leq y$ or foreground atoms) we set $A^\flat\defeq A$; we lift $\cdot^\flat$ in the obvious way to clauses\footnote{i.e.\ ${(\neg A_1\lor\cdots\lor\neg A_n\lor(\neg) A)^\flat}\defeq\neg A_1^\flat\lor\cdots\lor\neg A_n^\flat\lor(\neg) A^\flat$} and define $\Set^\flat\defeq\{C^\flat\mid C\in \Set\}$.
Note that $\Set^\flat$ is a set of HoCHCs for $\Sigma^\flat$ and
$(\Sigma')^\flat$, and that $\Ad^\flat$ is finite.

Clearly, there is an inverse $\cdot^\sharp$ of $\cdot^\flat$ on
formulas, e.g.\ satisfying $(x\leq c_M)^\sharp=(x\leq M)$.

%\changed[dw]{
  % Suppose $\As\in\Ad$. For $M\in\sinti{\As^\flat}\iota=\igt(S)$, let
  % $M^\sharp\defeq\sinti\As M\in\sinti\As\iota=\mathbb Z$. Trivially,
  % for all $M,N\in\igt(S)$,
  % $\As\models M\leq N$ iff $M^\sharp\leq N^\sharp$.

  % Conversely, for $n\in\sinti\As\iota=\mathbb Z$, let
  %   \begin{align*}
  %     m^\flat&\defeq
  %                 \begin{cases}
  %                   \argmax_{M\in\igt(S)}\sinti\As M\\
  %                   &\hspace{-4cm}\text{if $\{M\in\igt(S)\mid\sinti\As
  %                     M\geq m\}=\emptyset$}\\
  %                   \argmin_{M\in\igt(S),\sinti\As M\geq m}\sinti\As M&\text{otherwise}
  %                 \end{cases}
  %   \end{align*}
  %   and for $m\leq n\in\mathbb Z$ it holds $\As\models m^\flat\leq n^\flat$.

  Now, suppose $\As\in\Ad$. Then valuations $\alpha$ over (a frame induced
  by) $\sinti{\As^\flat}\iota$ naturally correspond to valuations
  $\alpha^\sharp$ over $\sinti\As\iota$ by evaluating ground
  terms\footnote{precisely, $\alpha^\sharp(x)=\sinti\As{\alpha(x)}$ for
  $x\from\iota\in\Delta$} and it
  holds $\sinti\As\phi(\alpha^\sharp)=\sinti{\As^\flat}{\phi^\flat}(\alpha)$
  for simple background atoms $\phi$.

  Conversely, for valuations $\alpha$ and $\alpha^\flat$ (over
  $\sinti\As\iota$ and $\sinti{\As^\flat}\iota$, respectively)
  satisfying
  
  % \begin{align*}
  %   \alpha^\flat(x)&=
  %                 \begin{cases}
  %                   \argmax_{M\in\igt(\Set)}\sinti\As M\\
  %                   &\hspace{-4cm}\text{if $\{M\in\igt(\Set)\mid\sinti\As M\geq\alpha(x)\}=\emptyset$}\\
  %                   \argmin_{M\in\igt(\Set)\land\sinti\As M\geq\alpha(x)}\sinti\As M&\text{otherwise}
  %                 \end{cases}
  % \end{align*}
  \noindent
  \begin{minipage}{0.6\linewidth}
    \vspace*{-3mm}\begin{align*}
    \alpha^\flat(x)&=
                  \begin{cases}
                    \argmax_{M\in\igt(\Set)}\sinti\As M\\\\
                    \argmin_{M\in\igt(\Set)\land\sinti\As M\geq\alpha(x)}\sinti\As M
                  \end{cases}
  \end{align*}
  \end{minipage}
  \begin{minipage}{0.1\linewidth}
    \vspace*{-3mm}
    \begin{align*}
      \\
      \hspace*{-4cm}\text{if $\{M\in\igt(\Set)\mid\sinti\As M\geq\alpha(x)\}=\emptyset$}\\
      \hspace*{-4cm}\text{otherwise}
  \end{align*}
  \end{minipage}
  
  \noindent for $x\from\iota\in\Delta$, it holds
    $\sinti\As\phi(\alpha)\leq\sinti{\As^\flat}{\phi^\flat}(\alpha^\flat)$
    if $\igt(\phi)\subseteq\igt(\Set)$. Therefore:
%}

  % \begin{lemma}
  %   Let $\As\in\Ad$ and $\phi_1,\ldots,\phi_n$ be simple background atoms satisfying $\bigcup_{i=1}^n\igt(\phi_i)\subseteq\igt(S)$.

  %   Then there is a valuation $\alpha$ satisfying $\As,\alpha\models\phi_1\land\cdots\land\phi_n$ iff there is a valuation $\alpha'$ satisfying $\widehat\As,\alpha'\models\widehat\phi_1\land\cdots\land\widehat\phi_n$.
  % \end{lemma}
  % \begin{proof}
  %   First, suppose that $\alpha'$ satisfies $\widehat\As,\alpha'\models\widehat\phi_1\land\cdots\land\widehat\phi_n$. Let $\alpha$ be such that for $x\from\iota\in\Delta$, $\alpha(x)=\sinti\As{\alpha'(x)}$. It is easy to see that $\As,\alpha\models\phi_1\land\cdots\land\phi_n$.

  %   Conversely, suppose that $\alpha$ satisfies $\As,\alpha\models\phi_1\land\cdots\land\phi_n$. For $x\from\iota\in\Delta$, we define
  %   \begin{align*}
  %     \alpha'(x)&\defeq
  %                 \begin{cases}
  %                   \argmax_{M\in\igt(S)}\sinti\As M\\
  %                   &\hspace{-4cm}\text{if $\{M\in\igt(S)\mid\sinti\As M\geq\alpha(x)\}=\emptyset$}\\
  %                   \argmin_{M\in\igt(S)\land\sinti\As M\geq\alpha(x)}\sinti\As M&\text{otherwise}
  %                 \end{cases}
  %   \end{align*}
  %   It is easy to see that $\sinti\As{\phi_i}(\alpha)\leq\sinti{\widehat\As}{\phi_i}(\alpha')$ % \dw{more details?} \lo{I think not necessary.}
  %   and therefore, $\widehat\As,\alpha'\models\widehat{\phi}_1\land\cdots\land\widehat{\phi}_n$.
  % \end{proof}
  % \dw{move to appendix}
  \begin{lemma}
    \label{cor:bgiff}
    Let $\Set'$ be a set of goal clauses of simple background atoms
    satisfying $\igt(\Set')\subseteq\igt(\Set)$.
    
    Then $\Set'$ is $\Ad$-satisfiable iff $(\Set')^\flat$ is $\Ad^\flat$-satisfiable.
  \end{lemma}

  \begin{restatable}{lemma}{hobhcslaSWidehatS}
    Let $\Set'$ be a set of HoBHC(SLA) satisfying $\igt(\Set')\subseteq\igt(\Set)$. Then
    \begin{thmlist}
    \item\label{cor:liftgr} $\Set'\Respstr\Ad \Set'\cup\{G\}$ implies $(\Set')^\flat\Respstr{\Ad^\flat} (\Set')^\flat\cup\{G^\flat\}$
    \item $(\Set')^\flat\Respstr{\Ad^\flat}
      (\Set')^\flat\cup\{G\}$ implies $\Set'\Respstr\Ad
      \Set'\cup\{G^\sharp\}$.
    \end{thmlist}
  \end{restatable}

  The proof of the Completeness \cref{thm:compcomp} can be
  strengthened
  \ifproceedings\cite{OW19} \else(\cref{thm:compphi} in
  \cref{sec:appHoBHC(SLA)}) \fi
  to yield:
  \begin{restatable}{proposition}{complHoBHCSLA}
    \label{prop:complHoBHCSLA}
    If $\Set$ is $\Ad$-unsatisfiable then $\Set\Resp^* \Set'\cup\{\bot\}$ for some $\Set'$.
  \end{restatable}

  Consequently, if $\Set$ is $\Ad$-unsatisfiable then $\Set\Respstr\Ad^* \Set'\cup\{\bot\}$ for
    some $\Set'$. 
    It is easy to see that sets $\Set'$ of HoBHC(SLA)
      satisfying $\igt(\Set')\subseteq\igt(\Set)$ are closed under the rules
      of the proof system\ifproceedings\else\ (\cref{lem:respres})\fi. 
      Hence, by
    \cref{cor:liftgr}, $\Set^\flat\Respstr{\Ad^\flat}^*
    (\Set')^\flat\cup\{\bot\}$ and therefore by soundness (\cref{prop:soundness}), $\Set^\flat$ is $\Ad^\flat$-unsatisfiable.

    The converse can be derived similarly and we conclude:
    % Conversely, any refutation of $\widehat S$ (which exists by \cref{thm:compcomp} in case $\widehat S$ is $\widehat\Ad$-unsatisfiable) can be lifted to a refutation of $S$ (\cref{lem:respres,cor:liftgr2}), implying $\Ad$-unsatisfiability of $S$ (\cref{prop:soundness}).
  \begin{proposition}
    $\Set$ is $\Ad$-satisfiable iff $\Set^\flat$ is \mbox{$\Ad^\flat$-satisfiable}.
  \end{proposition}
  % \begin{proof}
  %   % First, suppose that $\widehat S$ is $\widehat\Ad$-unsatisfiable. By the Completeness \cref{thm:compcomp} and the fact that $\widehat\Ad$ is finite (hence compact), $\widehat S\Respstr{\widehat\Ad}^*\widehat S'\cup\{\bot\}$ for some $S'$. By \cref{lem:respres,cor:liftgr2}, $S\Respstr\Ad^* S'\cup\{\bot\}$ and by the Soundness \cref{prop:soundness}, $S$ is $\Ad$-unsatisfiable.
    
  %   First, suppose $S$ is $\Ad$-unsatisfiable. 
  %   % Define $\Phi(A)=1$ just if $\igt(A)\subseteq\igt(S)$ and $\neg A$ is a HoBHC(SLA).
  %   % simple.
  %   By \cref{prop:complHoBHCSLA} $S\Respstr\Ad^* S'\cup\{\bot\}$ for
  %   some $S'$. \changed[dw]{Note that sets $S'$ of HoBHC(SLA)
  %     satisfying $\igt(S')\subseteq\igt(S)$ are closed under the rules
  %     of the proof system.} Hence, by \cref{cor:liftgr1}, $\widehat S\Respstr{\widehat\Ad}^*\widehat S'\cup\{\bot\}$ and therefore by \cref{prop:soundness}, $\widehat S$ is $\widehat\Ad$-unsatisfiable.

  %   The converse is similar and omitted.
  %   % Conversely, any refutation of $\widehat S$ (which exists by \cref{thm:compcomp} in case $\widehat S$ is $\widehat\Ad$-unsatisfiable) can be lifted to a refutation of $S$ (\cref{lem:respres,cor:liftgr2}), implying $\Ad$-unsatisfiability of $S$ (\cref{prop:soundness}).
  % \end{proof}

  Finally, $\Ad^\flat$, which is finite, can be
  effectively obtained as a result of the decidability of Linear Integer Arithmetic (or \emph{Presburger arithmetic}) \cite{P29}. 
%\lo{Isn't Presburger arithmetic (appearing here for the first time) and LIA the same thing?} \dw{yes}
  Moreover, for every $\As^\flat\in\Ad^\flat$ and type $\sigma$,
  $\sinti{\As^\flat}\sigma$ is finite. Hence, by \cref{rem:dec}, we obtain:
  \begin{theorem}
    % \changed[lo]{
    Let $\Set$ be a finite set of HoBHC(SLA).
    % }\dw{isn't that already contained in \cref{ass:HoBHCSLA}?} \lo{Yes it is. My rationale is that important theorems should be stated in a self-standing way, so that it is fully intelligible when read in situ.}
    It is decidable if there is a $\Sigma'$-expansion $\Bs$ of $\As_{\LIA}$ satisfying $\Bs\models\Set$. 
  \end{theorem}

\section{Related Work}
\label{sec:relwork}
\paragraph{Higher-order Automated Theorem Proving}
\label{sec:holatp}
There is a long history of resolution-based procedures for
higher-order logic \emph{without} background theories which are refutationally complete for Henkin semantics e.g.~\cite{A71,H72,BK98,BBCW18}. 
Furthermore, a tableau-style proof system has been proposed \cite{B11}.
Their completeness proofs construct \emph{countable} Henkin models out of terms in case the proof system is unable to refute a problem. 
Hence, these proofs do not seem to be extendable to provide standard
models when restricted to HoCHCs.

%\changed[dw]{
Furthermore, there are efforts to extend SMT solvers to higher-order logic \cite{B17,B19} but the techniques seem to be incomplete even for Henkin semantics.
%}

\paragraph{Theorem Proving for 1st-order Logic Modulo Theories}
In the 1990s, superposition \cite{BG90}---the basis of most state-of-the-art theorem provers \cite{WDFKSW09,KV13}---was extended to a setting with background theories \cite{BGW94,AKW09}.
The proof system is sound and complete, assuming a compact background
theory and some technical conditions.
Abstractly, their proof system is very similar to ours: there is a clear separation between logical / foreground reasoning and reasoning in the background theory. 
Moreover, the search is directed purely by the former whilst the latter is only used in a final step to check satisfiability of a conjunction of theory atoms.

\paragraph{Defunctionalisation}
Our translation to 1st-order logic (\cref{SEC:APPENC}) resembles Reynolds' \emph{defunctionalisation} \cite{R72}. 
A whole-program transformation, defunctionalisation reduces higher-order functional programs to 1st-order ones. It eliminates higher-order features, such as partial applications and \mbox{$\lambda$-abstractions,} by storing arguments in data types and recovering them in an application function, which performs a matching on the data type.

Recently, the approach was adapted to the satisfiability problem
for HoCHC \cite{PRO18} and implemented in the tool
  \emph{DefMono}\footnote{see \url{http://mjolnir.cs.ox.ac.uk/dfhochc/}}: given a set of HoCHCs, it generates an equi-satisfiable set of 1st-order Horn clauses over the original background theory and additionally the theory of data types. %by taking a detour via programs and defunctionalisation.
By contrast, our translation is purely logical, directly yielding
1st-order Horn clauses, without recourse to inductive data types.

\paragraph{Extensional Higher-order Logic Programming}
\label{sec:extlp}

The aim of higher-order logic programming is not only to establish satisfiability of a set of Horn clauses without background theories but also to find (representatives of) ``answers to queries'', i.e.\ witnesses that goal clauses are falsified in every model of the definite clauses. 
Thus, \cite{CHRW13} proposes a rather complicated domain-theoretic semantics (equivalent to the continuous semantics \cite[Prop.~5.14]{CHRW13}).
They design a resolution-based proof system that supports a strong notion of completeness (\cite[Thm.~7.38]{CHRW13}) with respect to this semantics.

Their proof system is more complicated because it operates on more general formulas (which are nonetheless translatable to clauses). 
Moreover it requires the instantiation of variables with certain terms, which we avoid by implicitly instantiating all remaining relational variables with $\top_\rho$ in the constraint refutation rule.

\paragraph{Refinement Type Assignments}
\cite{BOR18} also introduces a refinement type system, the aim of
which is to automate the search for models. In this respect, the
approach is orthogonal to our resolution proof system, which can be
used to refute all unsatisfiable problems (but might fail on
satisfiable instances). However, for satisfiable clause sets the
method by \cite{BOR18}, which is implemented in the tool \emph{Horus}\footnote{see 
  \url{http://mjolnir.cs.ox.ac.uk/horus/}}, may also be unable to generate models.

\section{Conclusion and Future Directions}
\label{sec:conc}

In sum, HoCHC lies at a ``sweet spot'' in higher-order logic, semantically robust and useful for algorithmic verification.

\paragraph*{Future work} %it would be interesting to investigate  
%\changed[dw]{
  We
  expect that our proof system's robustness on \emph{satisfiable} instances can
  be improved by
tightening the rules (cf.~\cref{sec:relwork}) or
combining it with a search for models \cite{BOR18,UTK13,KSU11}.
Crucially, soundness and completeness even for standard semantics can be retained thanks to HoCHC's
semantic invariance. 
To facilitate comparison of approaches, it would also be important to obtain
an implementation of our techniques and conduct an empirical evaluation.

%   Besides, this paper has leveraging HoCHC's semantic invariance there is great
%   potential to obtain a proof system which is more useful in practice 

% As future work, it would be interesting to make resolution more
% robust, i.e. practically useful, on satisfiable instances by
% tightening proof rules (e.g.\ in a superposition-styled way) or to
% combinine it with a search for models \cite{BOR18,UTK13,KSU11}. Approaches we have
% in mind include tightening the proof rules (see discussion about
% review 1 \cref{sec:holatp}) and to combine resolution with the refinement type system [1].
  
  On the more theoretical side it would be interesting to identify
  extensions of HoCHC sharing its excellent properties.
\paragraph*{Acknowledgments}
We gratefully acknowledge support of EPSRC grants EP/N509711/1 and EP/M023974/1.



\bibliographystyle{IEEEtran}
\bibliography{lit}

\ifproceedings
\else
\appendix

\subsection{Supplementary Materials for \cref{sec:prelims}}
\label{sec:appprelims}
\subsubsection{Supplementary Materials for \cref{sec:relhol}}
The following lemma is completely standard and can be proven by a routine structural induction (exploiting the variable convention).
\begin{lemma}[Substitution]
  \label{lem:substitution}
  Let $\Hf$ be a pre-frame, $\As$ be a $(\Sigma,\Hf)$-structure and $\alpha$ be a $(\Delta,\Hf)$-valuation. Furthermore, let $x\in\dom(\Delta)$ and let $M$ and $N$ be terms such that $\Delta\vdash N\from\Delta(x)$.
  Then $\sinti\As {M[N/x]}(\alpha)=\sinti\As{M}(\alpha[x\mapsto\sinti\As N(\alpha)])$.
\end{lemma}
 The following lemma states that in frames, the denotation is stable under $\beta$- and $\eta$-conversion.

 \begin{lemma}
   \label{lem:betaeta}
  Let $\Hf$ be a frame, let $M$ and $M'$ be $\Sigma$-terms, $\As$ be a $(\Sigma,\Hf)$-structure and let $\alpha$ be a $(\Delta,\Hf)$-valuation. Then
  \begin{thmlist}
  \item\label{lem:betapreserve} if $M\bred M'$ then $\sinti\As M(\alpha)=\sinti\As{M'}(\alpha)$;
  \item\label{lem:etaequal} if $M\rightarrow_\eta M'$ then $\sinti\As M(\alpha)=\sinti\As{M'}(\alpha)$.
  \end{thmlist}
\end{lemma}
\begin{proof}
  We prove the lemma by induction on the compatible closure. The only interesting cases are the base cases $((\lambda x\ldotp N)N',N[N'/x])\in\beta$ and $(\lambda y\ldotp L\,y,L)\in\eta$, respectively. Then
  \begin{align*}
    \sinti\As{(\lambda x\ldotp N)N'}(\alpha)&=\sinti\As{N}(\alpha[x\mapsto\sinti\As{N'}(\alpha)])\\
                                    &=\sinti\As{N[N'/x]}(\alpha)
  \end{align*}
  because of the fact that $\Hf$ is a frame and
  \cref{lem:substitution}, and
  \begin{align*}
    \sinti\As{\lambda y\ldotp Ly}(\alpha)=\lambda r\in\sinti\Hf{\Delta(y)}\ldotp\sinti\As{L}(\alpha)(r)
                                         =\sinti\As{L}(\alpha)
  \end{align*}
  because $\Hf$ is a frame.
\end{proof}

\subsubsection{Supplementary Materials for \cref{sec:nftrans}}
\label{sec:appnftrans}

Let $\Set$ be a finite set of HoCHCs. W.l.o.g.\ we can
assume that for each $R\in\Sigma'\setminus\Sigma$ there is at least
one Horn clause $G\lor R\,\overline x_R$ and each definite clause has
this form.

For a goal clause $G=\neg A_1\lor\cdots\lor\neg A_n$ let $\posex(G,V)\defeq\exists y_1,\ldots,y_m\ldotp A_1\land\cdots\land A_n$, where $\{y_1,\ldots,y_m\}=\free(G)\setminus V$ and $\posex(G)\defeq\posex(G,\emptyset)$. Clearly, $\posex(G)$ is a positive existential closed formula.
Let $\Prgm_\Set$ be the set of definite formulas
\begin{align*}
  \neg (\posex(G_{R,1},\overline x_R)\lor\cdots\lor \posex(G_{R,n},\overline x_R))\lor R\,\overline x_R,
\end{align*}
where $G_{R,1}\lor R\,\overline x_R,\ldots,G_{R,n}\lor R\,\overline
x_R$ are the unnegated occurrences of $R\,\overline x_R$ in $\Set$ and
$R\in\Sigma'\setminus\Sigma$.
% and let
% \begin{align*}
%   F(S)=\posex(G_1)\lor\cdots\lor\posex(G_m),
% \end{align*}
% where $G_1,\ldots,G_m$ are the goal clauses in $F$.
Clearly, $\Prgm_\Set$ is a program.
Furthermore, the following is obvious by definition:
\begin{lemma}
  Let $\Set$ be a finite set of HoCHCs, $\Hf$ a frame and
  $\Bs$ a $(\Sigma',\Hf)$-expansion of $\As$.

  Then
  % \begin{thmlist}
  % \item
    % \label{lem:nfdefinite}
    $\Bs\models\{D\in\Set\mid D\text{ definite}\}$ iff $\Bs\models \Prgm_\Set$.
  % \item $\Bs\models_\Hf\{G\in S\mid G\text{ goal}\}$ iff $\Bs\not\models_\Hf F(S)$.
  % \end{thmlist}
  
  \end{lemma}

  \subsubsection{$\Sf$, $\Mf$ and $\Cf$ are Complete Frames}
  \label{sec:appsmcHenkin}
  First, note that $\arel$ (as defined in \cref{def:arelc}) and $\prel$ coincide for the continuous frame.
\begin{lemma}
  \label{cl:cprel}
  Let $\prel$ be the pointwise ordering on $\Cf$
  and $\arel$ be the relation defined in \cref{def:arelc}. Then
  ${\arel}={\prel}$.
\end{lemma}
\begin{proof}
  We prove by induction on $\sigma$ that ${\arel_{\sigma}}={\prel_\sigma}$. For $\iota$ and $o$ this is obvious. Hence, suppose $\sigma=\tau\to\sigma'$ and let $r,r'\in\sinti\Cf{\tau\to\sigma'}$.

  First, suppose $r\arel r'$. Let $s\in\sinti\Cf\tau$ be
  arbitrary. By the inductive hypothesis, $s\arel s$ ($\prel$ is
  reflexive). Thus, $\{s\}\in\dir_\arel(s)$ and therefore, $r(s)\arel
  r'(s)$. Again by the inductive hypothesis, $r(s)\prel
  r'(s)$. Consequently, $r\prel r'$.
  
  Conversely, suppose $r\prel r'$. Let $s\in\sinti\Cf\tau$ and $\Sd'\in\dir_\arel(s)$ be arbitrary. By the inductive hypothesis, $s\prel\bigsqcup\Sd'$ and $\Sd'$ is $\prel$-directed. Hence,
  \begin{align*}
    r(s)\prel r\left(\bigsqcup\Sd'\right)
    \prel r'\left(\bigsqcup\Sd'\right)
    =\bigsqcup_{s'\in\Sd'}r'(s')
  \end{align*}
  exploiting the monotonicity of $r$, the continuity of $r'$ and the fact that $r\prel r'$.
  Again by the inductive hypothesis, $r(s)\arel\bigsqcup_{s'\in\Sd'}r'(s')$.  Therefore, $r\arel r'$.
\end{proof}

\begin{lemma}
  \label{lem:conhenkin}
  Let $\Sigma$ be a signature, $\Delta$ be a type environment and $\Bs$ be a $(\Sigma,\Cf)$-structure. 
  Then for any positive existential term $M$, $(\Delta,\Cf)$-valuation $\alpha$ and $\vals'\in\dir_\arel(\alpha)$,
  \begin{thmlist}
  \item\label{it:conhenkin1} if $M$ is a $\lambda$-abstraction
    $\Delta\vdash\lambda x\ldotp M'\from\tau\to\rho$ then for every
    $s\in\sinti\Bs\tau$, $\sinti\Bs
    M(\alpha)(s)=\sinti\Bs{M'}(\alpha[x\mapsto s])$, and
  \item\label{it:conhenkin2} $\sinti\Bs M(\alpha)\arel\bigsqcup_{\alpha'\in\vals'}\sinti\Bs M(\alpha')$.
  \end{thmlist}
\end{lemma}
\begin{proof}
  We prove both parts of the lemma simultaneously by induction on the structure of $M$. For all cases except $\lambda$-abstractions, \itemref{it:conhenkin1} is trivially true and \itemref{it:conhenkin2} is proven as in \cref{lem:contmain}.

  Hence, suppose $M$ is a $\lambda$-abstraction $\Delta\vdash\lambda
  x\ldotp M'\from\tau\to\rho$.
  We define $r\defeq \lambda s\in\sinti\Bs\tau\ldotp\sinti\Bs{M'}(\alpha[x\mapsto s])$
  \begin{claim}
    \label{cl:lambdacont}
    $r\in\sinti\Bs{\tau\to\rho}=\sinti\Cf{\tau\to\rho}$.
  \end{claim}
  \begin{claimproof}
    First, let $s,s'\in\sinti\Bs\tau$ be such that $s\prel s'$. By the reflexivity of $\prel$ and \cref{cl:cprel}, $\{\alpha[x\mapsto s']\}\in\dir_\arel(\alpha[x\mapsto s])$. Hence, by the inductive hypothesis,
    \begin{align*}
      r(s)&=\sinti\Bs{M'}(\alpha[x\mapsto s])
      \arel\sinti\Bs{M'}(\alpha[x\mapsto s'])
            =r(s')
    \end{align*}
    Consequently by \cref{cl:cprel}, $r(s)\prel r(s')$ and $r$ is monotone.

    Next, suppose that $\Sd\subseteq\sinti\Bs\tau$ is
    $\prel$-directed. Note that by \cref{cl:cprel} and the reflexivity
    of $\prel$, $\{\alpha[x\mapsto s]\mid s\in\Sd\}\in\dir_\arel(\alpha[x\mapsto\bigsqcup\Sd])$. Therefore, by the inductive hypothesis
    \begin{align*}
      r\left(\bigsqcup\Sd\right)&=\sinti\Bs{M'}\left(\alpha\left[x\mapsto\bigsqcup\Sd\right]\right)\\
                                                          &\arel\bigsqcup_{s\in\Sd}\sinti\Bs{M'}(\alpha[x\mapsto s])
                                                            =\bigsqcup_{s\in\Sd}r(s)
    \end{align*}
    Again by \cref{cl:cprel}, $r(\bigsqcup\Sd)\prel\bigsqcup_{s\in\Sd}r(s)$.
    Furthermore, by monotonicity of $r$, for
    every $s\in\Sd$, $r(s)\prel r(\bigsqcup\Sd)$. % Due to the fact that this holds for every $s\in\Sd$ and $\prel^\Cf$ induces a lattice structure,
    % \begin{align*}
    %   \bigsqcup\{\Bs\langle M\rangle(\alpha)(s)\mid s\in\Sd\}\prel^\Cf\Bs\langle M\rangle(\alpha)\left(\bigsqcup\Sd\right).
    % \end{align*}
    Consequently, by the antisymmetry of $\prel$, $r(\bigsqcup\Sd)=\bigsqcup_{s\in\Sd}r(s)$.

    This concludes the proof of the claim that $r:\sinti\Bs\tau\to\sinti\Bs\rho$ is continuous.
  \end{claimproof}
  As a consequence of \cref{cl:lambdacont}, for every $s\in\sinti\Bs\tau$, $\sinti\Bs
    M(\alpha)(s)=\sinti\Bs{M'}(\alpha[x\mapsto s])$ and the same argument as in the proof of
  \cref{lem:contmain} can be used to demonstrate
  \itemref{it:conhenkin2} of the lemma.
\end{proof}

Similarly, we get the following for the monotone frame
\begin{lemma}
  \label{lem:monhenkin}
  \begin{thmlist}
  \item\label{lem:moncoin} ${\arel}={\prel}$, where $\prel$ be the pointwise ordering on $\Mf$
  and $\arel$ is the relation defined in \cref{def:arelm};
  \item If $\Sigma$ is a signature, $\Delta$ is a type environment,
    $\Bs$ is a $(\Sigma,\Mf)$-structure, $\alpha$ is a
    $(\Delta,\Mf)$-valuation, $\Delta\vdash\lambda x\ldotp
    M\from\tau\to\rho$ is a positive existential $\Sigma$-term and
    $s\in\sinti\Bs\tau$ then
    \begin{align*}
      \sinti\Bs M(\alpha)(s)=\sinti\Bs{M'}(\alpha[x\mapsto s]).
    \end{align*}
  \end{thmlist}
\end{lemma}
\cref{lem:monhenkin,lem:conhenkin} immediately imply the following
(completeness is trivial):
\begin{proposition}
  \label{prop:Henkin}
  $\Sf$, $\Mf$ and $\Cf$ are complete frames.
\end{proposition}

\subsection{Supplementary Materials for \cref{sec:quasi-mon}}
\knastertarski*
\begin{proof}
  Since, $L$ is a set (in contrast to $\On$) there exists
  $\beta\in\On$ satisfying $a_F=a_\beta$.
  \begin{thmlist}
  \item Thus,
  $F(a_F)=a_{\beta+1}\leq a_F$.
  \item By what we have just shown it suffices to prove by transfinite
    induction on $\beta$ that $a_\beta\arel b$.
    \begin{itemize}[noitemsep]
    \item First, suppose $\beta=\widetilde\beta+1$ is a successor
      ordinal. By the inductive hypothesis and quasi-monotonicity of $F$,
      $a_\beta=F(a_{\widetilde\beta})\arel F(b)\leq b$. Consequently by
      \ref{def:rel1}, $a_\beta\arel b$.
    \item Otherwise, $\beta$ is a limit ordinal and by the inductive
      hypothesis, $a_{\widetilde\beta}\arel b$ for all
      $\widetilde\beta<\beta$. Consequently, by \ref{def:rel2},
      $a_\beta\arel b$.
    \end{itemize}  
  \end{thmlist}  
\end{proof}

\begin{lemma}
  \label{cor:relm}
  $\arel$ \changed[dw]{is compatible with $\prel$}.
\end{lemma}
\begin{proof}
  We prove by induction on the type $\sigma$ that $\arel_\sigma$ is
  compatible with $\prel_\sigma$. For $o$ and $\iota$
  this is obvious. Hence, suppose that $\sigma=\tau\to\sigma'$.
  \begin{enumerate}[noitemsep]
  \item Let $r,r',r''\in\sinti\Hf\sigma$ be such that $r\arel_\sigma
    r'\prel_\sigma r''$. Besides, let $s,s'\in\sinti\Hf\tau$ be such that $s\arel_\tau s'$. Clearly, it holds
    that $r(s)\arel_{\sigma'} r'(s')\prel_{\sigma'} r''(s')$ and by the inductive
    hypothesis, $r(s)\arel_{\sigma'} r''(s')$. Hence, $r\arel_\sigma
    r''$.
  \item Let $r'\in\sinti\Hf\sigma$ and
    $\Rd\subseteq\{r\in\sinti\Hf\sigma\mid r\arel_\sigma r'\}$ be
    arbitrary. Suppose $s,s'\in\sinti\Hf\tau$ are such that $s\arel_\tau s'$. By definition, $\{r(s)\mid
  r\in\Rd\}\subseteq\{t\in\sinti\Hf{\sigma'}\mid t\arel_{\sigma'}
  r'(s')\}$. Therefore, by the inductive hypothesis,
  $\left(\bigsqcup\Rd\right)(s)=\bigsqcup_{r\in\Rd}r(s)\arel_{\sigma'}
  r'(s')$. Consequently, $\bigsqcup\Rd\arel_\sigma r'$.\qedhere
  \end{enumerate}
\end{proof}

\termmon*
\begin{proof}
  We prove the claim by induction on the structure of $M$.
  \begin{itemize}[noitemsep]
  \item If $M$ is a variable $x$ then $\sinti\Bs M(\alpha)=\alpha(x)\arel\alpha'(x)=\sinti{\Bs'}M(\alpha')$ because of $\alpha\arel\alpha'$.
  \item If $M$ is a logical symbol (other than $\neg$) then this is a consequence of \cref{ex:qmlog2}. 
  \item If $M$ is a symbol $R\in\Sigma'$ then $\sinti\Bs M(\alpha)=R^\Bs\arel R^{\Bs'}=\sinti{\Bs'}M(\alpha')$ because of $\Bs\arel\Bs'$.
  \item If $M$ is an application $N\,N'$ then by the inductive
    hypothesis $\sinti\Bs N(\alpha)\arel\sinti{\Bs'}N(\alpha')$ and
    $\sinti\Bs{N'}(\alpha)\arel\sinti{\Bs'}{N'}(\alpha')$. Therefore, by definition of $\arel$,
    \begin{align*}
      \sinti\Bs{M}(\alpha)&=\sinti\Bs
                            N(\alpha)(\sinti\Bs{N'}(\alpha))\\
                          &\arel \sinti{\Bs'}
                            N(\alpha')(\sinti{\Bs'}{N'}(\alpha'))
                          =\sinti{\Bs'}{M}(\alpha')
    \end{align*}
  \item Finally, suppose $M$ is an abstraction $\lambda x\ldotp N$. Let $s\arel s'$. 
    By the inductive hypothesis $\sinti\Bs N(\alpha[x\mapsto s])\arel\sinti{\Bs'}N(\alpha'[x\mapsto s'])$ and hence,
    \begin{align*}
      \sinti\Bs M(\alpha)(s)&=\sinti\Bs N(\alpha[x\mapsto s])\\
                            &\arel\sinti{\Bs'}N(\alpha'[x\mapsto
                              s'])
                            =\sinti{\Bs'}M(\alpha)(s')
    \end{align*}
    because $\Hf$ is a frame.
    Due to the fact that this holds for every $s\arel s'$, $\sinti\Bs M(\alpha)\arel \sinti{\Bs'}M(\alpha')$.\qedhere
  \end{itemize}
\end{proof}

\subsection{Supplementary Materials for \cref{sec:resolution}}
\begin{lemma}
  \label{lem:ignoresimple}
  Let $\phi_1,\ldots,\phi_n$ be background atoms.
  Then $\neg x_1\,\overline M_1\lor\neg x_m\,\overline M_m\lor\neg\phi_1\lor\cdots\lor\neg\phi_n\models\neg\phi_1\lor\cdots\lor\neg\phi_n$.
\end{lemma}
\begin{proof}
  Let $\Hf$ be a frame and suppose $\Bs$ is an arbitrary $(\Sigma',\Hf)$-structure satisfying
  \begin{align}
    \label{eq:ignsimple}
    \Bs\models \neg x_1\,\overline M_1\lor\neg x_m\,\overline M_m\lor\neg\phi_1\lor\cdots\lor\neg\phi_n.
  \end{align}
  Let $\alpha$ be an arbitrary $(\Delta,\Hf)$-valuation. We define another $(\Delta,\Hf)$-valuation $\alpha'$ by
  \begin{align*}
    \alpha'(x)=
    \begin{cases}
      \top_\rho&\text{if }\Delta(x)=\rho\\
      \alpha(x)&\text{otherwise ($\Delta(x)=\iota$)}
    \end{cases}
  \end{align*}
 Clearly, $\Bs,\alpha'\not\models \neg x_1\,\overline M_1\lor\neg x_m\,\overline M_m$. Hence, by \cref{eq:ignsimple}, $\Bs,\alpha'\models\neg\phi_i$ for some $i$. Note that by \cref{rem:isimple}, $\phi_i$ only contains variables of type $\iota$. Hence, also $\Bs,\alpha\models\neg\phi_i$. This proves, $\Bs\models \neg\phi_1\lor\cdots\lor\neg\phi_n$.
\end{proof}

Now, \cref{prop:soundness} is a simple consequence of the following:
\begin{lemma}
  Let $\Set'$ be a set of HoCHCs and suppose that $\Set'\Res \Set'\cup\{C\}$. Then
  \begin{thmlist}
  \item\label{it:sound1} if $C\neq\bot$ then $\Set'\models C$;
  \item if $\Bs$ is an expansion of $\As$ and $\Bs\models \Set'$ then $\Bs\models C$.
  \end{thmlist}
\end{lemma}
\begin{proof}
  \begin{enumerate}
  \item Note that by assumption the constraint refutation rule cannot have been applied. Besides, for $\beta$-reduction this is a consequence of \cref{lem:betapreserve}. Finally, suppose that $\neg R\,\overline M\lor G$ and $G'\lor R\,\overline x$ are in $\Set'$ (modulo renaming of variables). The proof for this case uses the same ideas as the classic one for 1st-order logic (see e.g.\ \cite{R65,F96}):

    % Let $\Hf$ be a Henkin frame, $\Bs$ be an arbitrary structure satisfying $\Bs\models_\Hf S'$ and let $\alpha$ be a valuation. Furthermore assume that $\Bs,\alpha\not\models G$. Due to $\Bs\models_\Hf S'$ it must hold that $\Bs,\alpha\not\models R\,\overline M$. Hence, $\Bs,\alpha[\overline x\mapsto\sinti\Bs{\overline M}(\alpha)]\not\models R\,\overline x$ and % therefore due to
    % since $\Bs\models_\Hf S'$, $\Bs,\alpha[\overline x\mapsto\sinti\Bs{\overline M}(\alpha)]\models G'$, which by the Substitution \cref{lem:substitution} implies $\Bs,\alpha\models G'[\overline M/\overline x]$. Consequently, $\Bs,\alpha\models G\lor G'[\overline M/\overline x]$.
  \item Let $\Bs$ be an expansion of $\As$ satisfying $\Bs\models
    \Set'$. By \itemref{it:sound1} it suffices to consider the case when
    the constraint refutation rule is applicable to some goal clause
    $\neg x_1\,\overline M_1\lor\neg x_m\,\overline
    M_m\lor\neg\phi_1\lor\cdots\lor\neg\phi_n$, where each $\phi_i$ is a background atom and there exists a valuation $\alpha$ such that $\As,\alpha\models\phi_1\land\cdots\land\phi_n$. However, by \cref{lem:ignoresimple}, $\Bs,\alpha\models\neg\phi_1\lor\cdots\lor\neg\phi_n$, which is clearly a contradiction to the fact that $\Bs$ is an expansion of $\As$.\qedhere
  \end{enumerate}
\end{proof}

% \subsubsection{Supplementary Materials for \cref{sec:canmod}}
% \propsasp*
% \begin{proof}
%   \begin{enumerate}
%   \item First note that $\Bd=\{\As^\Hf_\beta\mid\beta\in\On\}$ is a set. For $\Bs\in\Bd$ let $\beta_\Bs$ be the minimal ordinal such that $\Bs=\As^\Hf_\beta$. By the replacement axiom, $\mathfrak O=\{\beta_\Bs\mid\Bs\in\Bd\}$ is a set of ordinals. Hence, $\bigcup\mathfrak O$ is an ordinal. Let $\gamma\geq\bigcup\mathfrak O$ be a limit ordinal. Then it holds
%     \begin{align*}
%       \As^\Hf_\gamma=\bigsqcup\{\As^\Hf_\beta\mid\beta\in\mathfrak O\}=\bigsqcup\Bd=\As^\Hf_P.
%     \end{align*}
%   \item Assume towards contradiction that $\As^\Hf_P\not\models P$. Then there exists $\neg F_R\lor R\,\overline x_R\in P$ and $\overline s$ satisfying $\As^\Hf_P,\top_\Delta^\Hf[\overline x_R\mapsto\overline s]\not\models\neg F_R\lor R\,\overline x_R$ (because $\free(F_R)=\overline x_R$). By \itemref{lem:plateau}, there is an ordinal $\beta$ such that $\As^\Hf_P=\As^\Hf_\beta$. Hence, $\As^\Hf_\beta,\top_\Delta^\Hf[\overline x_R\mapsto\overline s]\models F_R$ and due to the fact that $\Hf$ is a Henkin frame, $\sinti{\As^\Hf_\beta}{\lambda\overline x_R\ldotp F_R}(\top_\Delta^\Hf)(\overline s)=1$.
%     Consequently, $R^{T_P^\Hf(\As^\Hf_\beta)}(\overline s)=1$, which implies $R^{\As^\Hf_P}(\overline s)=1$. Clearly, this is a contradiction to $\As^\Hf_P,\top_\Delta^\Hf[\overline x_R\mapsto\overline s]\not\models\neg F_R\lor R\,\overline x_R$.\qedhere
%   \end{enumerate}
% \end{proof}

\subsubsection{Supplementary Materials for \cref{subsec:quasicont}}
\label{sec:suppres}

\begin{lemma}
  \label{lem:arelrefl}
  If $\arel$ is compatible with $\leq$ and $F$ is quasi-continuous
  then for all ordinals $\beta$, $a_\beta\arel a_\beta$.
\end{lemma}
\begin{proof}
  \begin{itemize}[noitemsep]
  \item If $\beta=\widetilde\beta+1$ is a successor ordinal then by the
    inductive hypothesis, $\{a_{\widetilde\beta}\}\in\dir(a_{\widetilde\beta})$ and
    by quasi-continuity, $a_\beta=F(a_{\widetilde\beta})\arel
    F(a_{\widetilde\beta})=a_\beta$.
  \item If $\beta$ is a limit ordinal then by the inductive
    hypothesis for all $\widetilde\beta<\beta$, $a_{\widetilde\beta}\arel
    a_{\widetilde\beta}$, and by definition, $a_{\widetilde\beta}\leq a_\beta$. By
    \ref{def:rel1}, $a_{\widetilde\beta}\arel a_\beta$ (for all
    $\widetilde\beta<\beta$) and thus by \ref{def:rel2}, $a_\beta\arel a_\beta$.\qedhere 
  \end{itemize}
\end{proof}
  
% \begin{lemma}
%   \label{lem:amon}
%   then for all ordinals $\beta\leq\beta'$, $a_\beta\arel a_{\beta'}$.
% \end{lemma}
% \begin{proof}
%   First, we prove the following by transfinite induction:
% \end{proof}
\kleene*
\begin{proof}
  \begin{thmlist}
  \item We prove by transfinite induction on $\beta$
    that for all $\beta'\geq\beta$, $a_\beta\arel a_{\beta'}$.
    \begin{itemize}[noitemsep]
    \item First, suppose $\beta=\widetilde\beta+1\leq\beta'$ is a
      successor ordinal.
      If $\beta'=\beta''+1$ is a successor ordinal, too,
      then $\widetilde\beta\leq\beta''$. Therefore, by the
      inductive hypothesis and \cref{lem:arelrefl},
      $\{a_{\beta''}\}\in\dir(a_{\widetilde\beta})$. By
      quasi-continuity of $F$,
      $a_\beta=F(a_{\widetilde\beta})\arel
      F(a_{\beta''})=a_{\beta'}$.

      Otherwise, $\beta'$ is a limit ordinal and $\beta<\beta'$. By
      \cref{lem:arelrefl} and definition, $a_\beta\arel a_\beta\leq
      a_{\beta'}$. Hence, by \ref{def:rel1}, $a_\beta\arel a_{\beta'}$.
    \item Finally, suppose $\beta$ is a limit ordinal. By the
      inductive hypothesis for every $\widetilde\beta<\beta$,
      $a_{\widetilde\beta}\arel a_{\beta'}$ and thus by
      \ref{def:rel2}, $a_\beta\arel a_{\beta'}$.
    \end{itemize}
  \item Next, we prove by transfinite induction that for every ordinal $\beta$,
    $a_\beta\arel a_\omega$. Then the claim follows from
  \ref{def:rel2}.
  \begin{itemize}
  \item First, suppose that $\beta=\beta'+1$ is a successor
    ordinal. By the inductive hypothesis, $a_{\beta'}\arel
    a_\omega$. Therefore (also using the first part), $\{a_n\mid n\in\omega\}\in\dir_\arel(a_{\beta'})$. Thus, by
    quasi-continuity, $a_\beta=F(a_{\beta'})\arel\biglor_{n\in\omega} F(a_n)=a_\omega$.
  \item Next, suppose that $\beta$ is a limit ordinal. By the
    inductive hypothesis, $a_{\beta'}\arel a_\omega$ for
    every $\beta'<\beta$. Therefore, by \ref{def:rel2},
    $a_\beta\arel a_\omega$.\qedhere
  \end{itemize}
  \end{thmlist}
\end{proof}

\begin{lemma}
  \label{cor:relc}
  $\arel$ \changed[dw]{is compatible with $\prel$}.
\end{lemma}
\begin{proof}
  We prove by induction on the type $\sigma$ that $\arel_\sigma$ is
  compatible with $\prel_\sigma$. For $o$ and $\iota$
  this is obvious. Hence, suppose that $\sigma=\tau\to\sigma'$.
  \begin{enumerate}[noitemsep]
  \item Let $r,r',r''\in\sinti\Hf\sigma$ be such that
  $r\arel_\sigma r'\prel_\sigma r''$. Besides let
  $s\in\sinti\Hf\tau$ and $\Sd'\in\dir_\arel(s)$ be arbitrary. Clearly, it holds
  that $r(s)\arel_{\sigma'}\bigsqcup_{s'\in\Sd'} r'(s')\prel_{\sigma'}\bigsqcup_{s'\in\Sd'} r''(s')$. Hence, by the inductive
  hypothesis, $r(s)\arel_{\sigma'}\bigsqcup_{s'\in\Sd'}
  r''(s')$. Consequently, $r\arel_\sigma r''$.
  \item Let $r'\in\sinti\Hf\sigma$ and
    $\Rd\subseteq\{r\in\sinti\Hf\sigma\mid r\arel_\sigma r'\}$ be
    arbitrary. Besides, suppose $s\in\sinti\Hf\tau$ and
    $\Sd'\in\dir_\arel(s)$. Clearly, $\{r(s)\mid
    r\in\Rd\}\subseteq\{t\in\sinti\Hf{\sigma'}\mid
    t\arel_{\sigma'}\bigsqcup_{s'\in\Sd'}r'(s')\}$ and hence by
    the inductive hypothesis,
    $\left(\bigsqcup\Rd\right)(s)=\bigsqcup_{r\in\Rd}r(s)\arel_{\sigma'}\bigsqcup_{s'\in\Sd'}r'(s')$. Consequently,
    $\bigsqcup\Rd\arel_\sigma r'$.\qedhere
  \end{enumerate}
\end{proof}

The relation $\arel$ is transitive but neither reflexive
(\cref{ex:deltanotcont}) nor antisymmetric and coincides with the pointwise ordering $\prel$ on the continuous
frame $\Cf$ (\cref{cl:cprel}).
\begin{example}
  \begin{thmlist}
  \item\label{ex:qclog1} For all relational types $\rho$ and
    $s\in\sinti\Hf\rho$, $\bot_\rho\arel s\arel\top_\rho$.
  \item\label{ex:qclog11} $\orf\arel\orf$ and $\andf\arel\andf$.
  \item\label{ex:qclog2} Next, suppose $r\in\sinti \Hf{\tau\to o}$,
    $\Rd'\in\dir(r)$ and $s\in\sinti \Hf\tau$ are such that
    $r(s)=1$. If $\tau=\iota$ then $r'(s)=1$ for all
    $r'\in\Rd'$. Otherwise, there exists $r'\in\Rd'$ satisfying
    $r'(\top_\tau)=1$ because $\{\top_\tau\}\in\dir(s)$. Consequently,
    $\hexists_\tau\arel \hexists_\tau$ holds as well.
  \item\label{ex:deltanotcont} $\delta_\omega\arel \delta_\omega$ does
    \emph{not hold} (see \cref{ex:notstopomega}): clearly $\{r_n\mid n\in\nat\}\in\dir(r_\omega)$ but $\delta_\omega(r_\omega)=1>0=\max\{\delta_\omega(r_n)\mid n\in\nat\}$. This shows that $\arel$ is not reflexive.
  \end{thmlist}
\end{example}

\contmain*
\begin{proof}
  We prove that for all positive existential terms $M$, expansions $\Bs$ of $\As$, $\Bd'\in\dir(\Bs)$, valuations $\alpha$ and $\vals'\in\dir(\alpha)$ \cref{eq:ncontmain} holds by induction on the structure of $M$.
  \begin{itemize}
  \item If $M$ is a logical constant (other than $\neg$) then this is due to \cref{ex:qclog11,ex:qclog2}.
  \item If $M$ is a symbol in $\Sigma'$ or a variable then this is by assumption.
  %   \begin{align*}
  %     \sinti\Bs
  %     M(\alpha)&=R^\Bs\arel\bigsqcup\{R^{\Bs'}\mid\Bs'\in\Bd'\}\\
  %              &=\bigsqcup\{\sinti{\Bs'}M(\alpha')\mid\Bs'\in\Bd'\land\alpha'\in\Ad'\}
  %   \end{align*}
  %   because $\Bd'\in\dir(\Bs)$ and clearly the expression on the right-hand side is well-defined.
  % \item If $M$ is a variable $x$ then
  %   \begin{align*}
  %     \sinti\Bs
  %     M(\alpha)&=\alpha(x)\arel\bigsqcup\{\alpha'(x)\mid\alpha'\in\Ad'\}\\
  %              &=\bigsqcup\{\sinti{\Bs'} M(\alpha')\mid\Bs'\in\Bd'\land\alpha'\in\Ad'\}
  %   \end{align*}
  %   because $\Ad'\in\dir(\alpha)$ and clearly the expression on the right-hand side is well-defined.
  \item Next, suppose $M$ is an application $M_1\,M_2$.
    By the inductive hypothesis,
    \begin{align}      
      \label{eq:ncontih1}
      \sinti\Bs{M_1}(\alpha)&\arel\bigsqcup_{\Bs_1\in\Bd',\alpha_1\in\vals'}\sinti{\Bs_1}{M_1}(\alpha_1)\\
      \label{eq:ncontih2}
      \sinti\Bs{M_2}(\alpha)&\arel\bigsqcup_{\Bs_2\in\Bd',\alpha_2\in\vals'}\sinti{\Bs_2}{M_2}(\alpha_2).
    \end{align}
    Let $s\defeq\sinti\Bs{M_2}(\alpha)$ and $\Sd'\defeq\{\sinti{\Bs_2}{M_2}(\alpha_2)\mid\Bs_2\in\Bd'\land\alpha_2\in\vals'\}$.
    \begin{claim}
      \label{cl:appdir}
      $\Sd'\in\dir(s)$.
    \end{claim}
    \begin{claimproof}
      By \cref{eq:ncontih2}, $s\arel\bigsqcup\Sd'$. To prove that
      $\Sd'$ is directed, let $\Bs^{(1)},\Bs^{(2)}\in\Bd'$ and
      $\alpha^{(1)},\alpha^{(2)}\in\vals'$. Since $\Bd'$ and $\vals'$
      are directed, $\{\Bs^{(1)}\}\in\dir(\Bs^{(1)})$,
      $\{\alpha^{(1)}\}\in\dir(\alpha^{(1)})$ and there are
      $\Bs'\in\Bd'$ and $\alpha'\in\vals'$ such that 
      $\{\Bs'\}\in\dir(\Bs^{(1)})\cap\dir(\Bs^{(2)})$ and
      $\{\alpha'\}\in\dir(\alpha^{(1)})\cap\dir(\alpha^{(2)})$. Hence,
      by the inductive hypothesis,
      $\sinti{\Bs^{(1)}}{M_2}(\alpha^{(1)})\arel
      \sinti{\Bs^{(1)}}{M_2}(\alpha^{(1)})$ and for $j\in\{1,2\}$, $\sinti{\Bs^{(j)}}{M_2}(\alpha^{(j)})\arel\sinti{\Bs'}{M_2}(\alpha')$. 
    \end{claimproof}

    Next, we define 
    \begin{align*}
      \Td&\defeq\{\sinti{\Bs_1}{M_1}(\alpha_1)(\sinti{\Bs_2}{M_2}(\alpha_2))\\
      &\hspace{3cm}\mid\Bs_1,\Bs_2\in\Bd'\land\alpha_1,\alpha_2\in\vals'\}\\
      \Td'&\defeq\{\sinti{\Bs'}{M_1}(\alpha')(\sinti{\Bs'}{M_2}(\alpha'))\mid\Bs'\in\Bd'\land\alpha'\in\vals'\}.
    \end{align*}
    \vspace{-0.3cm}
    \begin{claim}
      \label{cl:apppreceq}$\bigsqcup\Td\arel\bigsqcup\Td'$.
    \end{claim}
    \begin{claimproof}
      It suffices to prove that for every $t\in\Td$ there exists
      $t'\in\Td'$ satisfying $t\arel t'$. Then the claim is a
      % n immediate
      consequence of
      \ref{def:rel1} and \ref{def:rel2}.
      
     Hence, let $\Bs_1,\Bs_2\in\Bd'$ and $\alpha_1,\alpha_2\in\vals'$ be arbitrary. By directedness of $\Bd'$ and $\vals'$, there are $\Bs'\in\Bd'$ and $\alpha'\in\vals'$ such that % $\Bs_j\arel\Bs'$ and $\alpha_j\arel\alpha'$ for $j\in\{1,2\}$. 
      % Note that
      $\{\Bs'\}\in\dir(\Bs_1)\cap\dir(\Bs_2)$ and $\{\alpha'\}\in\dir(\alpha_1)\cap\dir(\alpha_2)$. Therefore, again by the inductive hypothesis,
    \begin{align}
      \label{eq:nclaimmainih1}
      \sinti{\Bs_1}{M_1}(\alpha_1)&\arel\sinti{\Bs'}{M_1}(\alpha')\\
      \label{eq:nclaimmainih2}
      \sinti{\Bs_2}{M_2}(\alpha_2)&\arel\sinti{\Bs'}{M_2}(\alpha').
    \end{align}
    Furthermore, due to $\{\Bs'\}\in\dir(\Bs')$,
    $\{\alpha'\}\in\dir(\alpha')$ and the inductive hypothesis, $\sinti{\Bs'}{M_2}(\alpha')\arel \sinti{\Bs'}{M_2}(\alpha')$.
    Hence (by \cref{eq:nclaimmainih2}), $\{\sinti{\Bs'}{M_2}(\alpha')\}\in\dir(\sinti{\Bs_2}{M_2}(\alpha_2))$. Therefore, by \cref{eq:nclaimmainih1},
    \begin{align*}
      &\sinti{\Bs_1}{M_1}(\alpha_1)(\sinti{\Bs_2}{M_2}(\alpha_2))\\
      \arel\; & \sinti{\Bs'}{M_1}(\alpha')(\sinti{\Bs'}{M_2}(\alpha'))\in\Td'.\qedhere
    \end{align*}
    % Consequently, by \cref{def:rel1},
    % $\sinti{\Bs_1}{M_1}(\alpha_1)(\sinti{\Bs_2}{M_2}(\alpha_2))\arel\bigsqcup\Td'$.
    % Since, this holds for every $\Bs_1,\Bs_2\in\Bd'$ and
    % $\alpha_1,\alpha_2\in\Ad'$, by \cref{def:rel2}, $\bigsqcup\Td\arel\bigsqcup\Td'$.
    \end{claimproof}
    Combining everything (\cref{cl:appdir,cl:apppreceq}, \cref{eq:ncontih1} and \cite[Prop.~2.1.4]{AJ95}), we obtain
    \begin{align*}
      \sinti\Bs M(\alpha)
      &=\sinti\Bs{M_1}(\alpha)(s)\\
      &\arel\bigsqcup_{s'\in\Sd'}\bigg(\bigsqcup_{\Bs_1\in\Bd', \alpha_1\in\vals'}\{\sinti{\Bs_1}{M_1}(\alpha_1)\bigg)(s')\\
      &=\bigsqcup\Td
      \arel\bigsqcup\Td'
      =\bigsqcup_{\Bs'\in\Bd',\alpha'\in\vals'}\sinti{\Bs'}{M}(\alpha').
    \end{align*}
    This concludes the proof of
    $\sinti\Bs{M_1M_2}(\alpha)\arel\bigsqcup_{\Bs'\in\Bd',\alpha'\in\vals'}\sinti{\Bs'}{M_1M_2}(\alpha')$
    (because $\arel$ is transitive).
  \item Finally, suppose $M$ is $\lambda x\ldotp M'$.
    Let $s\in\sinti \Hf{\Delta(x)}$ and $\Sd'\in\dir(s)$.
    Note that $\{\alpha'[x\mapsto s']\mid\alpha'\in\vals'\land s'\in\Sd'\}\in\dir(\alpha[x\mapsto s])$. Therefore by the inductive hypothesis,
    \begin{equation*}
      \label{eq:ncontmainlambda}
      \sinti\Bs{M'}(\alpha[x\mapsto
      s])\arel\bigsqcup_{\substack{\Bs'\in\Bd',\alpha'\in\vals', \\s'\in\Sd'}}\sinti{\Bs'}{M'}(\alpha'[x\mapsto s']).
    \end{equation*}
    Consequently,
    \begin{align*}
      \sinti\Bs M(\alpha)(s)
      &=\sinti\Bs{M'}(\alpha[x\mapsto s])\\
      &\arel\bigsqcup_{\substack{\Bs'\in\Bd',\alpha'\in\vals',\\s'\in\Sd'}}\sinti{\Bs'}{M'}(\alpha'[x\mapsto s'])\\
      &=\bigsqcup_{\substack{\Bs'\in\Bd', \alpha'\in\vals',\\s'\in\Sd'}}\sinti{\Bs'}M(\alpha')(s')\\
      &=\bigsqcup_{s'\in\Sd'}\bigg(\bigsqcup_{\Bs'\in\Bd', \alpha'\in\vals'}\sinti{\Bs'}M(\alpha')\bigg)(s')
    \end{align*}
    using \cite[Prop.~2.1.4]{AJ95} in the last step.
    This concludes the proof of $\sinti\Bs{\lambda x\ldotp M'}(\alpha)\arel\bigsqcup_{\Bs'\in\Bd',\alpha'\in\vals'}\sinti{\Bs'}{\lambda x\ldotp M'}(\alpha')$.\qedhere
  \end{itemize}
\end{proof}

\subsubsection{Supplementary Materials for \cref{sec:syntactic}}
\label{sec:appsyntactic}

The relation $\pred$ is defined in \cref{fig:parallel}.

\begin{figure*}[!t]
\centering
\begin{align*}
  \begin{tabularx}{0.7\textwidth}{XcXcXcX}
    &\infer[R\in\Sigma'\setminus\Sigma]{R\pred\lambda\overline
      x_R\ldotp
      F_R}{}&&\infer[c\in\Sigma\cup\{\land,\lor,\exists\}]{c\pred
      c}{}&
    &\infer{x\pred x}{}&\\[5pt]
    &\infer{M_1\,M_2\pred N_1\,N_2}{M_1\pred N_1&&M_2\pred N_2}&
    &\infer{\lambda x\ldotp M\pred\lambda x\ldotp N}{M\pred N}&
  \end{tabularx}
\end{align*}
\caption{Syntactic unfolding}
\label{fig:parallel}
\end{figure*}

\parallelcorr*
\begin{proof}
  We prove the lemma by induction on $\pred$:
  \begin{itemize}[noitemsep]
  \item For variables, symbols from $\Sigma$ and logical constants (other than $\neg$) this is trivial.
  \item If $M$ is a symbol $R\in\Sigma'\setminus\Sigma$ then $\sinti{T_\Prgm(\Bs)}R(\alpha)=\sinti\Bs{\lambda\overline x_R\ldotp F_R}(\alpha)$.
  \item Next, if $M$ is an application $M_1\,M_2$, $M_1\pred N_1$ and $M_2\pred N_2$ then
    \begin{align*}
      \sinti{T_\Prgm(\Bs)}{M_1\,M_2}(\alpha)&=\sinti{T_\Prgm(\Bs)}{M_1}(\alpha)(\sinti{T_\Prgm(\Bs)}{M_2}(\alpha))\\
                                          &=\sinti\Bs{N_1}(\alpha)(\sinti\Bs{N_2}(\alpha))\\
                                          &=\sinti\Bs{N_1\,N_2}(\alpha)
    \end{align*}
    using the inductive hypothesis in the second step.
  \item Finally, if $M$ is a $\lambda$-abstraction $\lambda x\ldotp M'$ and $M'\pred N'$ then
    \begin{align*}
      &\sinti{T_\Prgm(\Bs)}{\lambda x\ldotp M'}(\alpha)\\
      =\;&\lambda r\in\sinti \Hf{\Delta(x)}\ldotp\sinti{T_\Prgm(\Bs)}{M'}(\alpha[x\mapsto r])\\
         =\;&\lambda r\in\sinti \Hf{\Delta(x)}\ldotp\sinti\Bs{N'}(\alpha[x\mapsto r])\\
      =\;&\sinti\Bs{\lambda x\ldotp N'}(\alpha),
    \end{align*}
    exploiting the fact that $\Hf$ is a frame and the inductive hypothesis.\qedhere
  \end{itemize}
\end{proof}

\begin{definition}
  A positive existential formula $F$ is \emph{ex-normal} if for all subterms $\exists M$ of $F$, $M$ is a $\lambda$-abstraction.  
\end{definition}
\begin{restatable}[Basic properties of $\bured$]{lemma}{buredprops}
  Suppose $M\bured N$. Then
  \begin{thmlist}
  \item\label{lem:parallelbu} $\pred\;\subseteq\;\buredrt$,
  \item\label{lem:nonewfree} $\free(N)\subseteq\free(M)$,
  \item\label{lem:seminormal} if $M$ is ex-normal then $N$ is ex-normal, too.
  \end{thmlist}
\end{restatable}
\begin{proof}
  \begin{thmlist}
  \item Straightforward induction on the definition of $\pred$.
  \item We prove the first part of the lemma by induction on the compatible closure of $\bured$. If $(M,N)\in\beta$ this is a standard fact of $\beta$-reduction. If $(M,N)\in\upsilon$ then $\free(M)=\free(N)=\emptyset$. In the inductive cases the claim immediately follows from the inductive hypothesis.
  \item We prove the claim by induction on the compatible closure of $\bured$. 
    \begin{itemize}
    \item First, suppose that $(R,\lambda\overline x_R\ldotp F_R)\in\upsilon$. Obviously, $F_R$ is ex-normal and hence, $\lambda\overline x_R\ldotp F_R$ is ex-normal, too.
    \item Next, suppose that $((\lambda x\ldotp M)M',M[M'/x])\in\beta$. Clearly, $M$ and $M'$ must be ex-normal. We prove by induction on $M$ that $M[M'/x]$ is ex-normal. If $M$ is a variable this is obvious (because $M'$ is ex-normal). The cases for (logical) constants and $\lambda$-abstractions are straightforward. 

      Finally, suppose that $M$ is an application and let $\exists L$ be a subterm of
      \begin{align*}
        M_1[M'/x]M_2[M'/x]
      \end{align*}
      By the inductive hypothesis, both $M_1[M'/x]$ and $M_2[M'/x]$ are ex-normal. Hence, if $\exists K$ is a subterm of either $M_1[M'/x]$ or $M_2[M'/x]$ then $K$ must be a $\lambda$-abstraction. Otherwise $M_1=\exists$ and $K=M_2[M'/x]$. Then by assumption $M_2$ is a $\lambda$-abstraction and clearly, $M_2[M'/x]$ is a $\lambda$-abstraction, too.
    \item Next, suppose that $M_1M_2\bured N_1M_2$ because $M_1\bured N_1$. Clearly, $M_1$ is ex-normal. Therefore, by the inductive hypothesis, $N_1$ is ex-normal. Note that $N_1=\exists$ is impossible. Therefore, any subterm $\exists L$ of $N_1M_2$ is either a subterm of $N_1$ or $M_2$, which are both ex-normal. Hence, $L$ is a $\lambda$-abstraction.
    \item Suppose $M_1M_2\bured M_1N_2$ because $M_2\bured N_2$. Clearly, $M_2$ is ex-normal. Therefore, by the inductive hypothesis, $N_2$ is ex-normal. 
      Let $\exists L$ be a subterm of $M_1N_2$. If $\exists L$ is a subterm of $M_1$ or $N_2$ the argument is as in the previous case. Hence, suppose $M_1=\exists$ and $L=N_2[M'/x]$. By assumption $M_2$ is a $\lambda$-abstraction. Due to $M_2\bured L$, $L$ is a $\lambda$-abstraction, too.
    \item Finally, suppose that $\lambda x\ldotp M\bured\lambda x\ldotp N$ because $M\bured N$. Clearly, $M$ is ex-normal and hence by the inductive hypothesis $N$ is ex-normal. Let $\exists L$ be a subterm of $\lambda x\ldotp N$. Obviously, $\exists L$ must be a subterm of $M$, which is ex-normal. Hence, $L$ is a $\lambda$-abstraction.\qedhere
    \end{itemize}
  \end{thmlist}
\end{proof}

\begin{restatable}[Subject Reduction]{lemma}{subred}
  \label{lem:subred2}
  Let $\Delta\vdash M\from\sigma$ be a term such that $M\bured N$. Then
  \begin{thmlist}
  \item $\Delta\vdash N\from\sigma$ and
  \item\label{lem:functionalnobured} $\sigma$ is a relational type.
  \end{thmlist}
\end{restatable}
\begin{proof}
  \begin{thmlist}
  \item We prove the lemma by induction on the compatible closure of $\beta\upsilon$. For $(M,N)\in\beta$ this is \cite[Proposition 1.2.6]{BLCT13}. If  $(R,\lambda\overline x_R\ldotp F_R)\in\upsilon$ and $R\from\overline\tau\to o\in\Sigma'\setminus\Sigma$ then by convention $\Delta(\overline x_R)=\overline\tau$ and hence, $\Delta\vdash\lambda\overline x_R\ldotp F_R\from\overline\tau\to o$, too. The proofs for the recursive cases are exactly as in the proof of \cite[Proposition 1.2.6]{BLCT13}.
  \item Clearly, it suffices to prove by induction on the compatible closure of $\beta\upsilon$ that $M\bured N$ implies $\Delta\not\vdash M\from\iota^n\to\iota$ for all $n\in\nat$.
    \begin{itemize}
    \item If $(M,N)\in\beta\upsilon$ then clearly $\Delta\not\vdash M\from\iota^n\to\iota$ for all $n\in\nat$.
    \item Next, suppose $M_1M_2\bured N_1N_2$ due to $M_1\bured N_1$. Then by the inductive hypothesis $\Delta\not\vdash M_1\from\iota^n\to\iota$ for all $n\in\nat$. Hence, clearly $\Delta\not\vdash M_1M_2\from\iota^m\to\iota$ for all $m\in\nat$.
    \item Suppose $M_1M_2\bured N_1N_2$ due to $M_2\bured N_2$ and assume towards contradiction that $\Delta\vdash M_1M_2\from\iota^n\to\iota$. Then $\Delta\vdash M_1\from\sigma\to\iota^n\to\iota$ and $\Delta\vdash M_2\from\sigma$ for some $\sigma$. However, by the definition of types this implies $\sigma=\iota$, which contradicts the inductive hypothesis.
    \item Finally, if $\lambda x\ldotp M'\bured\lambda x\ldotp N$ then clearly $\Delta\not\vdash\lambda x\ldotp M'\from\iota^n\to\iota$ for all $n\in\nat$.\qedhere
    \end{itemize}
  \end{thmlist}
\end{proof}

\subsubsection{Supplementary Materials for \cref{sec:leftmost}}
\begin{restatable}[Basic Properties of $\lredrt$]{lemma}{propslredrt}
   Let $L,M,N$ and $Q$ be terms. Then
   \begin{thmlist}
   \item\label{lem:ltrans} $\lredrt$ is reflexive and transitive;
   \item\label{lem:lsubsetbu} $\lredrt\;\subseteq\;\buredrt$;
   \item\label{lem:lred0} if $M\lred{0} N$ then $M=N$;
   \item\label{lem:lredassoc} if $L\lred{m+1}N$ then there exists $M$ satisfying $L\lred 1 M\lred{m} N$;
   \item\label{lem:hleft} if $M\,Q$ is a term and $M\,\lredrt N$ then $M\,Q\lredrt N\,Q$;
   \item\label{lem:hrepl} if $M[Q/z]$ is a term and $M\lredrt N$ then $M[Q/z]\lredrt N[Q/z]$.
   \end{thmlist}
 \end{restatable}
\begin{proof}
  \begin{thmlist}
  \item Completely trivial.
  \item Straightforward induction on the definition of $\lred m$
  \item Straightforward induction on the definition of $\lred 0$
  \item Straightforward induction on the definition of $\lred{m+1}$
  \item Straightforward induction on the definition of $M \lred m N$ noting that for $\circ\in\{\land,\lor\}$ the cases $M_1\circ M_2\lred{m_1+m_2}N_1\circ N_2$ and $\exists x\ldotp M'\lred m\exists x\ldotp N'$ cannot occur because $(M_1\circ M_2)\,Q$ and $(\exists x\ldotp M')\,Q$ are not terms.
  \item We prove by induction on $M\lred m N$ that $M[Q/z]\lred m N[Q/z]$.
    \begin{itemize}
    \item If $M=N$ and $m=0$ then also $M[Q/z]\lred 0 N[Q/z]$.
    \item If there exist $L$, $m_1$ and $m_2$ such that $M\lred{m_1}L\lred{m_2}N$ and $m=m_1+m_2$ then by the inductive hypothesis $M[Q/z]\lred{m_1}L[Q/z]\lred{m_2}N[Q/z]$. Consequently, $M[Q/z]\lred{m}N[Q/z]$.
    \item Next, suppose that $M$ is $M_1\circ M_2$ for $\circ\in\{\land,\lor\}$ and that there exist $m_1$ and $m_2$ such that $m=m_1+m_2$ and $M_j\lred{m_j}N_j$ for $j\in\{1,2\}$. By the inductive hypothesis, $M_j[Q/z]\lred{m_j}N_j[Q/z]$. Consequently,
      \begin{align*}
        (M_1\circ M_2)[Q/z]=&\;(M_1[Q/z]\circ M_2[Q/z])\\
                           \lred{m_1+m_2}&\;(N_1[Q/z]\circ N_2[Q/z])\\
                           =&\;(N_1\circ N_2)[Q/z].
      \end{align*}
    \item Suppose that $M$ is $\exists x\ldotp M'$ and that $\exists x\ldotp M'\lred m \exists x\ldotp N'$. By the inductive hypothesis, $M'[Q/z]\lred m N'[Q/z]$.
      By the variable convention, $x\neq z$. Hence,
      \begin{align*}
        &(\exists x\ldotp M')[Q/z]=\exists x\ldotp M'[Q/z]\lred m\\
        &\exists x\ldotp N'[Q/z]=(\exists x\ldotp N')[Q/z].
      \end{align*}
    \item Suppose that $M$ is $R\,\overline M'$ for $R\in\Sigma'\setminus\Sigma$ and that $R\,\overline M'\lred 1(\lambda\overline x_R\ldotp F_R)\overline M'$. Clearly,
      \begin{align*}
        &(R\,\overline M')[Q/z]=R\,\overline M'[Q/z]\lred 1\\
        &(\lambda\overline x_R\ldotp F_R)\overline M'[Q/z]=((\lambda\overline x_R\ldotp F_R)\overline M')[Q/z]
      \end{align*}
      using the variable convention and the fact that $\lambda\overline x_R\ldotp F_R$ is closed.
    \item Finally, suppose that $M$ is $(\lambda x\ldotp M')M''\overline M'''$ and that
      \begin{align*}
        (\lambda x\ldotp M')M''\overline M'''\lred 1 M'[M''/x]\overline M'''.
      \end{align*}
      By the variable convention $x\neq z$. Hence
      \begin{align*}
        &\;((\lambda x\ldotp M')M''\overline M''')[Q/z]\\
        =&\;(\lambda x\ldotp M'[Q/z])M''[Q/z]\overline M'''[Q/z]\\
        \lred 1&\; M'[Q/z][M''[Q/z]/x]\overline M'''[Q/z]\\
        =&\; (M'[M''/x]\overline M''')[Q/z]
      \end{align*}
      using the Nested Substitution Lemma from \cite[2.1.16.\ Substitution Lemma]{B12}.\qedhere
    \end{itemize}
  \end{thmlist}
\end{proof}
 Besides, the following Inversion Lemma is immediate by definition.
 \begin{lemma}[Inversion]
   \label{lem:linversion}
   \begin{thmlist}
   \item If $\exists\overline x\ldotp E\lred{m}F$ then there exists $F'$ such that $F'=\exists\overline x\ldotp F'$ and $E\lred{m} F'$.
   \item If $E_1\circ\cdots\circ E_n\lred{m}F$, where $\circ\in\{\land,\lor\}$, then there exist $F_1,\ldots,F_n$ and $m_1,\ldots,m_n$ satisfying $F=F_1\circ\cdots\circ F_n$, $m=\sum_{i=1}^n m_i$ and $E_j\lred{m_j}F_j$ for each $1\leq j\leq n$.
   \item If $\exists\overline x\ldotp A_1\land\cdots\land A_n\lred{m}F$ then there exist $F_1,\ldots,F_n$ and $m_1,\ldots,m_n$ satisfying $F=\exists\overline x\ldotp F_1\land\cdots\land F_n$, $m=\sum_{i=1}^n m_i$ and $A_j\lred{m_j}F_j$ for each $1\leq j\leq n$.
   \item If $(\lambda x\ldotp K)L\overline M\lred{1}N$ then $N=K[L/x]\overline M$.
   \end{thmlist}
 \end{lemma}

\begin{restatable}[Basic Properties of $\sred$]{lemma}{propssred}
  \begin{thmlist}
  \item\label{lem:isrefl} $\sred$ is reflexive (on positive existential terms).
  \item\label{lem:ssubsetbu} $\sred\;\subseteq\;\buredrt$.
  \item\label{lem:ilsr} If $L\sred N$ and $\overline O\sred\overline Q$ then $L\overline O\sred N\,\overline Q$.
  \item\label{lem:hs} If $K\lredrt L\sred N$ then $K\sred N$.
  \item\label{lem:srepl} If $L\sred N$ and $O\sred Q$ then $L[O/z]\sred N[Q/z]$.
  \end{thmlist}
\end{restatable}
\begin{proof}
  \begin{thmlist}
  \item We prove by structural induction on $M$ that $M\sred M$. $M$ has the form $M_1\cdots M_n$, where $M_1$ is either a variable, a symbol from $\Sigma'\cup\{\land,\lor\exists_\tau\}$ or a $\lambda$-abstraction. In any case the inductive hypothesis and the reflexivity of $\lredrt$ % (\cref{lem:ltrans})
    immediately yield that $M_1\cdots M_n\sred M_1\cdots M_n$.  
  \end{thmlist}
  We prove the remaining four parts by induction on the definition of $\sred$. We only show the detailed proof for the case $L\sred x\,\overline N$ due to $L\lredrt x\,\overline M$ and $\overline M\sred\overline N$ for some $\overline M$ (the other cases are analogous).
  \begin{thmlist}
    \setcounter{thmlisti}{1}
  \item By \cref{lem:lsubsetbu} and the inductive hypothesis, $L\buredrt x\,\overline M$ and $\overline M\buredrt\overline N$. Therefore clearly, $L\buredrt x\,\overline M\buredrt x\,\overline N$ and hence also, $L\buredrt x\,\overline N$.
  \item By \cref{lem:hleft}, $L\overline O\lredrt x\,\overline M\,\overline O$ and hence by definition $L\overline O\sred x\,\overline N\,\overline Q$.
  \item By transitivity of $\lredrt$ (\cref{lem:ltrans}), $K\lredrt x\,\overline M$ and hence by definition $K\sred x\,\overline N$.
  \item % We prove the lemma by induction on $L\sred N$.
    % \begin{itemize}
    % \item First, suppose $L\sred x\,\overline N$ because for some $\overline M$, $L\lredrt x\,\overline M$ and $\overline M\sred\overline N$.
      By the inductive hypothesis, $\overline M[O/z]\sred\overline N[Q/z]$ and by assumption or \itemref{lem:isrefl}, $x[O/z]\sred x[Q/z]$. Therefore by \itemref{lem:ilsr} and \cref{lem:hrepl},
      \begin{align*}
        L[O/z]\lredrt x[O/z]\overline M[O/z]\sred x[Q/z]\overline N[Q/z],
      \end{align*}
      which proves $L[O/z]\sred (x\,\overline N)[Q/z]$ by \itemref{lem:hs}.\qedhere
    % \item Similarly, suppose $L\sred c\,\overline N$, because for some $\overline M$, $L\lredrt c\,\overline M$ and $\overline M\sred\overline N$. By the inductive hypothesis, $\overline M[O/z]\sred\overline N[Q/z]$ and by \cref{lem:hrepl}, $L[O/z]\lredrt c\,\overline M[O/z]$. Hence by definition, $L[O/z]\sred c\,\overline N[Q/z]$.
    % \item Finally, suppose $L\sred (\lambda x\ldotp N')\overline N$ because for some $M'$ and $\overline M$, $L\lredrt (\lambda x\ldotp M')\overline M$, $M'\sred N'$ and $\overline M\sred\overline N$. By the inductive hypothesis, $M'[O/z]\sred N'[Q/z]$ and $\overline M[O/z]\sred\overline N[Q/z]$. By the variable convention (\cref{con:vars}), $x\neq z$. Furthermore, by \cref{lem:hrepl}, $L[O/z]\lredrt (\lambda x\ldotp M'[O/z])\overline M[O/z]$. Hence, by definition, $s\sred(\lambda x\ldotp N'[O/z])\overline N[O/z]$.
    % \end{itemize}
  \end{thmlist}
\end{proof}

\begin{restatable}[Inversion]{lemma}{inversion}
  \label{lem:inversion}
  Let $E$ be an ex-normal formula.
  \begin{thmlist}
  \item\label{it:inversion1} If $E\sred x\,\overline N$ then there exists $\overline M$ such that $E\lredrt x\,\overline M$.
  \item\label{it:inversion2} If $E\sred c\,\overline N$, where $c\in\Sigma'\cup\{\land,\lor,\exists_\tau\}$, then there exists $\overline M$ such that $E\lredrt c\,\overline M$ and $\overline M\sred\overline N$.
  \item If $E\sred\exists N$ then there exist $x$, $N'$ and $M$ such that $N=(\lambda x\ldotp N')$, $E\lredrt\exists x\ldotp M$ and $M\sred N'$.
  \end{thmlist}
\end{restatable}
\begin{proof}
  The first two parts are obvious by definition of $\sred$. Hence, suppose that $E\sred\exists N$. 

  By \cref{lem:seminormal,lem:ssubsetbu}, $N$ has the form $\lambda x\ldotp N'$ for some $N'$. Furthermore, by \itemref{it:inversion2} there exists $L$ such that $E\lredrt\exists_\tau L$ and $L\sred\lambda x\ldotp N'$. Again, by \cref{lem:seminormal,lem:lsubsetbu}, $L$ has the form $\lambda y\ldotp L'$. By definition of $\sred$, $(\lambda y\ldotp L')\sred(\lambda x\ldotp N')$ implies that there exists $M$ such that $(\lambda y\ldotp L')\lredrt(\lambda x\ldotp M)$ and $M\sred N'$. However, $(\lambda y\ldotp L')\lredrt(\lambda x\ldotp M)$ clearly implies $y=x$ and $L'=M$. Consequently, $E\lredrt\exists x\ldotp M$ and $M\sred N'$.
\end{proof}

\elimbum*
\begin{proof}
  We prove the lemma by induction on $K\sred M$.
  \begin{itemize}
  \item First, suppose $K\sred x\,M_1\cdots M_n$ because for some $L_1,\ldots,L_n$, $K\lredrt x\,L_1\cdots L_n$ and $L_i\sred M_i$ for each $i$. Clearly, $x\, M_1\cdots M_n\bured x\, N_1\cdots N_n$, because of $M_j\bured N_j$ for some $j$ and $M_i=N_i$ for $i\neq j$ are the only possible $\beta\upsilon$-reductions. By the inductive hypothesis, $L_j\sred N_j$ and therefore by definition, $L\sred x\,N_1\cdots N_n$.
  \item Next, suppose $K\sred c\,\overline M$ because for some $\overline L$, $K\lredrt c\,\overline L$ and $\overline L\sred\overline M$. If $c=R\in\Sigma'\setminus\Sigma$ and $R\,\overline M\bured(\lambda\overline x_R\ldotp F_R)\overline M$ then 
    $K\lredrt R\,\overline L\lred{1} (\lambda\overline x_R\ldotp F_R)\overline L$. Therefore, by reflexivity of $\sred$ (\cref{lem:isrefl}), $K\sred (\lambda\overline x_R\ldotp F_R)\overline M$.
    
    Otherwise, $\overline M$ is reduced and the argument is analogous to the case for $K\sred x\,\overline M$.
% c\,\overline M\bured c\,\overline N$ because of  $\overline M\bured\overline N$ . By the inductive hypothesis, $\overline L\twoheadrightarrow_s\overline N$ and by \cref{lem:hs}, $\overline K\twoheadrightarrow_s\overline N$. Hence, by definition,
%     $c\,\overline K\rightarrow_i c\,\overline N$.
  \item Finally, suppose $K\sred (\lambda x\ldotp M')\overline M$ because for some $L'$ and $\overline L$, $K\lredrt (\lambda x\ldotp L')\overline L$, $L'\sred M'$ and $\overline L\sred\overline M$. Let $\overline L=(L_1,\ldots,L_n)$ and $\overline M=(M_1,\ldots,M_n)$.

    First, suppose $(\lambda x\ldotp M')\overline M\bured(\lambda x\ldotp N')\overline M$, where $M'\bured N'$. By the inductive hypothesis, $L'\sred N'$. Therefore, by definition, $(\lambda x\ldotp K')\overline K\sred(\lambda x\ldotp N')\overline M$.
    
    The argument for the case $(\lambda x\ldotp M')M_1\cdots M_n\bured(\lambda x\ldotp M')N_1\cdots N_n$, where for some $j$, $M_j\bured N_j$ and $M_k=N_k$ for all $k\neq j$, is very similar.

    Finally, assume that $n\geq 1$ and $(\lambda x\ldotp M')M_1\cdots M_n\bured M'[M_1/x]M_2\cdots M_n$. Then
    \begin{align*}
      L&\lredrt(\lambda x\ldotp L')L_1\cdots L_n\\
       &\lred{1} L'[L_1/x]L_2\cdots L_n\\
       &\sred M'[M_1/x]M_2\cdots M_n,&\text{\cref{lem:ilsr,lem:srepl}},
    \end{align*}
    which proves $L\sred M'[M_1/x]M_2\cdots M_n$ by \cref{lem:hs,lem:ltrans}.\qedhere
  \end{itemize}
\end{proof}

\sdecompose*
\begin{proof}
  We prove the lemma by induction on the structure of $F$. Note that the case $(\lambda x\ldotp N')\overline N$ cannot occur for otherwise $F$ is not in $\beta$-normal form or does not have type $o$. If $F$ has the form $x\,\overline N$ then by the Inversion \cref{lem:inversion} there exists $\overline M$ such that $\posex(G)\lredrt x\,\overline M$, and clearly, $\alpha\force x\,\overline M$.

  Hence, the only remaining case is that $F$ has the form $c\,\overline N$. By the Inversion \cref{lem:inversion} there exist $\overline M$ such that $\posex(G)\lredrt c\,\overline M$ and $\overline M\sred\overline N$.
  Note that $c\in\Sigma'$ implies $c\from\iota^n\to o\in\Sigma$ for otherwise $\As_0,\alpha\not\models c\,\overline N$. 
  By \cref{lem:functionalnobured,lem:ssubsetbu}, $M=N$ and thus $\As_0,\alpha\models c\,\overline M$. Consequently, $\alpha\force c\,\overline M$.

  Next, suppose that $c$ is $\land$. Then $F$ is $N_1\land N_2$ and $c\,\overline M$ has the form $M_1\land M_2$. By \cref{lem:lsubsetbu} and the Subject Reduction \cref{lem:subred2}, $M_j$ is a positive existential formula and clearly, by assumption, $N_j$ is in $\beta$-normal form and $\As_0,\alpha\models N_j$ for all $j\in\{1,2\}$. By the inductive hypothesis, there are $N_1'$ and $N_2'$ satisfying $M_j\lredrt N_j'$ and $\alpha\force N_j'$. Consequently, $\alpha\force N_1'\land N_2'$ and by definition, $\posex(G)\lredrt N_1'\land N_2'$.

  The case where $c$ is $\lor$ is very similar.

  Finally, suppose that $c$ is $\exists_\tau$ and that $F$ is $\exists_\tau N_1$. By the Inversion \cref{lem:inversion} there exist $x$,
  $N'$ and $M$ such that $N_1=(\lambda x\ldotp N')$, $\posex(G)\lredrt \exists x\ldotp M$ and $M\sred N'$. By \cref{lem:lsubsetbu} and the Subject Reduction \cref{lem:subred2}, $M$ is a positive existential formula and clearly $N'$ is in $\beta$-normal form and $\As_0,\alpha[x\mapsto r]\models N'$ for some $r\in\sinti\Hf\tau$. By the inductive hypothesis, there exists $N''$ satisfying $M\lredrt N''$ and $\alpha[x\mapsto r]\force N''$. Consequently, $\alpha\force\exists x\ldotp N''$ and by definition, $\posex(G)\lredrt\exists x\ldotp N''$.
\end{proof}

\subsubsection{Supplementary Materials for \cref{sec:completeness}}
The proof of the following lemma is a straightforward induction on the
definition of $\force$:
\begin{lemma}
  \label{lem:vartrianglemodels}
  Let $\alpha,\alpha'$ be valuations and $F$ be positive existential formulas satisfying $\alpha\force F$. 
  If $\alpha(x)=\alpha'(x)$ for all $x\in\free(F)$ then $\alpha'\force F$.
\end{lemma}

\mudec*
\begin{proof}
  Let $G\in\Set'$ be a goal clause, $F$ be a (closed) positive existential
  formula, $\alpha$ be a valuation and let $m=\mu(G)>0$ be such that $\posex(G)\lred m F$ and $\alpha\force F$.
  W.l.o.g.\ we can assume that
  \begin{align}
    \label{eq:varsdisjoint}
    \free(G)\cap\free(C)&=\emptyset& \text{for all } C\in \Set'.
  \end{align}
  (Otherwise, rename all variables occurring in $G$ to obtain $\widetilde G$ satisfying \cref{eq:varsdisjoint} and clearly, by definition of $\Res$, $\Set'\cup\{\widetilde G\}\Res \Set'\cup\{\widetilde G,G'\}$ implies $\Set'\Res \Set'\cup\{G'\}$.)

  Furthermore, suppose that $G=\neg A_1\lor\cdots\lor\neg A_n$ and $\posex(G)=\exists\overline x\ldotp\bigland_{i=1}^n A_i$.
  By the Inversion \cref{lem:linversion}, there exist $F_1,\ldots,F_n$ and $m_1,\ldots,m_n$ such that $F=\exists\overline x\ldotp\bigland_{i=1}^n F_i$, $m=\sum_{i=1}^n m_i$ and $A_j\lred{m_j}F_j$ for each $1\leq j\leq m$. 
  Note that due to $\alpha\force F$ we can assume w.l.o.g.\ that also $\alpha\force F_j$ for each $1\leq j\leq n$, and
  furthermore we can assume that $m_1>0$. By \cref{lem:lredassoc}, there exists $E$ such that $A_1\lred{1}E\lred{m_1-1}F_1$. Since $A_1$ is an atom there are exactly two cases:
  \begin{thmlist}
  \item $A_1=(\lambda y\ldotp L)M\,\overline N$ and $E=L[M/y]\overline N$ or
  \item\label{it:proofcase2} $A_1=R\,\overline M$ and $E=(\lambda\overline x_R\ldotp F_R)\overline M$.
  \end{thmlist}
  The first case is easy because for $G'=\neg L[M/y]\overline N\lor\biglor_{i=2}^n\neg A_i$, $\Set'\cup\{G\}\Res \Set'\cup\{G,G'\}$ and $\posex(G')=\exists\overline x\ldotp L[M/y]\overline N\land\bigland_{i=2}^n A_i$.
    
  In the second case, note that $\lred{1}$ is functional on applied $\lambda$-abstractions (by the Inversion \cref{lem:linversion}). Hence, we can assume that
    \begin{align*}
      (\lambda\overline x_R\ldotp F_R)\overline M\lredrt F_R[\overline M/\overline x_R]\lred{m^*_1}F_1,
    \end{align*}
    where $m^*_1\leq m_1-1$ for otherwise $\alpha\force F_1$ would clearly not hold.

    $F_R$ has the form $\posex(G_{R,1},\overline x_R)\lor\cdots\lor\posex(G_{R,k},\overline x_R)$, where each $G_{R,j}$ is a goal clause and $G_{R,j}\lor R\,\overline x_R\in \Set'$. Let $\overline y_1,\ldots,\overline y_k$ and $E'_1,\ldots,E'_k$ be such that for each $j$, $\posex(G_{R,j},\overline x_R)=\exists\overline y_j\ldotp E_j'$. Note that by \cref{eq:varsdisjoint}, $\posex(G_R,\overline x_R)[\overline M/\overline x_R]=\exists\overline y_j\ldotp E'_j[\overline M/\overline x_R]$ for each $j$ and by the Inversion \cref{lem:linversion}, there exist $F_1',\ldots,F'_k$ and $m'_1,\ldots,m'_k$ such that $F_1=\bigvee_{j=1}^k(\exists\overline y_j\ldotp F_j')$, $E'_j[\overline M/\overline x_R]\lred{m_j'}F'_j$ and $m'_j\leq m_1^*$ for each $j$. 

Next, because of $\alpha\force F_1$ there exists $1\leq j\leq k$ and $\overline r\in\sinti\Hf{\Delta(\overline y_j)}$ satisfying $\alpha[\overline y_j\mapsto\overline r]\force F'_j$.
Furthermore, because of \cref{eq:varsdisjoint,lem:nonewfree,lem:vartrianglemodels}, $\alpha[\overline y_j\mapsto\overline r]\force F_i$ for all $2\leq i\leq n$. Therefore,
\begin{align}
  \label{eq:compalphamodels}\alpha[\overline y_j\mapsto\overline r]\force F'_j\land\bigwedge_{i=2}^n F_i.
\end{align}

Clearly, it holds that
\begin{align}
  \label{eq:compmakeres}\Set'\cup\{G\}\Res \Set'\cup\left\{G,G_{R,j}[\overline M/\overline x_R]\lor\biglor_{i=2}^n\neg A_i\right\}\\
  \label{eq:complred}E_j'[\overline M/\overline x_R]\land\bigland_{i=2}^n A_2\lred{m'_j+\sum_{i=2}^n m_i}F_j'\land\bigland_{i=2}^n F_i
\end{align}
and $\free(G')\subseteq\overline x\cup\overline y$. Let $\overline x'\subseteq\overline x$ and $\overline y'\subseteq\overline y$ be such that $\free(G')=\overline x'\cup\overline y'$. Hence, $\posex(G')=\exists\overline x',\overline y'\ldotp E_j'[\overline M/\overline x_R]\land\bigland_{i=2}^n A_2$.
We define
\begin{align*}
  G'&\defeq G_{R,j}[\overline M/\overline x_R]\lor\biglor_{i=2}^n\neg
      A_i\\
  F'&\defeq\exists\overline x',\overline y'\ldotp
      E_j\land\bigland_{n=2}^n F_i\\
  m'&\defeq m'_j+\sum_{i=2}^n m_i.
\end{align*}
By \cref{eq:compalphamodels,eq:compmakeres,eq:complred}, it holds that
\begin{enumerate*}
\item $\Set'\Res \Set'\cup\{G'\}$,
\item $\posex(G')\lred{m'}F'$,
\item $\alpha\force F'$ and
\item $m'\leq m_1^*+\sum_{i=2}^n m_i<\sum_{i=1}^n m_i=m$.
\end{enumerate*}
Consequently, also $\mu(\Set'\cup\{G'\})<\mu(\Set')$
\end{proof}

frame). By the inductive hypothesis and \cref{lem:termmon},
%       \begin{align*}
%         1=\sinti{\As^\Hf_{\beta'}}{F_R}(\top_\Delta^\Hf[\overline x_R\mapsto\overline s])\leq\sinti{\Bs}{F_R}(\top_\Delta^\Hf[\overline x_R\mapsto\overline s'])
%       \end{align*}
%       Consequently, due to $\Bs\models P$, $R^\Bs(\overline s')=1$.
%     \item Finally, if $\beta$ is a limit ordinal, by the inductive hypothesis, $\As^\Hf_{\beta'}\mrel\Bs$ for each $\beta'\in\beta$. Thus, by \cref{cor:unionlub}, $\As^\Hf_\beta\mrel\Bs$.\qedhere
%     \end{itemize}
%   \end{proof}

%   \propap*
%   \begin{proof}
%     \begin{thmlist}
%     \item \cref{lem:plateau,lem:Abetmon}.
%     \item \cref{lem:plateau,lem:canmod}.
%     \item Let $\Bs$ be an expansion of $\As$ satisfying $\Bs\models
%       S$. By \cref{lem:termmon} and \itemref{lem:canmod2},
%       $\sinti{\As_P^\Hf}G(\alpha)\geq\sinti{\Bs}G(\top_\Delta^\Hf)$
%       for every goal clause $G\in S$ and valuation $\alpha$. Hence, by
%       assumption, $\As^\Hf_P\models G$ and therefore by \cref{lem:modp} $\As^\Hf_P\models S$.\qedhere
%   \end{thmlist}
%   \end{proof}

%   % \subsection{Supplementary Materials for \cref{sec:equiv}}
%   % \label{sec:appequiv}

% % \subsection{Supplementary Materials for \cref{sec:compact-theories}}

\subsection{Supplementary Materials for \cref{sec:compact-theories}}
\label{sec:appcompact-theories}

\compcomp*
\begin{proof}
  The ``if''-direction is straightforward. For the converse, suppose
  that $\Set$ is $\Ad$-unsatisfiable.
  By the Completeness \cref{thm:completeness}, for each $\As\in\Ad$
  there exist $G_\As$, every atom of which has the form $x\,\overline M$, and
  background atoms $\phi_{\As,i}$ and $\Set_\As$ such that
  $\As\not\models\neg\phi_{\As,1}\lor\cdots\lor\neg\phi_{\As,m_\As}$
  and $\Set\Res^*
  \Set_\As\cup\{G_\As\lor\neg\phi_{\As,1}\lor\cdots\lor\neg\phi_{\As,m_\As}\}=\Set'_\As$. Hence,
  $\{\neg\phi_{\As,1}\lor\cdots\lor\neg\phi_{\As,m_\As}\mid\As\in\Ad\}$
  is $\Ad$-unsatisfiable and by compactness of $\Ad$ there exists finite $\Ad'\subseteq\Ad$
  such that $\{\neg\phi_{\As,1}\lor\cdots\lor\neg\phi_{\As,m_\As}\mid\As\in\Ad'\}$
  is $\Ad$-unsatisfiable. Consequently, $\Set\Resp^*\{\Set'_\As\mid\As\in\Ad'\}\Resp\{\bot\}\cup\{\Set'_\As\mid\As\in\Ad'\}$.
  % \lo{Nice argument.}
\end{proof}
% As an interesting special case, this shows that the proof system is
% also sound and complete in the unconstrained setting:
% the set of 1st-order $\Sigma$-structures (interpreting the equality symbol as identity)
% % $\Ad=\{\As\mid\As\text{ is a 1st-order $\Sigma$-structure}\}$
% is compact by the compactness theorem for 1st-order logic. 
% Consequently, there does not exist a $\Sigma'$-structure $\Bs$ (interpreting equality as identity) satisfying $\Bs\models S$ iff there exists $S'$ such that $S \Respstr{\Ad}^* S'\cup\{\bot\}$.

\subsection{Supplementary Materials for \cref{sec:appenc}}
\subsubsection{$\lambda$-lifting}
\label{sec:elimnest}
\label{SEC:ELIMNEST}
In this section, we examine how to eliminate
$\lambda$-abstractions. We make use of the notion of terms with holes
(cf.~\cite[p.~29]{B12}, \cite{BLCT13}).

%  First, note that the proof of \cref{lem:termmon} can be adapted in a straightforward manner to obtain:
% \begin{lemma}
%   \label{lem:termmonhole}
%   Let $\Bs$ be a $\Sigma'$- and $\Bs'$ be a $(\Sigma'',\Hf)$-structure,
%   where $\Sigma''\supseteq\Sigma'$, let $M[-]$ be a $\Sigma'$-term with a hole of type $\tau$, let $N$ be a $\Sigma'$- and $N'$ be a $\Sigma''$-formula satisfying
%   \begin{thmlist}
%   \item $M[N]$ is a positive existential formula,
%   \item $\Delta\vdash N\from\tau$ and $\Delta\vdash N'\from\tau$,
%   \item for all $c\in\Sigma'$, $c^\Bs\mrel c^{\Bs'}$
%   \item for all valuations $\alpha\mrel\alpha'$, $\sinti\Bs N(\alpha)\mrel\sinti{\Bs'}{N'}(\alpha')$.
%   \end{thmlist}
%   Then for all valuations $\alpha\mrel\alpha'$, $\sinti\Bs{M[N]}(\alpha)\mrel\sinti{\Bs'}{M[N']}(\alpha')$.
% \end{lemma}
Let $\lambda y\ldotp M$ be a positive existential $\Sigma'$-term not containing
logical symbols with free variables
$\overline x$ such that $\Delta\vdash M\from\overline\tau\to o$, let
$\widetilde \Set[-]$ be a set of terms with a hole of type
$\Delta(y)\to\overline\tau\to o$ such that $\Set[\lambda y\ldotp M]$ is a set of HoCHCs.

Let $\overline z$
be distinct variables (different from $\overline x,y$) satisfying
$\Delta(\overline z)=\overline\tau$. We define a signature
$\Sigma''\defeq\Sigma'\cup\{R_M\from\Delta(\overline x)\to\Delta(y)\to\overline\tau\to
o\}$ and
\begin{align*}
  \Set&\defeq\widetilde \Set[\lambda y\ldotp M]&
  \Set'&\defeq\widetilde \Set[R_M\,\overline x]\cup\{\neg M\,\overline z\lor R_M\,\overline x\,y\,\overline z\}.
\end{align*}
These are sets of HoCHCs.

\begin{proposition}
  \label{prop:elimnested}
  Let $\As\in\Ad$. Then $\Set$ is $\As$-monotone-satisfiable iff $\Set'$ is $\As$-monotone-satisfiable.
\end{proposition}
\begin{proof}
  \begin{itemize}
  \item First, suppose that there exists a $(\Sigma',\Mf)$-expansion $\Bs$
    of $\As$ satisfying $\Bs\models\Set$. We define a
    $(\Sigma'',\Mf)$-expansion $\Bs'$ of $\As$ by setting
    $R^{\Bs'}\defeq R^\Bs$
    for $R\in\Sigma'\setminus\Sigma$ and
    $R_M^{\Bs'}\defeq\sinti\Bs{\lambda\overline x,y,\overline z\ldotp
      M\,\overline z}$. By definition,
    $\Bs'\models\neg M\,\overline z\lor R_M\,\overline x\,y\,\overline
    z$. Furthermore for every positive existential $\Sigma'$-formula
    $E$ and valuation $\alpha$, $\sinti\Bs{E[\lambda y\ldotp M]}(\alpha)=\sinti{\Bs'}{E[R_M\,\overline x]}(\alpha)$.
    Consequently, $\Bs'\models\Set'$.
  \item Conversely, suppose $\Set'$ is $\As$-monotone-satisfiable. Let $\Bs\defeq \As_{\Prgm_{\Set'}}$.
    By
    \cref{thm:canmod}, $\Bs\models \Set'$.
    Furthermore, by \cref{lem:moncoin,lem:knaster-tarski}, $\Bs=T_{\Prgm_{\Set'}}(\Bs)$. 
    Thus,
    $R_M^\Bs=\sinti\Bs{\lambda\overline x,y,\overline z\ldotp
      M\,\overline z}$ and therefore, by \cref{lem:betaeta},
    $\sinti\Bs{R_M\,\overline x}(\alpha)=\sinti\Bs{\lambda
      y\ldotp M}(\alpha)$ for every valuation $\alpha$. Consequently, $\Bs\models \Set$. \qedhere
    % Assume towards contradiction that $\As_P^\Hf\not\models S$. Then,
    % by \cref{thm:contmain}, $\As^\Hf_{P,n}\not\models G$  for some
    % $n\in\omega$ and goal clause $G\in S$.
    % By \cref{lem:elimrel,lem:termmonhole}, $\As^\Hf_{P',2(n+1)}\not\models
    % G'$ (for the goal clause $G'$ corresponding to $G$
    % in $S'$) and therefore, by \cref{lem:termmon,it:allordsmaller}, $\As^\Hf_{P'}\not\models G'$. Note that by \cref{lem:canmod2}, $\As^\Hf_{P'}\mrel\Bs'$ and therefore, by \cref{lem:termmon}, $\Bs'\not\models G'$, which is clearly a contradiction.
    % Consequently, $\As^\Hf_P\models S$.
  \end{itemize}
\end{proof}
Hence, we conclude:
\begin{corollary}
  \label{thm:ellambdahochc}
  Let $\Ad$ be a set of 1st-order $\Sigma$-structures and $\Set$ be
  a finite set of HoCHCs. 
 Then there exists a set of HoCHCs (over an extended signature) which does not contain $\lambda$-abstractions and which is $\Ad$-satisfiable iff $\Set$ is $\Ad$-satisfiable.
\end{corollary}

\subsubsection{Proof of \cref{lem:follifting}}
Before turning to \cref{lem:follifting}, we prove the following auxiliary lemma:
\begin{lemma}  Let $M,M_1,\ldots,M_n$ and $\overline L$ be terms neither containing logical symbols nor $\lambda$-abstractions. Then:
  \begin{thmlist}
  \item\label{lem:folunifier1} If $R\in\Sigma'\setminus\Sigma$,  $R\,M_1\cdots M_n$ and $R\,x_1\cdots x_n$ are terms such that $\free(R\,M_1\cdots M_n)\cap\free(R\,x_1\cdots x_n)=\emptyset$ then
    \begin{align*}
      [\lfloor M_1\rfloor'/x_1,\ldots,\lfloor M_n\rfloor'/x_n]
    \end{align*}
    is a unifier of $\lfloor R M_1\cdots M_n\rfloor'$ and $\lfloor R x_1\cdots x_n\rfloor'$.
  \item\label{lem:folunifier2} If $\Delta(y)=\rho=\tau_1\to\cdots\to\tau_n\to o$ and $y\,M_1\cdots M_n$ is a formula such that $\free(\Comp_\rho)\cap\free(y\,M_1\cdots M_n)$ then
    \begin{align*}
      [c_\rho/y,\lfloor M_1\rfloor'/x_1,\ldots,\lfloor M_n\rfloor'/x_n]
    \end{align*}
    is a unifier of $\Comp_\rho$ and $\lfloor y\,M_1\cdots M_n\rfloor$.
  \item\label{lem:folsubstaux} If $M$ is a term neither containing logical symbols nor $\lambda$-abstractions then
    \begin{align*}
      \lfloor M[\overline L/\overline x]\rfloor'=\lfloor M\rfloor'[\lfloor\overline L\rfloor'/\overline x].
    \end{align*}
  \item\label{lem:folsubst} If $G$ is a goal clause then $\lfloor G[M_1/x_1,\ldots,M_n/x_n]\rfloor=\lfloor G\rfloor[\lfloor M_1\rfloor'/x_1,\ldots,\lfloor M_n\rfloor'/x_n]$.
  \end{thmlist}
\end{lemma}
\begin{proof}
  \begin{thmlist}
  \item We prove \itemref{lem:folunifier1} by induction on $n$. For $n=0$ this is trivial. Hence, suppose $n\geq 0$. By the inductive hypothesis,
    \begin{align}
      \label{eq:subst1ih}
      &\lfloor R\,M_1\cdots M_n\rfloor'[\lfloor M_1\rfloor'/x_1,\ldots,\lfloor M_n\rfloor'/x_n]\nonumber\\
      =\;&\lfloor R\,x_1\cdots x_n\rfloor'[\lfloor M_1\rfloor'/x_1,\ldots,\lfloor M_n\rfloor'/x_n].
    \end{align}
    Consequently,
    \begin{align*}
      &\lfloor R\,M_1\cdots M_{n+1}\rfloor'[\lfloor M_1\rfloor'/x_1,\ldots,\lfloor M_{n+1}\rfloor'/x_{n+1}]\\
      =\;&(\app\,\lfloor R\,M_1\cdots M_n\rfloor'\,\lfloor
           M_{n+1}\rfloor')\\
      &\hspace{3cm}[\lfloor M_1\rfloor'/x_1,\ldots,\lfloor M_{n+1}\rfloor'/x_{n+1}]\\
      =\;&\app\,\lfloor R\,M_1\cdots M_n\rfloor'[\lfloor
           M_1\rfloor'/x_1,\ldots,\lfloor M_n\rfloor'/x_n]\,\lfloor M_{n+1}\rfloor'\\
      =\;&(\app\,\lfloor R\, x_1\cdots x_n\rfloor'\,x_{n+1})[\lfloor M_1\rfloor'/x_1,\ldots,\lfloor M_{n+1}\rfloor'/x_{n+1}]\\
      =\;&\lfloor R\, x_n\cdots x_1\rfloor'[\lfloor M_1\rfloor'/x_1,\ldots,\lfloor M_{n+1}\rfloor'/x_{n+1}],
    \end{align*}
    using that $x_i\nin\free(M_j)$ in the second and
    \cref{eq:subst1ih} in the third step.
  \item Similar to \itemref{lem:folunifier1}.
  \item We prove the claim by structural induction. For variables, and symbols from $\Sigma'\setminus\Sigma$ this is obvious. 
    Next, consider a term of the form $c\overline N$, where $c\in\Sigma$. By \cref{rem:isimple}, $c\overline N$ only contains variables $y\from\iota$ and for each term $\Delta\vdash K\from\iota^n\to\iota$, $\lfloor K\rfloor'=K$. Hence,
    \begin{align*}
      \lfloor(c\,\overline N)[\overline L/\overline x]\rfloor'=(c\,\overline N)[\overline L/\overline x]=\lfloor c\,\overline N\rfloor'[\lfloor\overline L\rfloor'/\overline x].
    \end{align*}
    Finally, consider a term of the form $M\overline N\,N'$, where $M\nin\Sigma$. Then,
    \begin{align*}
      \lfloor (M\,\overline N\,N')[\overline L/\overline x]\rfloor'&=\lfloor (M\,\overline N)[\overline L/\overline x]N'[\overline L/\overline x]\rfloor'\\
                                          &=\app\,\lfloor(M\,\overline N)[\overline L/\overline x]\rfloor'\,\lfloor N'[\overline L/\overline x]\rfloor'\\
                                          &=\app\,\lfloor M\,\overline N\rfloor'[\lfloor \overline L\rfloor'/\overline x]\,\lfloor N'\rfloor'[\lfloor \overline L\rfloor'/\overline x]\\
                                          &=(\app\,\lfloor M\,\overline N\rfloor'\,\lfloor N'\rfloor')[\lfloor \overline L\rfloor'/\overline x]\\
                                          &=\lfloor M\,\overline N\,N'\rfloor'[\lfloor \overline L\rfloor'/\overline x],
    \end{align*}
    using the inductive hypothesis in the third step.
  \item Immediate from \itemref{lem:folsubstaux}.\qedhere
  \end{thmlist}
\end{proof}
\follifting*
\begin{proof}
  \begin{thmlist}
  \item Obvious.
  \item Note that by assumption the $\beta$-reduction rule is not applicable. Next, let $\neg R\,\overline M\lor G_1$ and $G_2\lor R\,\overline x$ be clauses in $\Set'$ modulo renaming of variables (such that they are variable-disjoint) and suppose $G=G_1\lor (G_2[\overline M/\overline x])$. By \cref{lem:folunifier1}, $\lfloor R\,\overline M\rfloor[\lfloor\overline M\rfloor'/\overline x]=\lfloor R\,\overline x\rfloor[\lfloor\overline M\rfloor'/\overline x]$, by soundness of 1st-order resolution \cite{R65,F96},
    \begin{align*}
      \lfloor \Set'\rfloor\models (\lfloor G_1\rfloor \lor\lfloor
          G_2\rfloor)[\lfloor\overline M\rfloor'/\overline x]
    \end{align*}
    and by \cref{lem:folsubst},
    \begin{align*}
      (\lfloor G_1\rfloor \lor\lfloor
    G_2\rfloor)[\lfloor\overline M\rfloor'/\overline x]=\lfloor G_1\lor (G_2[\overline M/\overline x])\rfloor.
    \end{align*}
    Furthermore, $\vars(G)\subseteq\vars(\neg R\,\overline M\lor G_1)\cup\vars(G_2)$
    and therefore $\lfloor \Set'\cup\{G\}\rfloor=\lfloor \Set'\rfloor\cup\{\lfloor G\rfloor\}$.
  \item Finally, suppose that there exists
  $\{\biglor_{i=1}^{m_j'}\neg x_{j,i}\,M_{j,i}\lor\biglor_{i=1}^{m_j}\neg\phi_{j,i}\mid 1\leq
  j\leq n\}\subseteq \Set'$ such that
  $\{\biglor_{i=1}^{m_j}\neg\phi_{j,i}\mid 1\leq
  j\leq n\}$ is $\Ad$-unsatisfiable.
  Note that each $x_{j,i}$ does not occur in any of the $\phi_{j',i'}$. Therefore, by \cref{lem:folunifier2}, the fact that $\Comp_{\Delta(x_{j,i})}\in\lfloor \Set'\rfloor$ and the soundness of 1st-order resolution, $\lfloor \Set'\rfloor\models \{\biglor_{i=1}^{m_j}\neg\phi_{j,i}\mid 1\leq
  j\leq n\}$. Hence, by assumption $\lfloor \Set'\rfloor$ is $\Ad$-unsatisfiable.\qedhere
  \end{thmlist}
\end{proof}

\subsection{Supplementary Materials for \cref{sec:decidable}}
% \trstr*
% \begin{proof}
%   We define $\ell'$ recursively by
%   \begin{align*}
%     \ell'(\iota)&\defeq\ell'(o)\defeq 2\\
%     \ell'(\tau\to\rho)&\defeq\ell'(\tau)\cdot\ell'(\rho)
%   \end{align*}
%   and we set
%   \begin{align*}
%     \ell(\lfloor\tau\rfloor)=\max\{0,\max\{\ell'(\tau\to\rho)\mid @_{\tau,\rho}\in\widetilde\Sigma\}-\ell'(\tau)\}.
%   \end{align*}
%   Obviously, $\ell'(\tau\to\rho)>\ell'(\tau)$ and $\ell'(\tau\to\rho)>\ell'(\rho)$.
%   Besides, note that all elements
%   $f\from\iota_1\to\cdots\to\iota_{n+1}\in\widetilde\Sigma$ have the
%   form
%   $@_{\tau,\rho}\from\lfloor\tau\to\rho\rfloor\to\lfloor\tau\rfloor\to\lfloor\rho\rfloor$. Therefore,
%   for
%   $@_{\tau,\rho}\from\lfloor\tau\to\rho\rfloor\to\lfloor\tau\rfloor\to\lfloor\rho\rfloor\in\widetilde\Sigma$,
%   \begin{align*}
%     \ell(\lfloor\tau\to\rho\rfloor)&=\max\{\ell'(\tau'\to\rho')\mid
%                                      @_{\tau',\rho'}\in\widetilde\Sigma\}-\ell'(\tau\to\rho)\\
%                                    &>\max\{\ell'(\tau'\to\rho')\mid
%                                      @_{\tau',\rho'}\in\widetilde\Sigma\}-\ell'(\tau)\\
%                                    &=\ell(\lfloor\tau\rfloor)
%   \end{align*}
%   and similarly, $\ell(\lfloor\tau\to\rho\rfloor)>\ell(\lfloor\rho\rfloor)$.
% \end{proof}

\subsubsection{Higher-order Datalog}
\label{sec:hodatalog}
   First consider the higher-order extension of Datalog.% in \cite{VRW08}.

  \begin{assumption}
    Let $\Sigma\supseteq\{{\approx},{\not\approx}\from\iota\to\iota\to o,c_0\from\iota\}$ be a finite 1st-order signature containing ${\approx},{\not\approx}$ and symbol(s) of type $\iota$ (but nothing else);
    let $\Sigma'$ be a relational extension of $\Sigma$ and $\Set$ be a
    finite set of HoCHCs.

    Besides, let $\Ad$ be the set of (1st-order) $\Sigma$-structures
    $\As$ satisfying ${\approx}^\As(a)(b)=1$ iff
    ${\not\approx}^\As(a)(b)=0$ iff $a=b$, for $a,b\in\sinti\As\iota$.
  \end{assumption}
  In this setting we refer to HoCHCs as \emph{higher-order Datalog
    clauses (HoDC)}.
  For $\As\in\Ad$ we define $\As^\flat$ by
  \begin{align*}
    \sinti{\As^\flat}\iota&\defeq\{c\from\iota\in\Sigma\}&
                                                           {\approx^{\As^\flat}}(c)(d)&\defeq\sinti\As{c\approx d}\\
    c^{\As^\flat}&\defeq c&
                            {\not\approx^{\As^\flat}}(c)(d)&\defeq\sinti\As{c\not\approx d}
  \end{align*}
  for
  $c,d\from\iota\in\Sigma$
  % \begin{enumerate*}
  % \item $\sinti{\widehat\As}\iota\defeq\{c\from\iota\in\Sigma\}$,
  % \item $c^{\widehat\As}\defeq c$ for $c\from\iota\in\Sigma$,
  % \item $\approx^{\widehat\As}(c)(d)\defeq\sinti\As{c\approx d}$ and
  % \item $\not\approx^{\widehat\As}(c)(d)\defeq\sinti\As{c\not\approx d}$ for
  %     $c,d\in\sinti{\widehat\As}\iota$.
  %   \end{enumerate*}
    and we set
    $\Ad^\flat\defeq\{\As^\flat\mid\As\in\Ad\}$. Clearly,
    $\Ad^\flat$ is finite and for each $\As\in\Ad$ and type
    $\sigma$, $\sinti{\As^\flat}\sigma$ is finite.
  \begin{lemma}
    Let $\phi$ be a background atom and $\alpha$ be a valuation
    satisfying $\As^\flat,\alpha\models\phi$ then
    $\As,\alpha^\sharp\models\phi$, where $\alpha^\sharp$ is a valuation
    such that for each $x\from\iota\in\Delta$, $\alpha^\sharp(x)=\sinti\As{\alpha(x)}$. 
  \end{lemma}

  \begin{corollary}
    Let $\Set'$ be a set of goal clauses of background atoms. $\Set'$ is
    $\Ad^\flat$-satisfiable if $\Set'$ is $\Ad$-satisfiable.
  \end{corollary}

  \begin{lemma}
    Let $\Set'$ be a set of goal clauses of background atoms. $\Set'$ is
    $\Ad$-satisfiable if $\Set'$ is $\Ad^\flat$-satisfiable.
  \end{lemma}
  \begin{proof}
    Let $\As^\flat\in\Ad^\flat$ be such that $\As^\flat\models
    \Set'$. Consider the element $\As^\flat/{\approx}$ of $\Ad$ with
    domain
    \begin{align*}
      \sinti{(\As^\flat/{\approx})}\iota\defeq\{\{d\in\sinti{\As^\flat}\iota\mid\As\models
      c\approx d% \approx^{\widehat\As}(c)(d)=1
      \}\mid
    c\in\sinti{\As^\flat}\iota\},
    \end{align*}
    i.e.\ the quotient of
    $\sinti{\As^\flat}\iota$ over $\approx^{\As^\flat}$. It is
    easy to see that $\As^\flat/{\approx}\models \Set'$.
  \end{proof}

  Note that both $\Ad$ and $\Ad^\flat$ are compact. Hence, by soundness and completeness of the proof system (\cref{thm:completeness,prop:soundness}) we obtain:
  \begin{proposition}
    $\Set$ is $\Ad$-satisfiable iff $\Set$ is $\Ad^\flat$-satisfiable.
  \end{proposition}
  Consequently, by \cref{rem:dec}, we conclude:
  \begin{theorem}
    It is decidable whether there exists a $\Sigma'$-structure $\Bs$
    satisfying $\Bs\models \Set$ and ${\approx}^\Bs(a)(b)=1$ iff
    ${\not\approx}^\Bs(a)(b)=0$ iff $a=b$, for $a,b\in\sinti\Bs\iota$.
  \end{theorem}

  \subsubsection{Supplementary Materials for \cref{sec:HoBHC(SLA)}}
  \label{sec:appHoBHC(SLA)}
  
    \begin{lemma}
    \label{lem:respres}
    Let $\Set'$ be a set of HoBHC(SLA).
    If $\Set'\Respstr{\Ad} \Set'\cup\{G\}$ then $G$ is a HoBHC(SLA) and $\igt(G)\subseteq\igt(\Set')$. 
    % \lo{Is there a typo in the statement?}\dw{I don't see a typo. What do you think is wrong?} 
    % \lo{Sorry - no typo. But we should define $\Respstr{\Ad}$---with $\Ad$ in the subscript---explicitly: see definition of Compact Constraint Resolution rule.}
  \end{lemma}
  \begin{proof}
    For the compact constraint refutation rule this is trivial.

  Next, if $\neg(\lambda x\ldotp L)\,M\,\overline N\lor G$ is a HoBHC(SLA) then neither $(\lambda x\ldotp L)\,M\,\overline N$ nor $L[M/x]\,\overline N$ contain symbols from $\Sigma$ and $\igt((\lambda x\ldotp L)\,M\,\overline N)=\igt(L[M/x]\,\overline N)=\emptyset$. Hence, $\neg L[M/x]\,\overline N\lor G$ is a HoBHC(SLA) and $\igt(\neg L[M/x]\,\overline N\lor G)=\igt(\neg(\lambda x\ldotp L)\,M\,\overline N\lor G)$.

  Finally, suppose $\neg R\,\overline M\lor G$ and $G'\lor R\,\overline x$ are HoBHC(SLA)s. Note that all terms in $\overline M$ of type $\iota$ must be variables. Therefore $G\lor G'[\overline M/\overline x]$ is a HoBHC(SLA) and $\igt(G\lor G'[\overline M/\overline x])=\igt(\neg R\,\overline M\lor G)\cup\igt(G'\lor R\,\overline x)$.
\end{proof}

\hobhcslaSWidehatS*
\begin{proof}
  For the rule compact constraint refutation this is due to \cref{cor:bgiff} and for the $\beta$-reduction rule this is obvious because $((\lambda x\ldotp L)\,M\,\overline N)^\flat=(\lambda x\ldotp L)\,M\,\overline N$.

  Finally, suppose that $\{\neg R\,\overline M\lor G,G'\lor R\,\overline x\}\subseteq \Set'$. It holds that $(R\,\overline M)^\flat=R\,\overline M$, $(R\,\overline x)^\flat=R\,\overline x$ and for every atom $A$, $A^\flat[\overline M/\overline x]=(A[\overline M/\overline x])^\flat$. Consequently, $(G')^\flat[\overline M/\overline x]=(G'[\overline M/\overline x])^\flat$ and the lemma also holds for applications of the resolution rule.
\end{proof}

\begin{theorem}
  \label{thm:compphi}
  Let $\Phi$ be a predicate on atoms\footnote{which is lifted to clauses by setting $\Phi(\neg A_1\lor\cdots\lor\neg A_n\lor(\neg)A)\defeq\min\{\Phi(A_1),\ldots,\Phi(A_n),\Phi(A)\}$ and to sets of HoCHCs $\Set'$ by setting $\Phi(\Set')=\min\{\Phi(C)\mid C\in \Set'\}$} satisfying
  \begin{thmlist}
  \item $\Phi(\Set)=1$,
  \item if $\Set'\Res \Set''$ then $\Phi(\Set'')\geq\Phi(\Set')$ and
  \item if $\Set'$ is an $\Ad$-unsatisfiable set of HoCHCs satisfying $\Phi(\Set')=1$ then there exists a finite subset $\Set''\subseteq \Set'$ which is $\Ad$-unsatisfiable.
  \end{thmlist}
  Then $\Set$ is $\Ad$-unsatisfiable iff $\Set\Resp^* \Set'\cup\{\bot\}$ for some $\Set'$.
\end{theorem}
\begin{proof}
  Similar to the proof of \cref{thm:compcomp}.
\end{proof}

\complHoBHCSLA*
\begin{proof}
  Define $\Phi(A)=1$ just if $\igt(A)\subseteq\igt(\Set)$ and $\neg A$ is a HoBHC(SLA). By \cref{lem:respres,cor:bgiff}, \cref{thm:compphi} is applicable, which yields the theorem.
\end{proof}

\newpage

\ifdraft
\tableofcontents
\fi
\fi

\dw{to discuss:
  \begin{itemize}
  \item don't define pre-terms, only use typing rules to define terms.
  \item valuations over different Henkin frames (in section on
    decidable fragments)
  \end{itemize}
  to do:
  \begin{itemize}
  \item spell check
  \item check that completeness of Henkin frames is assumed in the
    right places
  \item check if superscripts are omitted consistently
  \item check if $\prel$ used for valuations somewhere
  \item use of quasi-continuous, quasi-monotone in old sense?
  \item alignment of figures
  \item term ``easy'' no used for goal clauses
  \end{itemize}
}

\lo{COAUTHORING PROTOCOL. When writing joint papers, I aim to follow a useful protocol:
\begin{enumerate}
\item I colour-highlight important changes using command {\tt changed[lo]}. 
After checking them, please indicate agreement by deleting {\tt changed[lo]}.
\item Please do not delete comments (introduced via command {\tt lo}) -- I find them a useful record.
They can be hidden by setting the \LaTeX\ flag {\tt draft} to {\tt draftfalse}.
Alternatively, they can be commented out using \%, or via a matching pair of {\tt iffalse} and {\tt fi}.
\end{enumerate}
May I propose we use it.
}

\lo{Convention: I propose we do not end an item with a comma (cleaner / leaner this way)}

\lo{Convention: let's use marco {\tt defeq} for definitional equality.}

\lo{ENGLISH USAGE: I have a slight preference for using ``moreover'' first, and then ``furthermore'', if both appear in the same paragraph.} 

\lo{STYLE. It is not considered good style to begin a sentence with non-words such as references or  mathematical expressions / symbols. 
Some journals enforce this rule.}

\lo{\LaTeX: Enclosing a binary operator within braces turns it into a symbol. 
This can restore the intended spacing around a binary operator.
See, e.g., the \LaTeX\ expressions that render ${\pred}\not\subseteq{\lredrt}$ (good) vs $\pred\not\subseteq\lredrt$ (bad).}

\end{document}